\documentclass[11pt,a4paper]{elsarticle}

\usepackage[T1]{fontenc}
\usepackage[utf8]{inputenc}
\usepackage{lmodern}
\usepackage{amsmath,amssymb,amsfonts,amsthm,mathtools}
\usepackage{mathrsfs}
\usepackage{bm}
\usepackage{enumitem}
\usepackage{booktabs}
\usepackage{array}
\usepackage{microtype}
\usepackage{tikz,float}
\usetikzlibrary{arrows.meta,positioning}
\usepackage[colorlinks=true,linkcolor=blue,citecolor=blue,urlcolor=blue]{hyperref}

\allowdisplaybreaks
\makeatletter
\def\ps@pprintTitle{%
  \let\@oddhead\@empty
  \let\@evenhead\@empty
  \let\@oddfoot\@empty
  \let\@evenfoot\@empty}
\makeatother

\newtheorem{theorem}{Theorem}[section]
\newtheorem{proposition}[theorem]{Proposition}
\newtheorem{lemma}[theorem]{Lemma}

\theoremstyle{definition}
\newtheorem{definition}[theorem]{Definition}
\newtheorem{remark}[theorem]{Remark}

\newcommand{\C}{\mathbb C}
\newcommand{\Z}{\mathbb Z}
\newcommand{\End}{\operatorname{End}}
\newcommand{\Tr}{\operatorname{Tr}}
\newcommand{\Sym}{\operatorname{Sym}}
\newcommand{\F}{\mathcal F}
\newcommand{\Hq}{\mathcal H}
\newcommand{\Id}{\mathbf 1}
\newcommand{\diag}{\operatorname{diag}}
\newcommand{\wt}{\mathbf t}

\newcommand{\ws}{\mathbf s}

\newcommand{\gl}{\mathfrak{gl}}

\newcommand{\Ibar}{\overline I}

\hypersetup{hidelinks}

\begin{document}

\begin{frontmatter}

\title{
Baxter \texorpdfstring{$Q$}{Q}-operators from a Schwinger-Boson 
construction of the master \texorpdfstring{$T$}{T}-operator and 
the mKP hierarchy
}

\author[addr1]{Zengo Tsuboi}
\address[addr1]{Osaka Central Advanced Mathematical Institute (OCAMI),
 Osaka Metropolitan University,
 3-3-138 Sugimoto, Sumiyoshi-ku Osaka 558-8585, Japan}

\begin{abstract}
The master \(T\)-operator is a generating function for transfer matrices
and a tau-function of the modified KP (mKP) hierarchy. It is
conventionally introduced through a Schur-function expansion whose
coefficients are fused transfer matrices satisfying the
Cherednik--Bazhanov--Reshetikhin determinant formula. 
For rational inhomogeneous \(\mathfrak{gl}(M)\) spin chains, we give an
alternative definition: we realize \(\mathfrak{gl}(M)\) on an auxiliary
Fock space generated by finitely many families of Schwinger bosons,
form a monodromy matrix from the resulting \(L\)-operator, and define
the master \(T\)-operator as its trace over the Fock space. Howe
duality then reproduces the Schur-function expansion.
This realization relates two constructions of Baxter \(Q\)-operators.
In the construction of \cite{AKLTZ2013} (see also \cite{KLT2012}), the
\(Q\)-operators are obtained from residues of the master \(T\)-operator
with respect to selected Miwa variables. In the construction of
\cite{BFLMS2011}, they are defined using degenerate Yangian
\(L\)-operators in oscillator form. When selected Miwa variables approach
the inverses of boundary-twist eigenvalues, the defining trace develops
poles. After a suitable rescaling of this \(L\)-operator, only the terms
contributing to the highest-order pole survive in the normalized limit;
before the trace is taken, they form the degenerate Yangian
\(L\)-operator of the second construction. The trace then gives an
explicit residue formula for the corresponding \(Q\)-operator.
Independently, a Holstein--Primakoff-type large-occupation-number
contraction of this \(L\)-operator on subspaces with fixed total
occupation numbers yields the same degenerate Yangian \(L\)-operator.
\end{abstract}

\end{frontmatter}
\noindent
\textbf{Journal reference:} Nuclear Physics B 1033 (2026) 117716

\noindent
\textbf{DOI:} \url{https://doi.org/10.1016/j.nuclphysb.2026.117716}

\tableofcontents

\section{Introduction}

Transfer matrices are among the basic objects in the study of integrable
systems. Their construction as traces of monodromy matrices specifies
both the auxiliary space and the \(L\)-operators acting along the chain.
The main result of this paper is a definition of the master
\(T\)-operator in this same form, with a Fock space of Schwinger bosons
as the auxiliary space. Decomposing that space recovers the conventional
Schur function expansion. Keeping the bosons explicit instead allows
us to examine the residues and limits of its transfer matrices directly.

A rational inhomogeneous \(\gl(M)\) spin chain is built from the
rational \(R\)-matrix in the vector representation and an
\(L\)-operator with spectral parameter \(u \in \mathbb{C}\),
\[
  L(u)=u\,\Id\otimes\Id+
  \sum_{i,j=1}^{M} e_{ij}\otimes E_{ji}.
\]
Here \(E_{ij}\) are the matrix units acting on the one-site quantum space
\(\C^M\), while \(e_{ij}\) are the generators of \(\gl(M)\) in the
representation carried by the auxiliary space.  For an inhomogeneous
chain of length \(L\), one inserts \(L(u-\theta_\ell)\) at the
\(\ell\)-th site, where \(\theta_1,\ldots,\theta_L\in\C\) are the
inhomogeneity parameters.  If \(V_\lambda\) is the irreducible
polynomial \(GL(M)\)-module labelled by a Young diagram \(\lambda\), the
corresponding fused transfer matrix is
\[
  T_\lambda(u)
  =\Tr_{V_\lambda}
  \left[
    L_{0L}^{\lambda}(u-\theta_L)\cdots
    L_{01}^{\lambda}(u-\theta_1)\,\pi_\lambda(g)
  \right].
\]
Here the label \(0\) denotes the auxiliary space \(V_\lambda\), while
\(\ell\) denotes the \(\ell\)-th quantum space.  The operator
\(L_{0\ell}^{\lambda}(u)\) is the above \(L\)-operator with the
\(\gl(M)\) generators acting on \(V_\lambda\), and
\(\pi_\lambda(g)\) is the action of the boundary twist \(g\in GL(M)\) on
\(V_\lambda\).

In the formula for \(T_\lambda(u)\), the factor \(\pi_\lambda(g)\)
denotes \(\pi_\lambda(g)\otimes\Id_{(\C^M)^{\otimes L}}\).
Throughout the paper, when the tensor factors on which an operator
acts have been specified, we use the same symbol for its extension
by the identity on the remaining tensor factors.

These transfer
matrices can be assembled into the master \(T\)-operator
\cite{AKLTZ2013} (see \cite{KLT2012} for a preliminary form;
\cite{Zabrodin2013Trig,Alexandrov2014,Zabrodin2014XXX,TsuboiZabrodinZotov2015,
Rozhkovskaya2018,LiShou2020,LiShou2023,Zabrodin2025}
for further developments).  Conventionally, the master \(T\)-operator is
introduced by the Schur function expansion
\begin{align}
  \mathcal T(u,\mathbf t)
  =\sum_{\lambda}
   s_\lambda(\mathbf t)T_\lambda(u),
  \label{eq:intro-master-Schur}
\end{align}
where \(\mathbf t=(t_1,t_2,\ldots)\) are the KP time variables,
\(\lambda\) runs over Young diagrams, \(s_\lambda(\mathbf t)\)
is the Schur function associated with \(\lambda\),
and \(T_\lambda(u)\) is the corresponding fused transfer matrix.
The Cherednik--Bazhanov--Reshetikhin (CBR) determinant formula
\cite{Cherednik1989,BazhanovReshetikhin1990}
(see \cite{KazakovVieira2008} for a proof based on the co-derivative) is
the quantum Jacobi--Trudi form of the Pl\"ucker relations satisfied by
the coefficients \(T_\lambda(u)\) in \eqref{eq:intro-master-Schur}; hence
\(\mathcal T(u,\mathbf t)\) is a tau-function of the modified KP
hierarchy \cite{AKLTZ2013}.

The Schur function expansion organizes transfer matrices with different
auxiliary representations, but does not by itself specify a single
\(L\)-operator and auxiliary space for the master \(T\)-operator.
We provide these explicitly.
Let \(B=\{1,\ldots,M\}\), and let \(R\) be a finite set whose elements
label the families of Schwinger bosons.  
We call the elements of \(B\) colour labels and those of \(R\) flavour labels. We write the Schwinger bosons as
\(\mathsf a_{i\alpha}\), \(\mathsf a_{i\alpha}^{\dagger}\), with
\(i\in B\) and \(\alpha\in R\), satisfying

\begin{equation}
 \begin{aligned}
 [\mathsf a_{i\alpha},\mathsf a_{j\beta}^{\dagger}]
   &=\delta_{ij}\delta_{\alpha\beta},\\
 [\mathsf a_{i\alpha},\mathsf a_{j\beta}]
   &=[\mathsf a_{i\alpha}^{\dagger},\mathsf a_{j\beta}^{\dagger}]=0.
 \end{aligned}
 \label{eq:intro-boson-commutators}
\end{equation}

The left index is a colour label, and the right index is a flavour label.
Let \(\F_{B,R}\) be the Fock space generated by these oscillators.
They give the Schwinger boson realization

\begin{equation}
  e_{ij}^{(R)}=
  \sum_{\alpha\in R}
  \mathsf a_{i\alpha}^{\dagger}\mathsf a_{j\alpha}
  \label{eq:intro-Schwinger-generators}
\end{equation}
of the \(\gl(M)\) generators.  Substituting
\(e_{ij}=e_{ij}^{(R)}\) in the above \(L\)-operator gives 
\begin{equation}
  L^{(R)}(u)
  =u\,\Id \otimes \Id +
  \sum_{i,j\in B}e_{ij}^{(R)}\otimes E_{ji}
  \in \End(\F_{B,R}\otimes\C^M),
  \label{eq:intro-Schwinger-L-explicit}
\end{equation}
which we shall call the Schwinger \(L\)-operator.
For a diagonal twist \(g=\diag(g_1,\ldots,g_M)\) and complex variables
\(\mathbf y_R=(y_\alpha)_{\alpha\in R}\), we define the master
\(T\)-operator as the trace over \(\F_{B,R}\):

\begin{align}
  \mathcal T(u;\mathbf y_R)
  &=\Tr_{\F_{B,R}}
  \left[
  L_{0L}^{(R)}(u-\theta_L)\cdots L_{01}^{(R)}(u-\theta_1)
  \,\widehat g\,\widehat h(\mathbf y_R)
  \right],
  \label{eq:intro-master-trace}
  \\
  \widehat g\,\widehat h(\mathbf y_R)
  &=\prod_{i\in B}\prod_{\alpha\in R}
  (g_i y_\alpha)^{\mathsf N_{i\alpha}}.
  \label{eq:intro-diagonal-factor}
\end{align}
Here \(L_{0\ell}^{(R)}(u-\theta_\ell)\) is the operator
\eqref{eq:intro-Schwinger-L-explicit} acting on the auxiliary space
\(\F_{B,R}\) and the \(\ell\)-th quantum space.  The number operator is
\(\mathsf N_{i\alpha}
 =\mathsf a_{i\alpha}^{\dagger}\mathsf a_{i\alpha}\).
The precise definitions are given in
\eqref{eq:Schwinger-L}--\eqref{eq:finite-master-def}.
The variables \(y_\alpha\) are called Miwa variables and are related
to the KP time variables by the Miwa transformation \cite{Miwa1982}:
\begin{equation}
  t_k(\mathbf y_R)=\frac1k\sum_{\alpha\in R}y_\alpha^k,
  \qquad k\ge1.
  \label{eq:intro-Miwa-times}
\end{equation}

To see why this trace gives the same master \(T\)-operator as in
\eqref{eq:intro-master-Schur}, put \(r=|R|\). The Fock space carries commuting actions of \(GL(M)\) on the
colour labels and of \(GL(r)\) on the flavour labels. Howe duality gives
\cite{Howe1989,GoodmanWallach1998}

\begin{equation}
 \F_{B,R}=\Sym(\C^M\otimes\C^r)
 \simeq\bigoplus_{\ell(\lambda)\le\min(M,r)}
 V_\lambda^{GL(M)}\otimes V_\lambda^{GL(r)}.
 \label{eq:intro-Howe-decomposition}
\end{equation}

Here each factor is the irreducible polynomial representation labelled
by \(\lambda\), and \(\ell(\lambda)\) is the number of rows of the Young
diagram. Each pair occurs once. The \(L\)-operators act on the first representation and on the physical
spins, and as the identity on the second representation. The twists
\(\widehat g\) and \(\widehat h\) act on the first and second
representations, respectively. For \(h=\diag(y_\alpha)_{\alpha\in R}\), the contribution of
one summand to the trace over \(\F_{B,R}\) is therefore

\begin{align}
 &\underbrace{\Tr_{V_\lambda^{GL(M)}}\!\left[
 L^\lambda_{0L}(u-\theta_L)\cdots L^\lambda_{01}(u-\theta_1)
 \pi_\lambda(g)\right]}_{T_\lambda(u)}
 \notag\\[-1mm]
 &\hspace{18mm}\times
 \underbrace{\Tr_{V_\lambda^{GL(r)}}\pi_\lambda(h)}_{s_\lambda(\mathbf y_R)}.
 \label{eq:intro-Howe-traces}
\end{align}

The trace over \(V_\lambda^{GL(M)}\) gives the transfer matrix
\(T_\lambda(u)\), which acts on the quantum space of the spin chain.
The trace over \(V_\lambda^{GL(r)}\) gives the scalar
\(s_\lambda(\mathbf y_R)\). Summing the products of these two traces over
\(\lambda\) gives \eqref{eq:intro-master-Schur} with the KP time variables
set to the values in \eqref{eq:intro-Miwa-times}. Hence
\(\mathcal T(u;\mathbf y_R)=\mathcal T(u,\mathbf t(\mathbf y_R))\).
The Schur function expansion therefore follows from the definition
\eqref{eq:intro-master-trace}. The full calculation is given in
Theorem~\ref{thm:Schur-expansion-from-Howe}.

Baxter \(Q\)-operators originate in Baxter's solution of the eight-vertex
model \cite{Baxter1972}.  Two constructions of Baxter \(Q\)-operators
for rational \(\gl(M)\) spin chains are relevant here.  In the first,
the \(Q\)-operators are obtained from the master \(T\)-operator by taking
residues with respect to selected Miwa variables
\cite{AKLTZ2013,KLT2012}.  In the second, they are defined as traces of
products of degenerate Yangian \(L\)-operators written in oscillator form
\cite{BFLMS2011}.  The use of oscillator representations in
constructions of \(Q\)-operators has an important antecedent in the study
of integrability of the conformal field theory \cite{BLZ1999}.  A
different approach based on the factorization of \(R\)-operators for
generic \(\mathfrak{sl}(N)\) spin chains is given in
Ref.~\cite{DerkachovManashov2009}; its relation to the construction
studied here is not considered further.

The products \(g_i y_\alpha\) in \eqref{eq:intro-diagonal-factor} explain
the points at which the residues are taken.  For a single oscillator,
\begin{equation}
  \Tr\,(g_i y_\alpha)^{\mathsf N_{i\alpha}}
  =\sum_{n\ge0}(g_i y_\alpha)^n
  =\frac{1}{1-g_i y_\alpha}.
  \label{eq:intro-one-oscillator-pole}
\end{equation}
Choose a subset \(I\subset B\) and write \(\overline I=B\setminus I\).
For each \(\beta\in\overline I\), choose one flavour label in \(R\) and
denote it by the same symbol \(\beta\).  The first construction studies
the variables attached to these labels in the simultaneous limits
\begin{equation}
  y_\beta\longrightarrow g_\beta^{-1},
  \qquad \beta\in\overline I,
  \label{eq:intro-residue-points}
\end{equation}

while the other variables are kept away from the poles in
\eqref{eq:intro-one-oscillator-pole}.
We remark that increasing the chain length \(L\) can raise the orders
of these poles but does not change their locations.

The definition \eqref{eq:intro-master-trace} allows the residue to be
calculated using the Schwinger bosons. We multiply the Schwinger
\(L\)-operator on both sides by the diagonal operator specified in
\eqref{eq:scaled-local-L} and determine which terms in its matrix elements
contribute to the highest-order poles of \(\mathcal T(u;\mathbf y_{\Ibar})\).
In those terms, the factors involving \(\mathsf a_{\beta\beta}\) and
\(\mathsf a_{\beta\beta}^{\dagger}\), \(\beta\in\Ibar\), can be
separated from the matrix elements of the degenerate Yangian
\(L\)-operator \(\mathscr L_{I,\Ibar}(u)\) used in
\cite{BFLMS2011}. The matrix elements of \(\mathscr L_{I,\Ibar}(u)\)
in our conventions are given in
\eqref{eq:contracted-L-ac}--\eqref{eq:contracted-L-components}.
Taking the traces over the oscillators other than
\(\mathsf a_{a\beta}\), \(a\in I\), \(\beta\in\Ibar\), after forming
the product of the \(L\)-operators gives the residue formula in
Proposition~\ref{prop:main-contraction}. The \(Q\)-operator in that
formula is defined by the trace over \(\F_{I,\Ibar}\):

\begin{align}
  Q_I(u)
  &=\Tr_{\F_{I,\overline I}}
  \left[
  \mathscr L_{I,\overline I,L}(u-\theta_L)\cdots
  \mathscr L_{I,\overline I,1}(u-\theta_1)
  D_{I,\overline I}
  \right],
  \label{eq:intro-Q-trace}
  \\
  D_{I,\overline I}
  &=\prod_{a\in I}\prod_{\beta\in\overline I}
  \left(\frac{g_a}{g_\beta}\right)^{\mathsf N_{a\beta}}.
  \label{eq:intro-Q-diagonal-factor}
\end{align}
The full definitions are \eqref{eq:DAB-derived} and
\eqref{eq:Sch-Q-def}.

Second, restrict the Schwinger realization to subspaces on which the
total occupation number associated with each \(\beta\in\overline I\) is
fixed.  One occupation number is then determined by this constraint.
After the rescaling in \eqref{eq:S-m-def}, sending these fixed occupation
numbers to infinity gives a Holstein--Primakoff-type
large-occupation-number contraction of the Schwinger \(L\)-operator
\(L^{(R)}(u)\), yielding the same \(L\)-operator
\(\mathscr L_{I,\overline I}(u)\) without using the residue calculation
for the trace.

The relation between the two constructions of \(Q\)-operators is thus
an application of the new definition of the master \(T\)-operator.
This definition makes its auxiliary representation and \(L\)-operators
available for calculations, as for an ordinary transfer matrix. We
expect this to facilitate its use in other problems in integrable
systems. The mKP bilinear identity, obtained from the CBR determinant
formula, then yields the \(QQ\)-relation among the resulting
\(Q\)-operators. We state that relation in our conventions, following
Refs.~\cite{AKLTZ2013,KLT2012}.

The paper is organized as follows. Section~\ref{sec:conventions} fixes
the notation. Section~\ref{sec:master-schwinger} defines the master
\(T\)-operator and derives its Schur function expansion.
Section~\ref{sec:bilinear-context} recalls the mKP bilinear identity.
Sections~\ref{sec:local-contraction} and
\ref{sec:Q-contraction-theorem} specify the operators and normalization
used in the residue formula and state Proposition~\ref{prop:main-contraction}.
The calculations that determine the contributing terms, the proofs of
the conservation relations, and the full proof of the proposition are
collected in \ref{app:residue-details}.
Section~\ref{sec:HP-fixed-degree} derives the same degenerate Yangian
\(L\)-operator by fixing the total occupation numbers in the Fock space
and then letting those numbers tend to infinity. The final section states
the \(QQ\)-relation in our conventions. \ref{app:direct-trace-asymptotics}
evaluates traces of products of creation and annihilation operators with
weights depending on their occupation numbers. The other appendices
check the residue formula in special cases, compare conventions for the
degenerate Yangian \(L\)-operator, and treat additional families of
Schwinger bosons. Subsection~\ref{subsec:main-ideas-results} illustrates
the construction in Section~\ref{sec:master-schwinger}, the proof of
the residue formula in \ref{app:residue-details}, and the limit in
Section~\ref{sec:HP-fixed-degree} through a simple example. The subsequent
arguments are independent of that example, so
Subsection~\ref{subsec:main-ideas-results} may be skipped.

\par\medskip

\subsection{Definition of the master \texorpdfstring{$T$}{T}-operator by a trace and two methods for obtaining a \texorpdfstring{$Q$}{Q}-operator}
\label{subsec:main-ideas-results}

In \eqref{eq:intro-master-trace}, the master \(T\)-operator is
constructed by multiplying the \(L\)-operators along the chain and then
taking the trace over the Fock space on which all these \(L\)-operators
act. We take \(M=2\), \(L=2\), and
\(R=\{1\}\) as an example. There are two sites and two oscillators,
\(\mathsf a_{11}\) and \(\mathsf a_{21}\); both oscillators are used at
both sites. We write
\[
 y=y_1,\qquad u_\ell=u-\theta_\ell\quad(\ell=1,2),\qquad
 \rho=\frac{g_2}{g_1},\qquad g_1g_2\ne0,\quad g_1\ne g_2.
\]
At each site the quantum space is \(\C^2\), with basis
\[
 e_1=\begin{pmatrix}1\\0\end{pmatrix},\qquad
 e_2=\begin{pmatrix}0\\1\end{pmatrix}.
\]
The matrix units are defined by \(E_{ij}e_c=\delta_{jc}e_i\).
For the two-site quantum space \(\C^2\otimes\C^2\), we use the ordered basis
\[
 (|1,1\rangle,\ |1,2\rangle,\ |2,1\rangle,\ |2,2\rangle),
 \qquad |c_1,c_2\rangle=e_{c_1}\otimes e_{c_2}.
\]
The transfer matrices \(\mathcal T(u;y)\), \(T_{(N)}(u)\), and
\(Q_{\{2\}}(u)\) are \(4\times4\) matrices in this basis.
An \(L\)-operator is a \(2\times2\) matrix whose 
matrix elements act on the Fock space
over which the trace is taken.

The relations to be demonstrated are
\begin{equation}
 \mathcal T(u;y)=\sum_{N=0}^{\infty}y^N T_{(N)}(u)
 \label{eq:ex04-goal-series}
\end{equation}
and
\begin{align}
 &-g_1\,\Gamma(\mathsf M_1+1)^{-1}
   \operatorname*{Res}_{y=g_1^{-1}}
   \bigl[(1-g_1y)^{\mathsf M_1}\mathcal T(u;y)\,dy\bigr]
 \notag\\*
 &\qquad=\lim_{N\to\infty}g_1^{-N}N^{-\mathsf M_1}T_{(N)}(u)
 =Q_{\{2\}}(u).
 \label{eq:ex04-goal-Q}
\end{align}
Here \(N\) is the total occupation number of the two oscillators.
The operator \(\mathsf M_1\) counts the sites in state \(e_1\):
\[
 \mathsf M_1=E_{11}\otimes\Id+\Id\otimes E_{11}=\diag(2,1,1,0).
\]
Thus the matrices multiplying the transfer matrices in
\eqref{eq:ex04-goal-Q} are explicitly
\begin{align*}
 (1-g_1y)^{\mathsf M_1}
   &=\diag\bigl((1-g_1y)^2,1-g_1y,1-g_1y,1\bigr),\\*
 \Gamma(\mathsf M_1+1)^{-1}&=\diag(1/2,1,1,1),\\*
 N^{-\mathsf M_1}&=\diag(N^{-2},N^{-1},N^{-1},1).
\end{align*}
In the residue in \eqref{eq:ex04-goal-Q}, the matrix
\((1-g_1y)^{\mathsf M_1}\) multiplies \(\mathcal T(u;y)\) from the
left before the residue is taken. Its four diagonal elements multiply
the first, second, third and fourth rows of \(\mathcal T(u;y)\) by
\((1-g_1y)^2\), \(1-g_1y\), \(1-g_1y\) and \(1\), respectively.
These powers of \(1-g_1y\) reduce the possible orders \(3,2,2,1\)
of the poles at \(y=g_1^{-1}\) to at most one.
After taking the residue in \eqref{eq:ex04-goal-Q}, we multiply the
resulting \(4\times4\) matrix from the left by
\(\Gamma(\mathsf M_1+1)^{-1}\). This divides the first row by \(2\)
and leaves the other rows unchanged.
For the limit in \eqref{eq:ex04-goal-Q}, multiplication of
\(T_{(N)}(u)\) from the left by \(N^{-\mathsf M_1}\) divides the
first row of \(T_{(N)}(u)\) by \(N^2\) and the second and third
rows by \(N\), while leaving the fourth row unchanged.

The reference to rows follows from placing the diagonal matrices on
the left in \eqref{eq:ex04-goal-Q}. Both transfer matrices preserve
the number of sites in state \(e_1\), and therefore
\[
 \begin{aligned}
 (1-g_1y)^{\mathsf M_1}\mathcal T(u;y)
   &=\mathcal T(u;y)(1-g_1y)^{\mathsf M_1},\\
 N^{-\mathsf M_1}T_{(N)}(u)
   &=T_{(N)}(u)N^{-\mathsf M_1}.
 \end{aligned}
\]
Multiplication from the right would thus apply the corresponding
factors to the columns and give the same operators.
For the limit of \(T_{(N)}\), we take \(N\to\infty\) in the
coefficient of each fixed power \(\rho^n\), as the calculation
in (F) will show.

Equation~\eqref{eq:ex04-goal-series} specializes
\eqref{eq:intro-master-Schur} and
Theorem~\ref{thm:Schur-expansion-from-Howe} to
\(M=L=2\), \(R=\{1\}\), and \(t_j=y^j/j\).
For \eqref{eq:ex04-goal-Q}, take \(I=\{2\}\) and
\(\overline I=R=\{1\}\). The residue is the corresponding case of
Proposition~\ref{prop:main-contraction}, Eq.~\eqref{eq:main-residue-formula}.
The limit is the corresponding case of
Eqs.~\eqref{eq:fixed-coefficient-limit} and
\eqref{eq:scaled-fixed-trace-from-master}, with \(m_1=N\).
Here we verify these relations by direct calculation, rather than by
the general arguments given in the later sections. We work out the
matrix element \(\langle1,2|\,(\cdots ) \,|1,2\rangle\) at each step.
The other matrix elements are calculated in the same way and are
included in the matrices below.
Figure~\ref{fig:ex04-plan} indicates which relation is checked in each
paragraph.

\begin{figure}[H]
\centering
\begin{tikzpicture}[>=Stealth,x=1cm,y=1cm,font=\small,
 every node/.style={align=center,inner sep=4pt}]
 
 \node (L) at (0,0) {\(L^{(\{1\})}(u)\)};
 \node (T) at (4.2,0) {\(\mathcal T(u;y)\)};
 \node (S) at (8.6,0) {\(\displaystyle\sum_{N\ge0}y^N T_{(N)}(u)\)};
 \node (Lc) at (0,-3.3) {\(\mathscr L_{\{2\},\{1\}}(u)\)};
 \node (Q) at (4.2,-3.3) {\(Q_{\{2\}}(u)\)};
 \node (Lim) at (8.6,-3.3)
  {\(\displaystyle\lim_{N\to\infty}g_1^{-N}N^{-\mathsf M_1}T_{(N)}(u)\)};
 \draw[->] (L) -- node[above]{(A)}
  node[below,font=\footnotesize]{product and trace} (T);
 \draw[->] (T) -- node[above]{(B)}
  node[below,font=\footnotesize]{Howe duality} (S);
 \draw[->] (L) -- node[left]{(C)}
  node[right,font=\footnotesize,align=left]
   {fix \(N\), then \(N\to\infty\);\\or keep the terms\\that contribute to the\\residue of \(\mathcal T\)} (Lc);
 \draw[->] (Lc) -- node[above]{(D)}
  node[below,font=\footnotesize]{product and trace} (Q);
 \draw[->] (T) -- node[left]{(E)}
  node[right,font=\footnotesize,align=left]
   {residue with the factors\\in \eqref{eq:ex04-goal-Q}} (Q);
 \draw[->] (S) -- node[right]{(F)} (Lim);
 \draw[->] (Lim) -- node[above,font=\footnotesize]{equality (F)} (Q);
\end{tikzpicture}
\caption{The calculations in (A)--(F). For (A) and (D), the two
 \(L\)-operators are multiplied in the order \(2,1\), followed by the
 boundary operator, and then traced over the common Fock space.
 In (B), the trace is grouped by the total occupation number \(N\).
 The two calculations in (C) yield the same
 \(\mathscr L_{\{2\},\{1\}}\): one fixes \(N\) and takes
 \(N\to\infty\), while the other separates the terms contributing to
 the residue of \(\mathcal T\). In the latter calculation, the factors
 containing \(\mathsf a_{11}\) remain in the product until their trace
 is evaluated. The residue in (E) and the limit in (F) use the factors
 in \eqref{eq:ex04-goal-Q}.}
\label{fig:ex04-plan}
\end{figure}

\paragraph{(A) Taking the trace of the two \(L\)-operators}
For \(M=L=2\) and \(R=\{1\}\), evaluating the trace in
\eqref{eq:intro-master-trace} gives
\begingroup
\setlength{\arraycolsep}{3pt}
\begin{align}
 &\mathcal T(u;y)=\frac{1}{(1-g_1y)(1-g_2y)}
 \notag\\*
 &\times
 \begin{pmatrix}
 \begin{gathered}u_1u_2\\
 {}+\dfrac{(u_1+u_2)g_1y}{1-g_1y}\\
 {}+\dfrac{g_1y(1+g_1y)}{(1-g_1y)^2}\end{gathered}
 &0&0&0\\[4mm]
 0&\begin{gathered}\left(u_1+\tfrac{g_1y}{1-g_1y}\right)\\
 {}\times\left(u_2+\tfrac{g_2y}{1-g_2y}\right)\end{gathered}
 &\begin{gathered}\dfrac{g_1y}{1-g_1y}\\[1mm]{}\times\dfrac1{1-g_2y}\end{gathered}&0\\[4mm]
 0&\begin{gathered}\dfrac1{1-g_1y}\\[1mm]{}\times\dfrac{g_2y}{1-g_2y}\end{gathered}
 &\begin{gathered}\left(u_2+\tfrac{g_1y}{1-g_1y}\right)\\
 {}\times\left(u_1+\tfrac{g_2y}{1-g_2y}\right)\end{gathered}&0\\[4mm]
 0&0&0&\begin{gathered}u_1u_2\\
 {}+\dfrac{(u_1+u_2)g_2y}{1-g_2y}\\
 {}+\dfrac{g_2y(1+g_2y)}{(1-g_2y)^2}\end{gathered}
 \end{pmatrix}.
 \label{eq:ex04-master-matrix}
\end{align}
\endgroup
As a function of \(y\), \eqref{eq:ex04-master-matrix} has poles only at
\(y=g_1^{-1}\) and \(y=g_2^{-1}\). These poles already occur in the trace
\eqref{eq:intro-one-oscillator-pole}, which is independent of the
spectral parameter \(u\). Inserting the two \(L\)-operators raises
the orders of the poles but does not change their locations.
As an example, we derive the matrix element on \(|1,2\rangle\) directly from the
trace.

The normalized basis of \(\F_{B,\{1\}}\) is
\begin{equation}
 |k,n\rangle_{\F}
 =\frac{(\mathsf a_{11}^{\dagger})^k(\mathsf a_{21}^{\dagger})^n}
        {\sqrt{k!\,n!}}|0,0\rangle_{\F},\qquad k,n\ge0.
 \label{eq:ex04-basis}
\end{equation}
The vacuum $|0,0\rangle_{\F}$ is annihilated by \(\mathsf a_{11}\) and
\(\mathsf a_{21}\). 
Thus
\begin{align*}
 \mathsf a_{11}|k,n\rangle_{\F}&=\sqrt{k}|k-1,n\rangle_{\F},&
 \mathsf a_{11}^{\dagger}|k,n\rangle_{\F}&=\sqrt{k+1}|k+1,n\rangle_{\F},\\*
 \mathsf a_{21}|k,n\rangle_{\F}&=\sqrt{n}|k,n-1\rangle_{\F},&
 \mathsf a_{21}^{\dagger}|k,n\rangle_{\F}&=\sqrt{n+1}|k,n+1\rangle_{\F}.
\end{align*}
An annihilation operator gives zero when its occupation number is zero.
The number operators \(\mathsf N_{11}=\mathsf a_{11}^{\dagger}\mathsf a_{11}\)
and \(\mathsf N_{21}=\mathsf a_{21}^{\dagger}\mathsf a_{21}\) have
eigenvalues \(k\) and \(n\), respectively. The subscript \(\F\)
distinguishes the vectors \(|k,n\rangle_{\F}\) of the Fock space from
the vectors \(|c_1,c_2\rangle\) of the quantum space.

Substituting \(e_{ij}^{(\{1\})}=\mathsf a_{i1}^{\dagger}\mathsf a_{j1}\)
into \eqref{eq:intro-Schwinger-L-explicit} gives
\begin{equation}
 L^{(\{1\})}(u)=
 u\,\Id+\sum_{i,j=1}^2\mathsf a_{i1}^{\dagger}\mathsf a_{j1}\otimes E_{ji}
 =\begin{pmatrix}
 u+\mathsf N_{11}&\mathsf a_{21}^{\dagger}\mathsf a_{11}\\
 \mathsf a_{11}^{\dagger}\mathsf a_{21}&u+\mathsf N_{21}
 \end{pmatrix}.
 \label{eq:ex04-L}
\end{equation}
In \eqref{eq:ex04-L}, the first row and column correspond to \(e_1\),
and the second row and column to \(e_2\). Substituting this
\(L\)-operator into \eqref{eq:intro-master-trace} gives
\begin{equation}
 \mathcal T(u;y)=\Tr_{\F_{B,\{1\}}}\!\bigl[
 L^{(\{1\})}_{02}(u_2)L^{(\{1\})}_{01}(u_1)
 (g_1y)^{\mathsf N_{11}}(g_2y)^{\mathsf N_{21}}\bigr].
 \label{eq:ex04-master-trace}
\end{equation}
The label \(0\) denotes the Fock space, where the trace is taken. 

For \(c,d\in\{1,2\}\), \([L^{(\{1\})}(u)]_{cd}\) denotes the
\((c,d)\)-matrix element of the \(L\)-operator with respect to the basis
\((e_1,e_2)\). 
It is the operator multiplying \(E_{cd}\) in the expansion
\[
 \begin{aligned}
 L^{(\{1\})}(u)
 &=\sum_{c,d=1}^{2}[L^{(\{1\})}(u)]_{cd}\otimes E_{cd},\\
 [L^{(\{1\})}(u)]_{cd}
 &=u\delta_{cd}\Id+\mathsf a_{d1}^{\dagger}\mathsf a_{c1}.
 \end{aligned}
\]

A matrix element on the two-site quantum space is
\begin{equation}
 \langle c_1,c_2|L_{02}^{(\{1\})}(u_2)L_{01}^{(\{1\})}(u_1)
 |d_1,d_2\rangle
 =[L^{(\{1\})}(u_2)]_{c_2d_2}[L^{(\{1\})}(u_1)]_{c_1d_1}.
 \label{eq:ex04-selection}
\end{equation}
For \((c_1,c_2)=(d_1,d_2)=(1,2)\), the selection rule
\eqref{eq:ex04-selection} gives
\begin{align*}
 &\langle1,2|L_{02}^{(\{1\})}(u_2)L_{01}^{(\{1\})}(u_1)|1,2\rangle
       |k,n\rangle_{\F}\\*
 &\quad=[L^{(\{1\})}(u_2)]_{22}[L^{(\{1\})}(u_1)]_{11}
       |k,n\rangle_{\F}\\*
 &\quad=(u_2+\mathsf N_{21})(u_1+\mathsf N_{11})|k,n\rangle_{\F}
       =(u_1+k)(u_2+n)|k,n\rangle_{\F}.
\end{align*}
Substitution into the trace \eqref{eq:ex04-master-trace} then yields
\begin{align}
 &\langle1,2|\mathcal T(u;y)|1,2\rangle\notag\\*
 &=\Tr_{\F_{B,\{1\}}}\!\Bigl[
   \langle1,2|L_{02}^{(\{1\})}(u_2)L_{01}^{(\{1\})}(u_1)|1,2\rangle
   (g_1y)^{\mathsf N_{11}}(g_2y)^{\mathsf N_{21}}\Bigr]\notag\\*
 &=\Tr_{\F_{B,\{1\}}}\!\Bigl[
   (u_2+\mathsf N_{21})(u_1+\mathsf N_{11})
   (g_1y)^{\mathsf N_{11}}(g_2y)^{\mathsf N_{21}}\Bigr]\notag\\*
 &=\sum_{k,n\ge0}(u_1+k)(u_2+n)(g_1y)^k(g_2y)^n\notag\\*
 &=\left[\frac{u_1}{1-g_1y}+\frac{g_1y}{(1-g_1y)^2}\right]
   \left[\frac{u_2}{1-g_2y}+\frac{g_2y}{(1-g_2y)^2}\right].
 \label{eq:ex04-worked-master}
\end{align}
We used the geometric series and its derivatives by \(x\,d/dx\):
\begin{equation}
 \sum_{m\ge0}x^m=\frac1{1-x},\qquad
 \sum_{m\ge0}mx^m=\frac{x}{(1-x)^2},\qquad
 \sum_{m\ge0}m^2x^m=\frac{x(1+x)}{(1-x)^3}.
 \label{eq:ex04-sums}
\end{equation}
The other matrix elements of \eqref{eq:ex04-master-matrix} 
can be calculated similarly.

\paragraph{(B) Grouping the trace by the total occupation number}
We obtain the coefficient of \(y^N\) in \eqref{eq:ex04-goal-series}:
\begingroup
\setlength{\arraycolsep}{3.5pt}
\begin{align}
 &T_{(N)}(u)=g_1^N\sum_{n=0}^{N}\rho^n
 \notag\\*
 &\quad\times\begin{pmatrix}
 \begin{gathered}(u_1+N-n)\\{}\times(u_2+N-n)\end{gathered}&0&0&0\\[3mm]
 0&\begin{gathered}(u_1+N-n)\\{}\times(u_2+n)\end{gathered}
  &(N-n)(n+1)&0\\[3mm]
 0&(N-n+1)n&\begin{gathered}(u_2+N-n)\\{}\times(u_1+n)\end{gathered}&0\\[3mm]
 0&0&0&(u_1+n)(u_2+n)
 \end{pmatrix}.
 \label{eq:ex04-TN-sum}
\end{align}
\endgroup
Equation~\eqref{eq:ex04-TN-sum} is the result of evaluating the trace
that defines \(T_\lambda(u)\), with \(\lambda=(N)\),
\(V_\lambda=\Sym^N\C^2\), and \(L=2\).
To perform that trace, fix \(k+n=N\). The vectors
\(|N-n,n\rangle_{\F}\), \(0\le n\le N\), form a basis of
\(\Sym^N\C^2\). The generators
\(\mathsf a_{i1}^{\dagger}\mathsf a_{j1}\) preserve the total
occupation number, so the two \(L\)-operators map this space into itself.
The boundary twist on each basis vector is
\(g_1^{N-n}g_2^n=g_1^N\rho^n\). Setting \(k=N-n\) in
\eqref{eq:ex04-worked-master} therefore gives
\begin{equation}
 \langle1,2|T_{(N)}(u)|1,2\rangle
 =g_1^N\sum_{n=0}^{N}(u_1+N-n)(u_2+n)\rho^n.
 \label{eq:ex04-worked-TN}
\end{equation}
The other matrix elements follow in the same way.

The operator at the boundary of \eqref{eq:ex04-master-trace} differs
from the twist used for \(T_{(N)}\) by the factor \(y^N\).
Summing the traces over all total occupation numbers gives
\(\mathcal T(u;y)=\sum_{N\ge0}y^NT_{(N)}(u)\), as claimed in
\eqref{eq:ex04-goal-series}.
For \(M=2\) and \(|R|=1\), the Howe decomposition
\eqref{eq:intro-Howe-decomposition} reads
\[
 \F_{B,\{1\}}\simeq\bigoplus_{N\ge0}\Sym^N\C^2.
\]
The representation of \(GL(1)\) paired with \(\Sym^N\C^2\) is
one-dimensional, and its character is \(y^N=s_{(N)}(y)\).
Thus the sum over total occupation numbers is the Schur function
expansion \(\sum_N s_{(N)}(y)T_{(N)}(u)\), evaluated at
\(t_j=y^j/j\).

The finite sum in \eqref{eq:ex04-TN-sum} can also be evaluated explicitly:
\begingroup
\setlength{\arraycolsep}{1.25pt}
\begin{align}
 &g_1^{-N}T_{(N)}(u)
 =\frac{1-\rho^{N+1}}{1-\rho}
 \begin{pmatrix}
 (u_1+N)(u_2+N)&0&0&0\\
 0&(u_1+N)u_2&N&0\\
 0&0&(u_2+N)u_1&0\\
 0&0&0&u_1u_2
 \end{pmatrix}\notag\\*[2mm]
 &\quad+\frac{\rho-(N+1)\rho^{N+1}+N\rho^{N+2}}{(1-\rho)^2}
 \notag\\*
 &\qquad\times\begin{pmatrix}
 -(u_1+u_2+2N)&0&0&0\\
 0&u_1+N-u_2&N-1&0\\
 0&N+1&u_2+N-u_1&0\\
 0&0&0&u_1+u_2
 \end{pmatrix}\notag\\*[2mm]
 &\quad+\frac{\begin{gathered}
 \rho(1+\rho)-(N+1)^2\rho^{N+1}\\[-1mm]
 {}+(2N^2+2N-1)\rho^{N+2}-N^2\rho^{N+3}
 \end{gathered}}{(1-\rho)^3}
 \begin{pmatrix}
 1&0&0&0\\0&-1&-1&0\\0&-1&-1&0\\0&0&0&1
 \end{pmatrix}.
 \label{eq:ex04-TN-closed}
\end{align}
\endgroup
The three scalar factors in \eqref{eq:ex04-TN-closed} are
\(\sum_{n=0}^N\rho^n\), \(\sum_{n=0}^N n\rho^n\), and
\(\sum_{n=0}^N n^2\rho^n\). They are obtained by applying
\(1\), \(\rho\,d/d\rho\), and \((\rho\,d/d\rho)^2\) to
\((1-\rho^{N+1})/(1-\rho)\), respectively. For the limit in (F),
we will use the finite sum \eqref{eq:ex04-TN-sum}, where the dependence
on \(N\) is easier to follow.

\paragraph{(C) Obtaining the \(L\)-operator at fixed total occupation number}
For the first calculation in (C), we compare the subspaces with fixed
total occupation number by the maps
\begin{equation}
 \begin{aligned}
 \Upsilon_N:\ \operatorname{span}\{|n\rangle:0\le n\le N\}
 &\longrightarrow
 \operatorname{span}\{|N-n,n\rangle_{\F}:0\le n\le N\},\\
 \Upsilon_N|n\rangle&=|N-n,n\rangle_{\F}.
 \end{aligned}
 \label{eq:ex04-identification}
\end{equation}
The notation \(\Upsilon_N\) is the specialization of
\(\Upsilon_{\mathbf m}\) in Section~\ref{sec:HP-fixed-degree}
to \(I=\{2\}\), \(\overline I=\{1\}\), and \(m_1=N\).
Here \(|n\rangle\) belongs to the Fock space
\(\F_{\{2\},\{1\}}\), with
\(\mathsf a_{21}|n\rangle=\sqrt n|n-1\rangle\) and
\(\mathsf a_{21}^{\dagger}|n\rangle=\sqrt{n+1}|n+1\rangle\).
The maps \(\Upsilon_N\) act only on the auxiliary space.
With \(S_N=\diag(N^{-1/2},1)\), we shall verify
\begin{equation}
 \lim_{N\to\infty}
 \Upsilon_N^{-1}S_N L^{(\{1\})}(u)S_N\Upsilon_N
 =\mathscr L_{\{2\},\{1\}}(u)
 =\begin{pmatrix}1&\mathsf a_{21}^{\dagger}\\
 \mathsf a_{21}&u+\mathsf N_{21}\end{pmatrix}.
 \label{eq:ex04-L-limit}
\end{equation}
The limit is taken on each fixed \(|n\rangle\otimes e_c\), with
\(c=1,2\). It specializes \eqref{eq:HP-local-L-limit} to
\(M=2\), \(I=\{2\}\), \(R=\overline I=\{1\}\), and \(m_1=N\).
The right hand side is \eqref{eq:contracted-L} in the basis
\((e_1,e_2)\).

Both off-diagonal operators preserve the total occupation number:
\begin{equation}
 \begin{aligned}
 \frac{\mathsf a_{21}^{\dagger}\mathsf a_{11}}{\sqrt N}
 |N-n,n\rangle_{\F}
 &=\sqrt{n+1}\sqrt{1-\frac nN}
 |N-(n+1),n+1\rangle_{\F},\\
 \frac{\mathsf a_{11}^{\dagger}\mathsf a_{21}}{\sqrt N}
 |N-n,n\rangle_{\F}
 &=\sqrt n\sqrt{1-\frac{n-1}{N}}
 |N-(n-1),n-1\rangle_{\F}.
 \end{aligned}
 \label{eq:ex04-L-actions}
\end{equation}
The first line is zero for \(n=N\), and the second for \(n=0\).
For the nonzero vectors, applying \(\Upsilon_N^{-1}\) gives
\(|n+1\rangle\) and \(|n-1\rangle\), respectively.
Including the two diagonal operators, we obtain the exact
identity on \(|n\rangle\) \((0\le n\le N)\):
\begin{equation}
 \Upsilon_N^{-1}S_N L^{(\{1\})}(u)S_N\Upsilon_N
 =\begin{pmatrix}
 1+\dfrac{u-\mathsf N_{21}}{N}
   &\mathsf a_{21}^{\dagger}\sqrt{1-\mathsf N_{21}/N}\\[2mm]
 \sqrt{1-\mathsf N_{21}/N}\,\mathsf a_{21}
   &u+\mathsf N_{21}
 \end{pmatrix}.
 \label{eq:ex04-transported-L}
\end{equation}
For each fixed vector \(|n\rangle\), the actions of the two
off-diagonal matrix elements in \eqref{eq:ex04-transported-L}
contain the factors \(\sqrt{1-n/N}\) and
\(\sqrt{1-(n-1)/N}\), respectively; both tend to \(1\) as
\(N\to\infty\). The first diagonal matrix element acts as
\((1+(u-n)/N)|n\rangle\to|n\rangle\), while
\((u+\mathsf N_{21})|n\rangle=(u+n)|n\rangle\) is independent of
\(N\). These limits prove \eqref{eq:ex04-L-limit} on each fixed
\(|n\rangle\otimes e_c\), \(c=1,2\).
In \(|N-n,n\rangle_{\F}\), the first entry \(N-n\) is the
occupation number of \(\mathsf a_{11}\). The changes in this entry in
\eqref{eq:ex04-L-actions} are expressed by
\[
 \begin{aligned}
 \Upsilon_N|n+1\rangle&=|N-n-1,n+1\rangle_{\F},
   &&0\le n<N,\\
 \Upsilon_N|n-1\rangle&=|N-n+1,n-1\rangle_{\F},
   &&1\le n\le N.
 \end{aligned}
\]
In both cases the total occupation number remains \(N\).
The coefficient of \(u\) becomes \(\diag(0,1)\), giving the degenerate
Yangian \(L\)-operator used to construct \(Q_{\{2\}}\)
\cite{BFLMS2011,ShortcutQ}.

\paragraph{(C) Selecting the terms that contribute to the residue}
For the residue at \(y=g_1^{-1}\), put
\(\varepsilon=1-g_1y\) and \(S_\varepsilon=\diag(\sqrt\varepsilon,1)\).
Using \(\mathsf N_{11}=\mathsf a_{11}\mathsf a_{11}^{\dagger}-1\), 
one can show that 
the same \(\mathscr L_{\{2\},\{1\}}\) as in
\eqref{eq:ex04-L-limit} occurs in the identity
\begin{align}
 S_\varepsilon L^{(\{1\})}(u)S_\varepsilon
 &=\begin{pmatrix}\sqrt\varepsilon\mathsf a_{11}&0\\0&1\end{pmatrix}
   \mathscr L_{\{2\},\{1\}}(u)
   \begin{pmatrix}\sqrt\varepsilon\mathsf a_{11}^{\dagger}&0\\0&1\end{pmatrix}
   \notag\\*
 &\quad+\varepsilon(u-1)E_{11}.
 \label{eq:ex04-L-factorization}
\end{align}
Here multiplication by \(S_\varepsilon\) is the specialization of
\eqref{eq:scaled-local-L} to \(M=2\) and
\(R=\overline I=\{1\}\). We now calculate which terms in
\eqref{eq:ex04-L-factorization} contribute to the residue.

We calculate the contribution of \eqref{eq:ex04-L-factorization} to
\(\langle1,2|\mathcal T(u;y)|1,2\rangle\) using the trace
\eqref{eq:ex04-master-trace}. In the general formula
\eqref{eq:HP-matrix-element-start}, take
\(M=L=2\), \(R=\overline I=\{1\}\), \(I=\{2\}\), and
\(\mathbf c=\mathbf d=(1,2)\).
Using \(\mathsf M_1|1,2\rangle=|1,2\rangle\) and
\(dy=-d\varepsilon/g_1\), the residue in
\eqref{eq:ex04-goal-Q} becomes
\begin{align}
 &-g_1\operatorname*{Res}_{y=g_1^{-1}}
 \bigl[(1-g_1y)\langle1,2|\mathcal T(u;y)|1,2\rangle\,dy\bigr]
 \notag\\*
 &=\operatorname*{Res}_{\varepsilon=0}
 \bigl[\varepsilon\langle1,2|
 \mathcal T(u;(1-\varepsilon)/g_1)|1,2\rangle\,d\varepsilon\bigr]
 \notag\\*
 &=\lim_{\varepsilon\to0}\varepsilon^2\langle1,2|
 \mathcal T(u;(1-\varepsilon)/g_1)|1,2\rangle.
 \label{eq:ex04-residue-normalization}
\end{align}
The last equality uses the pole of order at most two in
\eqref{eq:ex04-worked-master}.

The action of \(S_\varepsilon\) at the two sites is
\begin{align*}
 (S_\varepsilon\otimes S_\varepsilon)|1,2\rangle
   &=(\sqrt\varepsilon e_1)\otimes e_2
    =\sqrt\varepsilon\,|1,2\rangle,\\*
 \langle1,2|(S_\varepsilon\otimes S_\varepsilon)
   &=\sqrt\varepsilon\,\langle1,2|.
\end{align*}
Substituting the expression for \(\mathcal T\) in
\eqref{eq:ex04-master-trace} gives
\begin{align}
 &\varepsilon^2\langle1,2|
    \mathcal T(u;(1-\varepsilon)/g_1)|1,2\rangle\notag\\*
 &=\varepsilon\langle1,2|
   (S_\varepsilon\otimes S_\varepsilon)
   \mathcal T(u;(1-\varepsilon)/g_1)
   (S_\varepsilon\otimes S_\varepsilon)|1,2\rangle\notag\\*
 &=\varepsilon\Tr_{\F_{B,\{1\}}}\!\Bigl[
   \langle1,2|(S_\varepsilon\otimes S_\varepsilon)
   L_{02}^{(\{1\})}(u_2)L_{01}^{(\{1\})}(u_1)
   (S_\varepsilon\otimes S_\varepsilon)|1,2\rangle\notag\\*
 &\hspace{33mm}{}\times
   (1-\varepsilon)^{\mathsf N_{11}}
   [\rho(1-\varepsilon)]^{\mathsf N_{21}}\Bigr]\notag\\*
 &=\varepsilon\Tr_{\F_{B,\{1\}}}\!\Bigl[
   [S_\varepsilon L^{(\{1\})}(u_2)S_\varepsilon]_{22}
   [S_\varepsilon L^{(\{1\})}(u_1)S_\varepsilon]_{11}\notag\\*
 &\hspace{33mm}{}\times
   (1-\varepsilon)^{\mathsf N_{11}}
   [\rho(1-\varepsilon)]^{\mathsf N_{21}}\Bigr].
 \label{eq:ex05-scaled-trace-start}
\end{align}
Equation~\eqref{eq:ex05-scaled-trace-start} also follows from
\eqref{eq:scaled-master-T} after taking the matrix element on
\(|1,2\rangle\) and multiplying by \(\varepsilon\).

The Fock space separates as
\[
 \F_{B,\{1\}}=
 \F_{\{1\},\{1\}}\otimes\F_{\{2\},\{1\}},
\]
where the two factors are the Fock spaces of \(\mathsf a_{11}\) and
\(\mathsf a_{21}\), respectively. Substituting
\eqref{eq:ex04-L-factorization} into the last line of
\eqref{eq:ex05-scaled-trace-start} and separating the two traces gives
\begin{align}
 &\varepsilon^2\langle1,2|
   \mathcal T(u;(1-\varepsilon)/g_1)|1,2\rangle\notag\\*
 &=\varepsilon^2\Tr_{\F_{B,\{1\}}}\!\Bigl[
   \bigl[\mathsf a_{11}\mathsf a_{11}^{\dagger}+(u_1-1)\bigr]
   (u_2+\mathsf N_{21})\notag\\*
 &\hspace{33mm}{}\times(1-\varepsilon)^{\mathsf N_{11}}
   [\rho(1-\varepsilon)]^{\mathsf N_{21}}\Bigr]\notag\\*
 &=\Bigl\{\varepsilon^2\Tr_{\F_{\{1\},\{1\}}}\!\bigl[
       \mathsf a_{11}\mathsf a_{11}^{\dagger}
       (1-\varepsilon)^{\mathsf N_{11}}\bigr]
   +\varepsilon^2(u_1-1)
       \Tr_{\F_{\{1\},\{1\}}}(1-\varepsilon)^{\mathsf N_{11}}\Bigr\}
       \notag\\*
 &\hspace{9mm}{}\times\Tr_{\F_{\{2\},\{1\}}}\!\bigl[
       (u_2+\mathsf N_{21})[\rho(1-\varepsilon)]^{\mathsf N_{21}}\bigr]
       \notag\\*
 &=\Bigl\{\underbrace{\varepsilon^2
       \sum_{k\ge0}(k+1)(1-\varepsilon)^k}_{1}
       +\underbrace{\varepsilon^2(u_1-1)
       \sum_{k\ge0}(1-\varepsilon)^k}_{\varepsilon(u_1-1)}\Bigr\}
       \notag\\*
 &\hspace{9mm}{}\times\Tr_{\F_{\{2\},\{1\}}}\!\bigl[
       (u_2+\mathsf N_{21})[\rho(1-\varepsilon)]^{\mathsf N_{21}}\bigr]
       \notag\\*
 &=\bigl[1+\varepsilon(u_1-1)\bigr]
   \left[\frac{u_2}{1-\rho(1-\varepsilon)}
   +\frac{\rho(1-\varepsilon)}{[1-\rho(1-\varepsilon)]^2}\right].
 \label{eq:ex04-trace-at-pole}
\end{align}
The term \(\varepsilon(u_1-1)E_{11}\) in
\eqref{eq:ex04-L-factorization} has produced the contribution
proportional to \(\varepsilon(u_1-1)\) in
\eqref{eq:ex04-trace-at-pole}. This contribution tends to zero 
in the limit $\varepsilon \to 0$.
The term containing \(\mathsf a_{11}\mathsf a_{11}^{\dagger}\)
contributes the coefficient \(1\). The operator left inside the trace
over \(\F_{\{2\},\{1\}}\) is
\[
 u_2+\mathsf N_{21}
 =[\mathscr L_{\{2\},\{1\}}(u_2)]_{22}
  [\mathscr L_{\{2\},\{1\}}(u_1)]_{11}.
\]
Taking the limit gives
\begin{align}
 &\lim_{\varepsilon\to0}\varepsilon^2
   \langle1,2|\mathcal T(u;(1-\varepsilon)/g_1)|1,2\rangle\notag\\*
 &=\frac{u_2}{1-\rho}+\frac{\rho}{(1-\rho)^2}
   =\Tr_{\F_{\{2\},\{1\}}}\!\bigl[
        (u_2+\mathsf N_{21})\rho^{\mathsf N_{21}}\bigr]\notag\\*
 &=\Tr_{\F_{\{2\},\{1\}}}\!\Bigl[
   [\mathscr L_{\{2\},\{1\}}(u_2)]_{22}
   [\mathscr L_{\{2\},\{1\}}(u_1)]_{11}\rho^{\mathsf N_{21}}\Bigr]
   \notag\\*
 &=\langle1,2|\Tr_{\F_{\{2\},\{1\}}}\!\Bigl[
   \mathscr L_{\{2\},\{1\},2}(u_2)
   \mathscr L_{\{2\},\{1\},1}(u_1)\rho^{\mathsf N_{21}}\Bigr]|1,2\rangle
   \notag\\*
 &=\langle1,2|Q_{\{2\}}(u)|1,2\rangle.
 \label{eq:ex04-worked-Q}
\end{align}
The last equality is the definition \eqref{eq:intro-Q-trace} for
\(M=L=2\), \(I=\{2\}\), and \(\overline I=\{1\}\), with boundary
twist \(\rho^{\mathsf N_{21}}\) from
\eqref{eq:intro-Q-diagonal-factor}. This gives the same matrix element of the trace defining
\(Q_{\{2\}}\).

The residue relation in \eqref{eq:ex04-goal-Q} can be verified
for the other matrix elements in the same way.

\paragraph{(D) Taking the trace that defines \(Q_{\{2\}}\)}
Taking the trace in \eqref{eq:intro-Q-trace} for \(M=L=2\) and
\(I=\{2\}\) gives
\begin{equation}
 Q_{\{2\}}(u)=\frac{1}{1-\rho}
 \begin{pmatrix}
 1&0&0&0\\[2mm]
 0&u_2+\dfrac{\rho}{1-\rho}&\dfrac{1}{1-\rho}&0\\[2mm]
 0&\dfrac{\rho}{1-\rho}&u_1+\dfrac{\rho}{1-\rho}&0\\[2mm]
 0&0&0&\begin{gathered}u_1u_2+\dfrac{(u_1+u_2)\rho}{1-\rho}\\[1mm]
 {}+\dfrac{\rho(1+\rho)}{(1-\rho)^2}\end{gathered}
 \end{pmatrix}.
 \label{eq:ex04-Q-matrix}
\end{equation}
The definition \eqref{eq:intro-Q-trace}, specialized to
\(M=L=2\), \(I=\{2\}\), \(\overline I=\{1\}\), is
\begin{equation}
 Q_{\{2\}}(u)=\Tr_{\F_{\{2\},\{1\}}}\!\bigl[
 \mathscr L_{\{2\},\{1\},2}(u_2)
 \mathscr L_{\{2\},\{1\},1}(u_1)\rho^{\mathsf N_{21}}\bigr].
 \label{eq:ex04-Q-trace}
\end{equation}
Here \(\F_{\{2\},\{1\}}\) is the Fock space of
\(\mathsf a_{21}\) alone, with
\(\mathsf a_{21}|n\rangle=\sqrt n|n-1\rangle\).
Equation~\eqref{eq:ex04-worked-Q} has already evaluated the
\((|1,2\rangle,|1,2\rangle)\) matrix element of this trace.
In each nonzero off-diagonal matrix element,
a creation operator at one site is paired with an annihilation
operator at the other site, so their product has a nonzero trace.

\paragraph{(E) Obtaining \(Q_{\{2\}}\) from a residue}
On the full quantum space, the residue relation is
\begin{equation}
 -g_1\,\Gamma(\mathsf M_1+1)^{-1}
 \operatorname*{Res}_{y=g_1^{-1}}
 \bigl[(1-g_1y)^{\mathsf M_1}\mathcal T(u;y)\,dy\bigr]
 =Q_{\{2\}}(u).
 \label{eq:ex04-residue}
\end{equation}
This is Eq.~\eqref{eq:main-residue-formula} with
\(M=L=2\), \(R=\overline I=\{1\}\), and \(I=\{2\}\).
Here \(C_{\overline I}=1\): no oscillators other than
\(\mathsf a_{11}\) and \(\mathsf a_{21}\) occur.

The state \(|1,2\rangle\) satisfies
\(\mathsf M_1|1,2\rangle=|1,2\rangle\). Multiply the explicit
matrix element \eqref{eq:ex04-worked-master} by \(1-g_1y\).
Its first bracket becomes
\[
 (1-g_1y)\left[\frac{u_1}{1-g_1y}
                +\frac{g_1y}{(1-g_1y)^2}\right]
 =u_1-1+\frac1{1-g_1y}.
\]
Only its last term has a pole at \(y=g_1^{-1}\).
Since \(1-g_1y=-g_1(y-g_1^{-1})\), we obtain
\begin{align}
 &-g_1\operatorname*{Res}_{y=g_1^{-1}}
   \bigl[(1-g_1y)\langle1,2|\mathcal T(u;y)|1,2\rangle\,dy\bigr]
 \notag\\*
 &\qquad=\frac{u_2}{1-\rho}+\frac{\rho}{(1-\rho)^2}
 =\langle1,2|Q_{\{2\}}(u)|1,2\rangle.
 \label{eq:ex04-worked-residue}
\end{align}
The other matrix elements follow by taking the same residues in
\eqref{eq:ex04-master-matrix}. The maximum orders of the poles in its
four rows are \(3,2,2,1\). Multiplication by
\((1-g_1y)^{\mathsf M_1}\) reduces them to at most simple poles.
Including the factor \(\Gamma(\mathsf M_1+1)^{-1}\) gives
\eqref{eq:ex04-Q-matrix}.

\paragraph{(F) Obtaining \(Q_{\{2\}}\) from \(T_{(N)}\) as \(N\to\infty\)}
The operator \(T_{(N)}(u)\) is the coefficient of \(y^N\) in the
expansion \eqref{eq:ex04-goal-series} of the master \(T\)-operator.
We shall verify
\begin{equation}
 \lim_{N\to\infty}g_1^{-N}N^{-\mathsf M_1}T_{(N)}(u)=Q_{\{2\}}(u).
 \label{eq:ex04-coefficient-limit}
\end{equation}
This is the case \(M=L=2\), \(I=\{2\}\),
\(R=\overline I=\{1\}\), and \(m_1=N\) of
Eqs.~\eqref{eq:fixed-coefficient-limit} and
\eqref{eq:scaled-fixed-trace-from-master}.
The tensor product of the factors \(S_N\) on the two quantum spaces is
\[
 \diag(N^{-1/2},1)^{\otimes2}
 =\diag(N^{-1},N^{-1/2},N^{-1/2},1).
\]
Multiplying \(T_{(N)}\) on both sides by this matrix gives
\(N^{-\mathsf M_1}T_{(N)}\), since its only nonzero off-diagonal
matrix elements connect states with the same eigenvalue of
\(\mathsf M_1\).
The factor \(g_1^{-N}\) removes the common factor of the boundary
twist \(g_1^N\rho^n\).

For \(0\le n\le N\), the coefficient of \(\rho^n\) in the
normalized matrix element
\(g_1^{-N}N^{-1}\langle1,2|T_{(N)}(u)|1,2\rangle\) is
\[
 \frac{(u_1+N-n)(u_2+n)}{N}
 =\left(1+\frac{u_1-n}{N}\right)(u_2+n)
 \longrightarrow u_2+n.
\]
Here \(N\to\infty\) with \(n\) fixed.

For each fixed power \(\rho^n\), the condition \(n\le N\) holds for
all sufficiently large \(N\). Hence
\begin{equation}
 \lim_{N\to\infty}\frac{g_1^{-N}}{N}
 \langle1,2|T_{(N)}(u)|1,2\rangle
 =\sum_{n\ge0}(u_2+n)\rho^n
 =\langle1,2|Q_{\{2\}}(u)|1,2\rangle.
 \label{eq:ex04-worked-limit}
\end{equation}
The other matrix elements follow in the same way. 

We have evaluated the trace over \(\F_{B,\{1\}}\) in two ways.
Grouping its terms by the total occupation number gives \(T_{(N)}\) and the Schur function
expansion. Taking a residue of \(\mathcal T(u;y)\), or first selecting
\(T_{(N)}\) and taking the limit \(N\to\infty\), gives the same
\(Q_{\{2\}}\) in \eqref{eq:ex04-Q-matrix}. Both calculations use the
\(L\)-operator \(\mathscr L_{\{2\},\{1\}}\) found in (C).

The common starting point is the definition of the master
\(T\)-operator by a monodromy matrix and a trace over a specified Fock
space. Proposition~\ref{prop:main-contraction} and the results of
Section~\ref{sec:HP-fixed-degree} give the corresponding relations for
arbitrary chain length and all subsets \(I\). The residue formula is
proved in \ref{app:residue-details}. The mKP bilinear identity, established
using the CBR formula, then gives the functional relations among those
\(Q\)-operators \cite{AKLTZ2013,KLT2012}.

\section{Basic conventions}
\label{sec:conventions}
Let $M \in \mathbb{Z}_{\ge 1}$. 
The labels of the basis vectors of the quantum spin space \(\C^M\) are called
\emph{colour} labels.  The set of colour labels is
\begin{equation}
  B=\{1,\ldots,M\}.
  \label{eq:B-colour-set}
\end{equation}
A finite set of labels used to distinguish the families of Schwinger bosons is
called a set of \emph{flavour} labels.  The full finite set of flavour labels
is denoted by \(R\).  When a smaller or different finite set of flavours is
needed, we use letters such as \(S\).  The set \(R\) is not ordered unless an ordering is explicitly chosen.  In the
residue construction we choose the flavour labels so that \(B\subset R\).  This
is only an identification of labels; the colour and flavour vector spaces below
remain different.

The quantum space of a length-\(L\) spin chain is
\begin{equation}
  \Hq_L=(\C^M)^{\otimes L}.
\end{equation}
The matrix unit acting on the \(\ell\)-th site is denoted by \(E_{ij}^{(\ell)}\).  Thus
\begin{equation}
  E_{ij}^{(\ell)}
  =\underbrace{\Id\otimes\cdots\otimes\Id}_{\ell-1\ \text{factors}}
  \otimes E_{ij}
  \otimes\underbrace{\Id\otimes\cdots\otimes\Id}_{L-\ell\ \text{factors}}.
\end{equation}
We use the convention
\begin{equation}
  E_{ij}e_k=\delta_{jk}e_i,
\end{equation}
where \(e_1,\ldots,e_M\) is the basis of \(\C^M\) indexed by the colour set
\(B\).  Throughout the paper, \(\Id_X\) denotes the identity operator on a vector space \(X\).  
If the vector space is clear from the context, we write simply \(\Id\).  
If \(A\in\End(\C^M)\), then \([A]_{cd}\) denotes its
\((c,d)\)-matrix element with respect to the basis \(e_i\), \(i\in B\).

When residues of the master \(T\)-operator at the points specified below are considered, we assume that the boundary twist of the transfer matrices is diagonal:
\begin{equation}
  g=\diag(g_1,\ldots,g_M),
  \qquad g_i\in\C^\times,
  \qquad g_i\ne g_j\quad(i\ne j).
  \label{eq:diagonal-twist}
\end{equation}
The non-zero assumption is needed because the residue points are
\(y_i=g_i^{-1}\).  The assumption that the \(g_i\) are distinct keeps
these points separate and rules out additional poles in the trace
calculations below.  
The inhomogeneity parameters are denoted by
\(\theta_1,\ldots,\theta_L \in \mathbb{C}\).

\subsection{Schwinger bosons}

For any finite subset $J\subset B$ we put
\begin{equation}
  V_J:=\bigoplus_{i\in J}\C e_i .
  \label{eq:VJ-definition}
\end{equation}
We write
\begin{equation}
  \gl(V_J):=\End(V_J)
  \label{eq:gl-VJ-def}
\end{equation}
for the general linear Lie algebra acting on \(V_J\).  For example,
\(\gl(V_I)\) means \(\End(V_I)\).
For any finite set $S$ of flavour labels we put
\begin{equation}
  U_S:=\bigoplus_{\alpha\in S}\C f_\alpha .
  \label{eq:US-definition}
\end{equation}
Here $f_\alpha$ denotes the basis vector of $U_S$ labelled by the flavour label $\alpha$.  Thus $e_i$ and $f_\alpha$ are basis vectors in different vector spaces: $e_i\in V_J$ is labelled by a colour, whereas $f_\alpha\in U_S$ is labelled by a flavour.
The symbol
\begin{equation}
  \operatorname{Fock}(V_J\otimes U_S)
  \label{eq:Fock-notation-fixed}
\end{equation}
is used in this paper with the following fixed meaning.  For brevity we also
write
\begin{equation}
  \F_{J,S}:=\operatorname{Fock}(V_J\otimes U_S)
  \label{eq:FJS-notation-fixed}
\end{equation}
for finite sets $J\subset B$ and $S$ of flavour labels.  This is the
lowest-weight Fock representation of the bosonic oscillator algebra generated by
\begin{equation}
  \mathsf a_{i\alpha}^{\dagger},\qquad \mathsf a_{i\alpha},
  \qquad i\in J,
  \quad \alpha\in S,
\end{equation}
with relations
\begin{equation}
  [\mathsf a_{i\alpha},\mathsf a_{j\beta}^{\dagger}]
  =\delta_{ij}\delta_{\alpha\beta},
  \qquad
  [\mathsf a_{i\alpha},\mathsf a_{j\beta}]
  =[\mathsf a_{i\alpha}^{\dagger},\mathsf a_{j\beta}^{\dagger}]=0.
  \label{eq:basic-boson-relations-fixed}
\end{equation}
Note that the left indices of these generators are labelled by elements of the colour set, whereas the right indices are labelled by elements of the flavour set. 
The vacuum vector is part of the definition.  It is denoted by
$|0\rangle_{J,S}$ and is characterized by
\begin{equation}
  \mathsf a_{i\alpha}|0\rangle_{J,S}=0,
  \qquad i\in J,
  \quad \alpha\in S.
  \label{eq:Fock-vacuum-fixed}
\end{equation}
The representation space is the linear span of the occupation vectors
\begin{equation}
  |\mathbf n\rangle_{J,S}
  =
  \prod_{i\in J}\prod_{\alpha\in S}
  \frac{(\mathsf a_{i\alpha}^{\dagger})^{n_{i\alpha}}}{\sqrt{n_{i\alpha}!}}
  |0\rangle_{J,S},
  \qquad n_{i\alpha}\in\Z_{\ge0},
  \label{eq:occupation-basis-fixed}
\end{equation}
where \(\mathbf n=(n_{i\alpha})_{i\in J,\alpha\in S}\).  Any fixed order
of the commuting creation operators may be used.  The dual vector
\(\langle \mathbf n|_{J,S}\) is defined by the dual occupation basis, so that,
for another family \(\mathbf m=(m_{i\alpha})_{i\in J,\alpha\in S}\),
\begin{equation}
  \langle \mathbf n|\mathbf m\rangle_{J,S}=\delta_{\mathbf n,\mathbf m}.
  \label{eq:dual-occupation-basis}
\end{equation}
For a single pair \((i,\alpha)\), the creation and annihilation operators act by
\begin{align}
  \mathsf a_{i\alpha}|\ldots,n_{i\alpha},\ldots\rangle_{J,S}
  &=\sqrt{n_{i\alpha}}\,|\ldots,n_{i\alpha}-1,\ldots\rangle_{J,S},
  \label{eq:annihilation-action-basis}
  \\
  \mathsf a_{i\alpha}^{\dagger}|\ldots,n_{i\alpha},\ldots\rangle_{J,S}
  &=\sqrt{n_{i\alpha}+1}\,|\ldots,n_{i\alpha}+1,\ldots\rangle_{J,S}.
  \label{eq:creation-action-basis}
\end{align}
These formulas follow directly from \eqref{eq:basic-boson-relations-fixed} and the normalization in \eqref{eq:occupation-basis-fixed}.  The
number operator of the oscillator labelled by $(i,\alpha)$ is
\begin{equation}
  \mathsf N_{i\alpha}=\mathsf a_{i\alpha}^{\dagger}\mathsf a_{i\alpha},
  \qquad
  \mathsf N_{i\alpha}|\mathbf n\rangle_{J,S}=n_{i\alpha}|\mathbf n\rangle_{J,S}.
  \label{eq:number-operator-fixed}
\end{equation}
For a flavour label \(\alpha\in S\), the total number of particles
with that flavour in \(\F_{J,S}\) is
\begin{equation}
  \mathsf N_\alpha^{(J,S)}=\sum_{i\in J}\mathsf N_{i\alpha}.
  \label{eq:flavour-number-operator}
\end{equation}
When \(J\) and \(S\) are clear, we write this operator simply as
\(\mathsf N_\alpha\).  In particular, on \(\F_{B,S}\) this means
\(\mathsf N_\alpha=\sum_{c\in B}\mathsf N_{c\alpha}\).
Whenever we write a trace over this Fock representation, it is the trace in the
occupation basis \eqref{eq:occupation-basis-fixed}:
\begin{equation}
  \Tr_{\F_{J,S}} X
  :=
  \sum_{\mathbf n}
  \langle \mathbf n|X|\mathbf n\rangle_{J,S},
  \label{eq:Fock-trace-fixed}
\end{equation}
provided the series converges.  When it does not converge as an analytic
series, we use the same formula algebraically: the diagonal factors are
expanded in the occupation basis, and the resulting coefficients are computed
term by term.  This is what we call the formal trace in the occupation basis.
If \(\rho_{i\alpha}\) are complex parameters, then
\begin{equation}
  \Tr_{\F_{J,S}}
  \prod_{i\in J,\alpha\in S}\rho_{i\alpha}^{\mathsf N_{i\alpha}}
  =
  \prod_{i\in J,\alpha\in S}\frac{1}{1-\rho_{i\alpha}}
  \label{eq:Fock-character-fixed}
\end{equation}
in the domain $|\rho_{i\alpha}|<1$, and elsewhere by meromorphic continuation
or as a formal product. 

For a finite flavour set $R$, this convention gives
\begin{equation}
  \F_{B,R}=\operatorname{Fock}(V_B\otimes U_R).
  \label{eq:FBR-fixed}
\end{equation}
For this space we use the abbreviated notation
\(\mathsf N_\alpha=\sum_{c\in B}\mathsf N_{c\alpha}\), as explained above.
Later, when a subset $I\subset B$ and its complement
$\overline I=B\setminus I$ are fixed, the condition
$\mathsf N_\beta=m_\beta$ for $\beta\in\overline I$ means that the total number
of particles with flavour $\beta$ is fixed to $m_\beta$.
If one wants to allow arbitrarily many flavour labels, each formula below should first be understood for finite \(R\).  
The coefficientwise limit is then taken by fixing a coefficient and increasing \(R\) until that coefficient no longer changes.
The Schwinger realization of $\gl(M)$ is
\begin{equation}
  e_{ij}^{(R)}=\sum_{\alpha\in R}\mathsf a_{i\alpha}^\dagger \mathsf a_{j\alpha}.
  \label{eq:Schwinger-eij-R}
\end{equation}
A direct computation gives
\begin{align}
 [e_{ij}^{(R)},e_{kl}^{(R)}]
 &=\sum_{\alpha,\beta\in R}
 [\mathsf a_{i\alpha}^\dagger \mathsf a_{j\alpha},\mathsf a_{k\beta}^\dagger \mathsf a_{l\beta}]
 \notag\\
 &=\sum_{\alpha\in R}
 \left(\delta_{jk}\mathsf a_{i\alpha}^\dagger \mathsf a_{l\alpha}
      -\delta_{li}\mathsf a_{k\alpha}^\dagger \mathsf a_{j\alpha}\right)
 \notag\\
 &=\delta_{jk}e_{il}^{(R)}-\delta_{li}e_{kj}^{(R)}.
 \label{eq:eij-commutator}
\end{align}
Thus $e_{ij}^{(R)}$ satisfy the defining relations of $\gl(M)$.

\subsection{The \texorpdfstring{$R$}{R}-matrix and the \texorpdfstring{$L$}{L}-operator in the Schwinger representation}

We first fix the meaning of lower indices attached to operators in tensor products.
Let \(W_1,\ldots,W_n\) be vector spaces, and let
\(A=\sum_s A_s'\otimes A_s''\in \End(U)\otimes\End(V)\).
Suppose that the \(a\)-th and \(b\)-th tensor factors of
\(W_1\otimes\cdots\otimes W_n\) are identified with \(U\) and \(V\),
respectively, with \(a<b\).  Then \(A_{ab}\) denotes
\begin{equation}
  A_{ab}
  =\sum_s
  \Id_{W_1}\otimes\cdots\otimes A_s'
  \otimes\cdots\otimes A_s''
  \otimes\cdots\otimes\Id_{W_n}.
  \label{eq:operator-embedding-convention}
\end{equation}
Here \(A_s'\) is placed in the \(a\)-th tensor factor and \(A_s''\) in the
\(b\)-th tensor factor; all omitted tensor factors carry the identity
operator.

The \(R\)-matrix of the rational \(\gl(M)\) spin chain is
\begin{equation}
  R(u)=u\,\Id_{\C^M\otimes\C^M}+P,
  \qquad
  P=\sum_{i,j\in B} E_{ij}\otimes E_{ji}
  \in \End(\C^M\otimes\C^M),
  \label{eq:rational-R-matrix}
\end{equation}
where \(u\in\mathbb C\) is the spectral parameter.  It satisfies the
Yang--Baxter equation in \(\End(\C^M\otimes\C^M\otimes\C^M)\),
\begin{equation}
  R_{12}(u)R_{13}(v)R_{23}(v-u)
  =R_{23}(v-u)R_{13}(v)R_{12}(u),
  \label{eq:R-matrix-YBE}
\end{equation}
where \(u,v\in\mathbb C\).

For a flavour set \(R\), the Schwinger \(L\)-operator is defined by
\begin{equation}
  L^{(R)}(v)
  =v\,\Id_{\F_{B,R}}\otimes\Id_{\C^M}
  +\sum_{i,j\in B} e_{ij}^{(R)}\otimes E_{ji}
  \in \End(\F_{B,R})\otimes\End(\C^M).
  \label{eq:Schwinger-L}
\end{equation}
The first tensor factor is the Fock representation \(\F_{B,R}\), and the
second tensor factor is \(\C^M\) with basis labelled by \(B\).  The
operator acting on the auxiliary space and on the \(\ell\)-th site of the
spin chain is
\begin{equation}
  L_{0\ell}^{(R)}(v)
  =v\,\Id_{\F_{B,R}}\otimes\Id_{\Hq_L}
  +\sum_{i,j\in B} e_{ij}^{(R)}\otimes E_{ji}^{(\ell)}
  \in \End(\F_{B,R}\otimes\Hq_L).
  \label{eq:Schwinger-L-embedded}
\end{equation}
Here the label \(0\) refers to the auxiliary space \(\F_{B,R}\), whereas
\(1,\ldots,L\) refer to the sites of the spin chain.

The \(L\)-operator \eqref{eq:Schwinger-L} satisfies the \(RLL\) relation in
\(\End(\F_{B,R}\otimes\C^M\otimes\C^M)\),
\begin{equation}
  L_{12}^{(R)}(u)L_{13}^{(R)}(v)R_{23}(v-u)
  =R_{23}(v-u)L_{13}^{(R)}(v)L_{12}^{(R)}(u).
  \label{eq:Schwinger-RLL}
\end{equation}
We also use the ordered product (monodromy matrix)
\begin{equation}
  M^{(R)}(u)
  =L_{0L}^{(R)}(u-\theta_L)\cdots L_{01}^{(R)}(u-\theta_1),
  \label{eq:Schwinger-monodromy}
\end{equation}
where \(\theta_1,\ldots,\theta_L\in\mathbb C\) are inhomogeneity
parameters.

\section{The master \texorpdfstring{$T$}{T}-operator in the Schwinger representation}
\label{sec:master-schwinger}

Let \(R\) be a finite flavour set.  For \(\alpha\in R\), let
\(\mathsf N_\alpha\) be the operator defined in
\eqref{eq:flavour-number-operator} with \(J=B\) and \(S=R\).  For a family of
variables \(\mathbf y_R=(y_\alpha)_{\alpha\in R}\), define
\begin{equation}
  \widehat h(\mathbf y_R)
  :=\prod_{\alpha\in R}y_\alpha^{\mathsf N_\alpha}.
  \label{eq:hhat-yR}
\end{equation}
Once the variables have been fixed, we abbreviate \(\widehat h(\mathbf y_R)\) to \(\widehat h\).  For diagonal \(g=\diag(g_1,\ldots,g_M)\), set
\begin{equation}
  \widehat g
  :=\prod_{i\in B}\prod_{\alpha\in R}g_i^{\mathsf N_{i\alpha}}.
  \label{eq:ghat-R}
\end{equation}
For non-diagonal \(g\), \(\widehat g\) denotes the operator on
\(\F_{B,R}\) induced by the action of \(g\) on \(V_B\), under the
identification \(\F_{B,R}=\Sym(V_B\otimes U_R)\) in
\eqref{eq:Fock-as-sym}.  The twist \(g\) is fixed throughout the paper.
In the diagonal case, \(\widehat g\) and \(\widehat h(\mathbf y_R)\) are
the actions of the Cartan tori of \(GL(V_B)\) and \(GL(U_R)\),
respectively.  Their infinitesimal generators lie in the Cartan
subalgebras of \(\gl(V_B)\) and \(\gl(U_R)\), and the two actions commute
on \(\F_{B,R}\).

\begin{definition}[Master \(T\)-operator with finitely many flavour variables]
\label{def:finite-master}
The master \(T\)-operator in the Schwinger representation with flavour set \(R\), spectral parameter \(u\in\mathbb C\), and inhomogeneity parameters \(\theta_1,\dots,\theta_L\in\mathbb C\), is
\begin{equation}
  \mathcal T(u;\mathbf y_R)
  =\Tr_{\F_{B,R}}
  \left[
  L_{0L}^{(R)}(u-\theta_L)\cdots L_{01}^{(R)}(u-\theta_1)
  \widehat g\,\widehat h(\mathbf y_R)
  \right].
  \label{eq:finite-master-def}
\end{equation}
It is an element of \(\End(\Hq_L)\).
\end{definition}

The trace is the trace in the occupation-number basis of the Fock
representation defined in Section~\ref{sec:conventions}.  For diagonal \(g\),
the expansion of the ordered product gives finite sums of products of geometric
series, so it converges in the domain \(|g_i y_\alpha|<1\) for all
\(i\in B\) and \(\alpha\in R\).  Elsewhere it is understood by meromorphic continuation, or coefficientwise as
a formal power series in \(\mathbf y_R\).

\subsection{Schur functions}
\label{subsec:Schur-functions}
We use the standard conventions for Schur functions, the Jacobi--Trudi formula,
and Cauchy's identity as in Ref.~\cite{Macdonald1995}.
We fix the convention for partitions, Schur functions, and polynomial
representations of general linear groups.  A partition is a sequence
\(\lambda=(\lambda_1,\lambda_2,\ldots)\) of non-negative integers such that
\(\lambda_1\ge\lambda_2\ge\cdots\) and \(|\lambda|:=\sum_i\lambda_i<\infty\).
We identify a partition with its Young diagram.  Its length is
\begin{equation}
  \ell(\lambda)=\#\{i\mid \lambda_i>0\}.
  \label{eq:partition-length-def}
\end{equation}
For a finite-dimensional vector space \(X\), let \(V_\lambda^{GL(X)}\) denote
the polynomial irreducible representation of \(GL(X)\) with highest weight
\(\lambda\), if \(\ell(\lambda)\le\dim X\).  If
\(\ell(\lambda)>\dim X\), we set \(V_\lambda^{GL(X)}=0\).  The differential of this group representation is denoted by
\(\pi_\lambda^X\) as well.  Thus, for \(A\in\End(X)\),
\[
  \pi_\lambda^X(A)
  =\left.\frac{d}{dt}\pi_\lambda^X(e^{tA})\right|_{t=0}.
\]
The superscript is displayed when it is useful to specify the vector space.  If the vector space is clear from the
context, the superscript is omitted.  The same conventions are used for
partitions denoted by capital letters, for example \(\Lambda\) in
\ref{app:nontrivial-glI}.

For the KP time variables \(\wt=(t_1,t_2,\ldots)\), define the complete
symmetric functions \(h_n(\wt)\) by
\begin{equation}
  \sum_{n\ge0} h_n(\wt)q^n
  =\exp\left(\sum_{k\ge1}t_k q^k\right),
  \qquad
  h_0(\wt)=1,
  \qquad
  h_n(\wt)=0\quad(n<0).
  \label{eq:complete-functions-times}
\end{equation}
For a partition \(\lambda\), choose any integer \(d\ge \ell(\lambda)\) and set
\begin{equation}
  s_\lambda(\wt)
  =\det_{1\le i,j\le d}
  \bigl(h_{\lambda_i-i+j}(\wt)\bigr).
  \label{eq:Schur-Jacobi-Trudi-times}
\end{equation}
The determinant does not depend on the choice of \(d\), once \(d\ge\ell(\lambda)\).
For a finite family of variables \(\mathbf y_R=(y_\alpha)_{\alpha\in R}\), the
Schur function \(s_\lambda(\mathbf y_R)\) is defined by the same formula with
\begin{equation}
  \sum_{n\ge0} h_n(\mathbf y_R)q^n
  =\prod_{\alpha\in R}\frac{1}{1-y_\alpha q}.
  \label{eq:complete-functions-finite-y}
\end{equation}
Thus \(s_\lambda(\mathbf y_R)=0\) if \(\ell(\lambda)>|R|\).  The same
symbols \(h_n\) and \(s_\lambda\) are used for the KP time variables
\(\wt\) and for the finite variables \(\mathbf y_R\); the argument specifies
which specialization is meant.

\subsection{Zero sites and Cauchy's identity}

When \(L=0\), the monodromy matrix is the empty product.  Hence
\begin{align}
  \left.\mathcal T(u;\mathbf y_R)\right|_{L=0}
  &=\Tr_{\F_{B,R}}(\widehat g\widehat h(\mathbf y_R))
  \notag\\
  &=\prod_{i\in B}\prod_{\alpha\in R}
  \sum_{n=0}^{\infty}(g_i y_\alpha)^n
  \notag\\
  &=\prod_{i\in B}\prod_{\alpha\in R}
  \frac1{1-g_i y_\alpha}.
  \label{eq:zero-Cauchy-kernel}
\end{align}
Cauchy's identity gives
\begin{equation}
  \prod_{i\in B}\prod_{\alpha\in R}
  \frac1{1-g_i y_\alpha}
  =\sum_{\lambda:\,\ell(\lambda)\le \min(M,|R|)}
  s_\lambda(g_1,\ldots,g_M)s_\lambda(\mathbf y_R).
  \label{eq:Cauchy}
\end{equation}
This is the zero-site master \(T\)-operator.

\subsection{Howe decomposition and the Schur function expansion}
\label{subsec:Howe-detailed}

We now derive the Schur function expansion of the master \(T\)-operator from
Howe duality.  In the present paper this expansion is a theorem, not a part of
the definition.  We use Howe duality; we recall the statement and then give the
trace calculation which turns the definition using Schwinger bosons
\eqref{eq:finite-master-def} into the expansion
\eqref{eq:Schur-expansion-main}.

Let
\[
  V=\C^M,\qquad W=U_R.
\]
The Fock representation is identified with the symmetric algebra
\begin{equation}
  \F_{B,R}=\Sym(V\otimes W)
  =\bigoplus_{n\ge0}\Sym^n(V\otimes W).
  \label{eq:Fock-as-sym}
\end{equation}
Here \(\Sym(V\otimes W)\) is the polynomial algebra on \(V\otimes W\).
If \(v_i\) is the basis of \(V\) and \(f_\alpha\), \(\alpha\in R\), is the
basis of \(W\), the vector \(v_i\otimes f_\alpha\) is represented by the
polynomial variable \(x_{i\alpha}\).  In this realization the creation and
annihilation operators act as
\begin{equation}
  \mathsf a_{i\alpha}^\dagger=x_{i\alpha},
  \qquad
  \mathsf a_{i\alpha}= \frac{\partial}{\partial x_{i\alpha}}.
  \label{eq:oscillator-polynomial-realization}
\end{equation}
The group \(GL(V)\times GL(W)\) acts on \(V\otimes W\) by
\begin{equation}
  (g,h)\cdot (v\otimes w)=(gv)\otimes(hw),
  \qquad g\in GL(V),\quad h\in GL(W),
  \label{eq:GLM-GLr-action-on-one-particle}
\end{equation}
and this action extends to \(\Sym(V\otimes W)\).  The two group actions commute.
Let \(E_{ij}^{V}\in\End(V)\) and \(F_{\alpha\beta}^{W}\in\End(W)\) be
the matrix units in the bases \((v_i)\) and \((f_\alpha)\).  In the Schwinger
boson realization, the infinitesimal action of \(E_{ij}^{V}\) is represented by
\begin{equation}
  \sum_{\alpha\in R}x_{i\alpha}\frac{\partial}{\partial x_{j\alpha}}
  =e_{ij}^{(R)}.
  \label{eq:Schwinger-as-differential-Howe}
\end{equation}
The infinitesimal action of \(F_{\alpha\beta}^{W}\) is represented by
\begin{equation}
  \sum_{i\in B}x_{i\alpha}\frac{\partial}{\partial x_{i\beta}},
  \qquad \alpha,\beta\in R.
  \label{eq:flavour-glr-differential-action}
\end{equation}

\begin{theorem}[Howe decomposition]
\label{thm:Howe-decomposition-paper}
As a representation of \(GL(V)\times GL(W)\), this Fock representation
\(\F_{B,R}=\Sym(V\otimes W)\) decomposes as
\begin{equation}
  \F_{B,R}
  \simeq
  \bigoplus_{\lambda:\,\ell(\lambda)\le \min(M,|R|)}
  V_\lambda^{GL(V)}\otimes V_\lambda^{GL(W)}.
  \label{eq:Howe-decomp-paper}
\end{equation}
Each summand appears with multiplicity one.  This is the classical Howe, or \(GL\)--\(GL\) Schur--Weyl, duality
decomposition for this symmetric algebra
\cite{Howe1989,GoodmanWallach1998}.
\end{theorem}

\begin{theorem}[Schur function expansion from Howe duality]
\label{thm:Schur-expansion-from-Howe}
The operator defined in \eqref{eq:finite-master-def} satisfies
\begin{equation}
  \mathcal T(u;\mathbf y_R)
  =
  \sum_{\lambda:\,\ell(\lambda)\le \min(M,|R|)}
  s_\lambda(\mathbf y_R)\,T_\lambda(u),
  \label{eq:Schur-expansion-main}
\end{equation}
where
\begin{equation}
  T_\lambda(u)
  =\Tr_{V_\lambda^{GL(V)}}
  \left[
  L_{0L}^{\lambda}(u-\theta_L)\cdots L_{01}^{\lambda}(u-\theta_1)
  \pi_\lambda(g)
  \right],
  \label{eq:fused-transfer-definition}
\end{equation}
and
\begin{equation}
  L_{0\ell}^{\lambda}(v)
  =v\,\Id_{V_\lambda}\otimes\Id_{\Hq_L}
  +\sum_{i,j\in B}\pi_\lambda(E_{ij})\otimes E_{ji}^{(\ell)}.
  \label{eq:L-lambda-definition}
\end{equation}
The trace in \eqref{eq:finite-master-def} is understood either in a domain of
absolute convergence or coefficientwise as a formal power series in
\(\mathbf y_R\).
\end{theorem}

\begin{proof}
By Theorem~\ref{thm:Howe-decomposition-paper}, the Fock representation is the
direct sum \eqref{eq:Howe-decomp-paper}.  We examine the operator inside the
trace \eqref{eq:finite-master-def} on one fixed summand
\[
  V_\lambda^{GL(V)}\otimes V_\lambda^{GL(W)}.
\]

By \eqref{eq:Schwinger-as-differential-Howe}, the operator \(e_{ij}^{(R)}\) represents the infinitesimal action of \(E_{ij}\in\gl(M)\).  Therefore
on this component it acts as
\begin{equation}
  e_{ij}^{(R)}
  \big|_{V_\lambda^{GL(V)}\otimes V_\lambda^{GL(W)}}
  =
  \pi_\lambda(E_{ij})\otimes \Id_{V_\lambda^{GL(W)}}.
  \label{eq:eij-action-on-Howe-summand}
\end{equation}
We shall also use the embedding convention \eqref{eq:operator-embedding-convention}
in the tensor product
\[
  V_\lambda^{GL(V)}\otimes V_\lambda^{GL(W)}\otimes\Hq_L.
\]
For an operator \(X\in
\End(V_\lambda^{GL(V)}\otimes\Hq_L)\), let \(\iota_{13}(X)\) denote
the operator on this tensor product which acts as \(X\) on the first and
third tensor factors and as the identity on the middle tensor factor.
The subscript \(13\) in \(\iota_{13}\) refers only to this threefold
tensor product and is separate from the lower indices \(0\ell\) in
\(L_{0\ell}^{\lambda}(v)\).
Substituting \eqref{eq:eij-action-on-Howe-summand} into the Schwinger
\(L\)-operator \eqref{eq:Schwinger-L} gives
\begin{equation}
L_{0\ell}^{(R)}(v)
\big|_{V_\lambda^{GL(V)}\otimes V_\lambda^{GL(W)}\otimes\Hq_L}
=
\iota_{13}\left(L_{0\ell}^{\lambda}(v)\right).
\label{eq:L-restriction-Howe-summand}
\end{equation}
Equivalently, the \(GL(W)\)-representation in this summand of the Howe decomposition is acted on
trivially.
Multiplying the factors \(L_{0\ell}^{\lambda}(u-\theta_\ell)\) in the order
\(\ell=L,L-1,\ldots,1\), we obtain
\begin{equation}
M^{(R)}(u)
\big|_{V_\lambda^{GL(V)}\otimes V_\lambda^{GL(W)}\otimes\Hq_L}
=
\iota_{13}\left(
 L_{0L}^{\lambda}(u-\theta_L)\cdots L_{01}^{\lambda}(u-\theta_1)
\right).
\label{eq:monodromy-restriction-Howe-summand}
\end{equation}

Second, the operator induced by \(g\in GL(V)\) on \(\Sym(V\otimes W)\) acts as \(\pi_\lambda(g)\) on the
\(GL(V)\)-representation in the summand of the Howe decomposition, while
\(h\in GL(W)\), with eigenvalues \((y_\alpha)_{\alpha\in R}\), acts as
\(\pi_\lambda^W(h)\) on \(V_\lambda^{GL(W)}\).  Hence
\begin{equation}
  \widehat g\,\widehat h(\mathbf y_R)
  \big|_{V_\lambda^{GL(V)}\otimes V_\lambda^{GL(W)}}
  =
  \pi_\lambda(g)\otimes \pi_\lambda^{W}(h).
  \label{eq:twist-restriction-Howe-summand}
\end{equation}
Put
\[
X_\lambda(u)=L_{0L}^{\lambda}(u-\theta_L)\cdots L_{01}^{\lambda}(u-\theta_1).
\]
Combining \eqref{eq:monodromy-restriction-Howe-summand} and
\eqref{eq:twist-restriction-Howe-summand}, the contribution of this summand
to \eqref{eq:finite-master-def} is
\begin{align}
&\Tr_{V_\lambda^{GL(V)}\otimes V_\lambda^{GL(W)}}
\left[
\iota_{13}\!\left(X_\lambda(u)\right)
\left(\pi_\lambda(g)\otimes\pi_\lambda^{W}(h)\otimes\Id_{\Hq_L}\right)
\right]
\notag\\
&=
\Tr_{V_\lambda^{GL(V)}}
\left[
X_\lambda(u)\pi_\lambda(g)
\right]
\Tr_{V_\lambda^{GL(W)}}
\left[
\pi_\lambda^{W}(h)
\right]
\notag\\
&=T_\lambda(u)\,s_\lambda(\mathbf y_R).
\label{eq:trace-factor-Howe}
\end{align}
The last equality uses the definition \eqref{eq:fused-transfer-definition} of
\(T_\lambda(u)\) and the character identity
\begin{equation}
  \Tr_{V_\lambda^{GL(W)}}\pi_\lambda^{W}(h)=s_\lambda(\mathbf y_R).
  \label{eq:GLr-character-is-Schur}
\end{equation}
Finally, the direct-sum decomposition \eqref{eq:Howe-decomp-paper} allows the
trace over \(\F_{B,R}\) to be written as the sum of the traces over the
summands in \eqref{eq:Howe-decomp-paper}.  Summing \eqref{eq:trace-factor-Howe} over \(\lambda\) gives
\eqref{eq:Schur-expansion-main}.
\end{proof}

With the convention \(V_\lambda^{GL(V)}=0\) for \(\ell(\lambda)>M\), one may instead sum over \(\ell(\lambda)\le |R|\).
Thus the Schur function expansion follows from Howe duality and from the
factorization of the trace; it is not an independent definition of the master
\(T\)-operator in this paper.

\subsection{KP time variables and the Miwa transformation}
\label{subsec:KP-times-Miwa}

The master \(T\)-operator as a function of the KP time variables is the formal series
\begin{equation}
  \mathcal T(u,\wt)
  =\sum_{\lambda:\,\ell(\lambda)\le M}
  s_\lambda(\wt)T_\lambda(u).
  \label{eq:full-master-paper}
\end{equation}
For a finite flavour set \(R\), define the Miwa transformation \(\wt(\mathbf y_R)\)
by
\begin{equation}
  t_k(\mathbf y_R)=\frac1k\sum_{\alpha\in R}y_\alpha^k,
  \qquad k\ge1.
  \label{eq:Miwa-times-paper}
\end{equation}
Then
\begin{equation}
  \exp\left(\sum_{k\ge1}t_k(\mathbf y_R)q^k\right)
  =\exp\left(\sum_{\alpha\in R}\sum_{k\ge1}\frac{(y_\alpha q)^k}{k}\right)
  =\prod_{\alpha\in R}\frac{1}{1-y_\alpha q}.
  \label{eq:Miwa-complete-functions}
\end{equation}
Comparing \eqref{eq:Miwa-complete-functions} with
\eqref{eq:complete-functions-times} and \eqref{eq:complete-functions-finite-y}
shows that
\begin{equation}
  h_n(\wt(\mathbf y_R))=h_n(\mathbf y_R),
  \qquad
  s_\lambda(\wt(\mathbf y_R))=s_\lambda(\mathbf y_R).
  \label{eq:Schur-Miwa-specialization}
\end{equation}
Therefore the finite-flavour trace is related to the master \(T\)-operator
as a function of the KP time variables by the identity
\begin{align}
  \mathcal T(u,\wt(\mathbf y_R))
  &=\sum_{\lambda:\,\ell(\lambda)\le M}
  s_\lambda(\wt(\mathbf y_R))T_\lambda(u)
  \notag\\
  &=\sum_{\lambda:\,\ell(\lambda)\le \min(M,|R|)}
  s_\lambda(\mathbf y_R)T_\lambda(u)
  \notag\\
  &=\mathcal T(u;\mathbf y_R).
  \label{eq:finite-master-is-Miwa-specialization}
\end{align}
The middle equality uses \(s_\lambda(\mathbf y_R)=0\) for
\(\ell(\lambda)>|R|\), and the last equality is
Theorem~\ref{thm:Schur-expansion-from-Howe}.

\subsection{A possible extension with a second monodromy matrix}
\label{subsec:second-monodromy}
We briefly mention a possible extension of the master \(T\)-operator
\eqref{eq:finite-master-def}.  Besides the monodromy matrix \(M^{(R)}(u)\)
in \eqref{eq:Schwinger-monodromy}, the same Fock representation carries a
commuting action of \(GL(U_R)\).  This second action is generated by
\begin{align}
  f^{(B)}_{\alpha\beta}
  =\sum_{i\in B}
  \mathsf a_{i\alpha}^{\dagger}\mathsf a_{i\beta},
  \qquad \alpha,\beta\in R,
\end{align}
and it commutes with the \(GL(V_B)\)-action generated by the operators
\(e^{(R)}_{ij}\).  Let \(\mathcal K_{L'}=U_R^{\otimes L'}\), and let
\(\theta'_1,\ldots,\theta'_{L'}\in\mathbb C\) be another set of
inhomogeneity parameters.  The \(GL(U_R)\) action gives the rational
\(L\)-operator with the spectral parameter $w \in \mathbb{C}$:
\begin{align}
  L^{(B)}(w)
  =w\,\Id_{\mathcal F_{B,R}}\otimes\Id_{U_R}
  +\sum_{\alpha,\beta\in R}
   f^{(B)}_{\alpha\beta}\otimes E^R_{\beta\alpha}
  \in \End(\mathcal F_{B,R})\otimes\End(U_R),
\end{align}
where \(E^R_{\alpha\beta}\) are the matrix units on \(U_R\).  We write
\(L^{(B)}_{0j'}(w)\) for the embedding which acts on \(\mathcal F_{B,R}\)
and on the \(j'\)-th copy of \(U_R\) in \(\mathcal K_{L'}\).  Define the
second monodromy matrix by
\begin{align}
  M_R(w)=
  L^{(B)}_{0L'}(w-\theta'_{L'})\cdots
  L^{(B)}_{01'}(w-\theta'_1),
  \qquad  M_B(u)=M^{(R)}(u).
\end{align}
This suggests considering the trace
\begin{align}
  \mathcal T_{L,L'}(u,w;g,\mathbf y_R)
  =
  \Tr_{\mathcal F_{B,R}}
  \left[
    M_B(u)M_R(w)\,\widehat g\,\widehat h(\mathbf y_R)
  \right]
  \in \End(\mathcal H_L\otimes\mathcal K_{L'}).
\end{align}
By the Howe decomposition \eqref{eq:Howe-decomp-paper}, the trace
factorizes over each summand and gives
\begin{align}
  \mathcal T_{L,L'}(u,w;g,\mathbf y_R)
  =
  \sum_{\lambda:\,\ell(\lambda)\le \min(M,|R|)}
  T^B_\lambda(u;g)\otimes T^R_\lambda(w;\mathbf y_R).
\end{align}
Here \(T^B_\lambda(u;g)\) is the fused transfer matrix in
\eqref{eq:fused-transfer-definition} for the \(GL(V_B)\) spin chain on
\(\mathcal H_L\); its auxiliary space carries the representation
\(V^{GL(V_B)}_\lambda\), and the twist is \(g\).  Similarly,
\(T^R_\lambda(w;\mathbf y_R)\) denotes the fused transfer matrix for the
\(GL(U_R)\) spin chain on \(\mathcal K_{L'}\); its auxiliary space carries
\(V^{GL(U_R)}_\lambda\), and the twist is the diagonal element of
\(GL(U_R)\) with eigenvalues \((y_\alpha)_{\alpha\in R}\).  Thus the Schur
function \(s_\lambda(\mathbf y_R)\) in \eqref{eq:Schur-expansion-main} is
replaced by a second fused transfer matrix.  

There are several elementary reductions.  If \(L'=0\), then
\(T^R_\lambda(w;\mathbf y_R)=s_\lambda(\mathbf y_R)\), and
\eqref{eq:finite-master-def} is recovered.  The same reduction is obtained
as \(w\to\infty\) after normalizing the second monodromy matrix so that it
tends to the identity.
If \(L=0\), the first monodromy matrix is absent and
\(T^B_\lambda(u;g)=s_\lambda(g_1,\ldots,g_M)\), giving the dual situation
in which the nontrivial spin chain is the \(GL(U_R)\) chain.  If both
\(L=0\) and \(L'=0\), the trace reduces to the Cauchy kernel
\eqref{eq:zero-Cauchy-kernel}--\eqref{eq:Cauchy}.

When the two quantum spaces \(\mathcal H_L\) and \(\mathcal K_{L'}\) are
kept separate, the Yang--Baxter relations for the two spin chains imply
commutativity of the family \(\mathcal T_{L,L'}(u,w;g,\mathbf y_R)\).  In
the expansion above, this is simply the commutativity of the fused transfer
matrices in each tensor component.  If the two monodromy matrices are
required to act on the same quantum space, the construction is more
restrictive.  At a minimum one has to identify \(U_R\simeq V_B\), so in
particular \(|R|=M\).  Even in that case, an argument based on 
\(RLL\)-relation applies only to one-parameter specializations, for example \(w=u+\delta\)
with \(\delta\in\mathbb C\), together with compatible choices such as
\(j'=j\) and \(\theta'_j=\theta_j\).


\section{mKP bilinear identity and the role of the CBR formula}
\label{sec:bilinear-context}

The tau-function property of \eqref{eq:full-master-paper} is known from fusion and from the CBR (Cherednik--Bazhanov--Reshetikhin), or quantum Jacobi--Trudi, determinant formulae for transfer matrices. We recall the statement in the present convention because it fixes notation; see \cite{Cherednik1989,BazhanovReshetikhin1990,
KazakovVieira2008,AKLTZ2013}.

Let
\begin{equation}
  T_\varnothing(u)=\prod_{\ell=1}^L(u-\theta_\ell)\,\Id_{\Hq_L}
\end{equation}
and define normalized transfer matrices
\begin{equation}
  \widehat T_\lambda(u)=(T_\varnothing(u))^{-1}T_\lambda(u).
\end{equation}
Set
\begin{equation}
  H_s(u)=\widehat T_{(s)}(u),
  \qquad H_0(u)=1,
  \qquad H_s(u)=0\quad(s<0).
\end{equation}
The CBR/Jacobi--Trudi formula in one common normalization reads
\begin{equation}
  \widehat T_\lambda(u)
  =\det_{1\le a,b\le d}
  \left(H_{\lambda_a-a+b}(u-b+1)\right),
  \label{eq:CBR-current-convention}
\end{equation}
for any $d\ge \ell(\lambda)$.  Different authors shift $u$ or transpose Young diagrams; these changes do not affect the interpretation of the normalized transfer matrices as Pl{\"u}cker coordinates.

The following identity is the mKP bilinear identity
\cite{DateJimboKashiwaraMiwa1983,JimboMiwa1983}, in the form used for
master $T$-operators in \cite{AKLTZ2013}.

\begin{theorem}[Master $T$-operator as an mKP tau-function \cite{AKLTZ2013}]
\label{thm:CBR-mKP-route}
For the rational inhomogeneous $\gl(M)$ spin chain with diagonal twist, the master $T$-operator \eqref{eq:full-master-paper} satisfies the mKP bilinear identity
\begin{align}
  \oint_C
  & z^{u-u'}e^{\xi(\wt-\wt',z)}
  \mathcal T(u,\wt-[z^{-1}])
  \mathcal T(u',\wt'+[z^{-1}])
  \,dz=0.
  \label{eq:mKP-bilinear-paper}
\end{align}
Here
\begin{equation}
  \xi(\wt,z)=\sum_{k\ge1}t_kz^k,
  \qquad
  [z^{-1}]=\left(z^{-1},\frac{z^{-2}}2,\frac{z^{-3}}3,\ldots\right).
\end{equation}
\end{theorem}

\begin{remark}
For the interpretation as a formal Laurent series in \(z\), we assume
\(u-u'\in\Z\).  Section~\ref{sec:QQ-relations} uses only \(u-u'=1\).
For non-integral values of \(u-u'\), one must instead choose a branch of
\(z^{u-u'}\) and specify the corresponding contour.
\end{remark}

\begin{remark}
The contour \(C\) is the contour \(C_{[0,\infty]}\) used in Refs.~\cite{AKLTZ2013,Zabrodin2014XXX}: it encircles the cut \([0,\infty]\), including the endpoints \(0\) and \(\infty\), and does not enclose singularities coming from the shifted master \(T\)-operators.  Below we use the corresponding coefficient form in the Laurent expansion in \(z^{-1}\).  Equivalently, if one denotes the added Miwa variable by \(y\), then \(y=z^{-1}\).  This contour in the \(z\)-plane should not be confused with the residues in the Miwa variables \(y_\beta\) at \(y_\beta=g_\beta^{-1}\).
\end{remark}

The role of the CBR/Jacobi--Trudi formula in this statement is that it
identifies the normalized transfer matrices \(\widehat T_\lambda(u)\) with
Pl{\"u}cker coordinates for the corresponding mKP tau-function.


\section{The limit \texorpdfstring{$g_\beta y_\beta\to1$}{g beta y beta -> 1} and the \texorpdfstring{$L$}{L}-operator}
\label{sec:local-contraction}

Proposition~\ref{prop:main-contraction} expresses the \(Q\)-operator
\(Q_I(u)\) as residues of the master \(T\)-operator. This section specifies
the variables, Fock spaces, and operators used in that formula. We also
state Lemma~\ref{lem:one-critical-trace} and
Lemma~\ref{lem:regular-trace-away-from-one}, which evaluate the traces
needed in its proof.

Subsection~\ref{subsec:degenerate-L-definition} defines the degenerate
Yangian \(L\)-operator used to construct \(Q_I(u)\). Its derivation from
the Schwinger \(L\)-operator, including the terms that do not contribute
to the residue, is given in \ref{app:local-matrix-details}--\ref{app:ordered-factorization}.
The proof of Proposition~\ref{prop:main-contraction} is completed in
\ref{app:full-residue-proof}.

Choose
\begin{equation}
 I\subset B,\qquad \Ibar=B\setminus I.
 \label{eq:I-Ibar-subsets}
\end{equation}
To obtain \(Q_I(u)\), set \(R=\Ibar\) in the definition
\eqref{eq:finite-master-def}. Thus, for each \(\beta\in\Ibar\),
we use the family of oscillators
\(\{\mathsf a_{c\beta},\mathsf a_{c\beta}^{\dagger}:c\in B\}\)
and its Miwa variable \(y_\beta\).

The same choice can be made within a master \(T\)-operator defined
with a larger set \(R\supset\Ibar\). Take the coefficient of
\(\prod_{\alpha\in R\setminus\Ibar}y_\alpha^0\) in the expansion
of \(\mathcal T(u;\mathbf y_R)\) in the additional Miwa variables.

The occupation number of every oscillator \(\mathsf a_{c\alpha}\),
with \(c\in B\) and \(\alpha\in R\setminus\Ibar\), must then be equal
to zero. Thus the trace over \(\F_{B,R\setminus\Ibar}\) is restricted
to its vacuum, which is annihilated by all these operators.
Each \(L^{(R)}(v)\) preserves the total occupation number in each
family of oscillators with a fixed flavour label \(\alpha\).
On the vacuum of \(\F_{B,R\setminus\Ibar}\), every bilinear
\(\mathsf a_{d\alpha}^{\dagger}\mathsf a_{c\alpha}\) with
\(\alpha\in R\setminus\Ibar\) acts as zero. Restricting
\(L^{(R)}(v)\) to \(\F_{B,\Ibar}\) tensored with this vacuum therefore
gives \(L^{(\Ibar)}(v)\). In particular,

\[
 \left.\mathcal T(u;\mathbf y_R)\right|_{y_\alpha=0\ (\alpha\in R\setminus\Ibar)}
 =\mathcal T(u;\mathbf y_{\Ibar}).
\]
\ref{app:nontrivial-glI} instead retains the additional
variables at nonzero values away from their poles and describes the
polynomial \(GL(V_I)\)-representations that then occur in the
\(L\)-operator.

We use the assumptions on the boundary twist stated in
\eqref{eq:diagonal-twist}: its eigenvalues \(g_1,\ldots,g_M\) are
nonzero and pairwise distinct.

\subsection{The Schwinger L-operator}

For the chosen set of flavour labels \(R=\Ibar\), we use the Fock space

\begin{equation}
  \F_{B,\Ibar}=\operatorname{Fock}(V_B\otimes U_{\Ibar})
  \label{eq:FMB-local}
\end{equation}
in the notation of \eqref{eq:FJS-notation-fixed}.  Its vacuum, occupation
basis, trace, and oscillator relations are those of
\eqref{eq:Fock-vacuum-fixed}--\eqref{eq:Fock-character-fixed} with
\(J=B\) and \(S=\Ibar\).
For \(c,d\in B\) put
\begin{equation}
  e_{cd}^{(\Ibar)}
  =\sum_{\beta\in\Ibar}\mathsf a_{c\beta}^{\dagger}\mathsf a_{d\beta}.
  \label{eq:eB-definition-local}
\end{equation}
For \(R=\Ibar\), the Schwinger \(L\)-operator of \eqref{eq:Schwinger-L}
can be written as
\begin{equation}
  L^{(\Ibar)}(v)
  =\sum_{c,d\in B}[L^{(\Ibar)}(v)]_{cd}\otimes E_{cd}
  \in \End(\F_{B,\Ibar})\otimes\End(\C^M),
  \label{eq:genuine-local-L}
\end{equation}
where the bracketed lower indices denote the coefficient of \(E_{cd}\) in this expansion:
\begin{equation}
  [L^{(\Ibar)}(v)]_{cd}
  =v\delta_{cd}\Id_{\F_{B,\Ibar}}+e_{dc}^{(\Ibar)}.
  \label{eq:local-L-component}
\end{equation}
The corresponding embedding into the \(\ell\)-th site of the spin chain is
\begin{equation}
  L^{(\Ibar)}_{0\ell}(v)
  =\sum_{c,d\in B}[L^{(\Ibar)}(v)]_{cd}\otimes E_{cd}^{(\ell)}.
  \label{eq:embedded-local-L}
\end{equation}
With the convention for lower indices fixed in \eqref{eq:operator-embedding-convention},
the relation \eqref{eq:Schwinger-RLL} becomes
\begin{equation}
  L^{(\Ibar)}_{12}(u)L^{(\Ibar)}_{13}(v)R_{23}(v-u)
  =R_{23}(v-u)L^{(\Ibar)}_{13}(v)L^{(\Ibar)}_{12}(u).
  \label{eq:RLL-local-unscaled}
\end{equation}

\subsection{The diagonal factors near \texorpdfstring{$g_\beta y_\beta=1$}{g beta y beta = 1}}

For the full flavour set $R$, the product of the operator induced by the boundary twist and the operator \(\widehat h(\mathbf y_R)\) in Definition~\ref{def:finite-master} is
\begin{equation}
  \widehat g\,\widehat h(\mathbf y_R)
  =\prod_{\alpha\in R}\prod_{c\in B}(g_c y_\alpha)^{\mathsf N_{c\alpha}},
  \qquad
  \mathsf N_{c\alpha}=\mathsf a_{c\alpha}^{\dagger}\mathsf a_{c\alpha}.
  \label{eq:full-twist-R}
\end{equation}

For the selected flavour variables \(y_\beta\), \(\beta\in\Ibar\),
the poles of the master \(T\)-operator under consideration occur at

\begin{equation}
  g_\beta y_\beta=1.
\end{equation}
We measure the distance from this point by
\begin{equation}
  \varepsilon_\beta=1-g_\beta y_\beta,
  \qquad
  y_\beta=\frac{1-\varepsilon_\beta}{g_\beta}.
  \label{eq:epsilon-beta-simple}
\end{equation}
For later use we write
\begin{equation}
  \mathbf y(\varepsilon)
  :=\{y_\beta(\varepsilon_\beta)\}_{\beta\in\Ibar},
  \qquad
  y_\beta(\varepsilon_\beta)=\frac{1-\varepsilon_\beta}{g_\beta}.
  \label{eq:y-of-epsilon-defined}
\end{equation}
Here the argument \(\varepsilon\) is only a shorthand for the family of
individual parameters \(\varepsilon_\beta\), \(\beta\in\Ibar\); limits are
always taken componentwise.

With \(R=\Ibar\), the operator at the boundary in
\eqref{eq:full-twist-R} becomes

\begin{align}
  W_{\Ibar}(\mathbf y_{\Ibar})
  &=\prod_{\beta\in\Ibar}\prod_{c\in B}
  (g_c y_\beta)^{\mathsf N_{c\beta}}
  \notag\\
  &=\prod_{\beta\in\Ibar}\prod_{c\in B}
  \left[\frac{g_c}{g_\beta}(1-\varepsilon_\beta)\right]^{\mathsf N_{c\beta}}.
  \label{eq:Wcrit-simple}
\end{align}
In Eq.~\eqref{eq:Wcrit-simple}, the base
\[
  \frac{g_c}{g_\beta}(1-\varepsilon_\beta)
\]
 tends to $1$ as  $\varepsilon_\beta \to 0$ 
 only when $c=\beta$, 

 since the eigenvalues \(g_1,\ldots,g_M\) of the boundary twist are
distinct. Among the traces of the individual factors in
\eqref{eq:Wcrit-simple}, only the trace for \(c=\beta\) develops a pole
as \(\varepsilon_\beta\to0\). We denote the corresponding oscillator
generators by

\begin{equation}
  \mathsf A_\beta:=\mathsf a_{\beta\beta},
  \qquad
  \mathsf A_\beta^{\dagger}:=\mathsf a_{\beta\beta}^{\dagger}.
  \label{eq:critical-oscillator-identified}
\end{equation}
Indeed,
\begin{equation}
  \Tr(1-\varepsilon_\beta)^{\mathsf N_{\beta\beta}}
  =\sum_{n\ge0}(1-\varepsilon_\beta)^n
  =\frac1{\varepsilon_\beta},
  \label{eq:critical-one-over-eps}
\end{equation}
whereas, for $c\ne\beta$,
\begin{equation}
  \Tr\left[\left(\frac{g_c}{g_\beta}(1-\varepsilon_\beta)\right)^{\mathsf N_{c\beta}}\right]
  \longrightarrow
  \frac{1}{1-g_c/g_\beta}.
  \label{eq:noncritical-trace-general}
\end{equation}
The limit in \eqref{eq:noncritical-trace-general} is finite because
\(g_c\ne g_\beta\).

We decompose the factor \eqref{eq:Wcrit-simple} as
\begin{equation}
  W_{\Ibar}(\mathbf y_{\Ibar})
  =W_{\mathrm{sing}}(\varepsilon)
   W_{I,\Ibar}(\varepsilon)
   W_{\Ibar,\Ibar}^{\mathrm{off}}(\varepsilon),
  \label{eq:WB-decomposition}
\end{equation}
where
\begin{align}
  W_{\mathrm{sing}}(\varepsilon)
  &=\prod_{\beta\in\Ibar}(1-\varepsilon_\beta)^{\mathsf N_{\beta\beta}},
  \label{eq:Wcrit-eps}
  \\
  W_{I,\Ibar}(\varepsilon)
  &=\prod_{a\in I}\prod_{\beta\in\Ibar}
  \left[\frac{g_a}{g_\beta}(1-\varepsilon_\beta)\right]^{\mathsf N_{a\beta}},
  \label{eq:WA-B-eps}
  \\
  W_{\Ibar,\Ibar}^{\mathrm{off}}(\varepsilon)
  &=\prod_{\beta\in\Ibar}
    \prod_{\substack{\gamma\in\Ibar\\ \gamma\ne\beta}}
  \left[\frac{g_\gamma}{g_\beta}(1-\varepsilon_\beta)\right]^{\mathsf N_{\gamma\beta}}.
  \label{eq:WBB-off}
\end{align}
Setting all \(\varepsilon_\beta=0\) in \eqref{eq:WA-B-eps} gives the diagonal operator
\begin{equation}
  D_{I,\Ibar}
  :=\prod_{a\in I}\prod_{\beta\in\Ibar}
  \left(\frac{g_a}{g_\beta}\right)^{\mathsf N_{a\beta}}.
  \label{eq:DAB-derived}
\end{equation}
Let \(\F_{\mathrm{int}}\) be the Fock representation, in the sense of
Section~\ref{sec:conventions}, generated by the oscillator generators
\(\mathsf a_{\gamma\beta},\mathsf a_{\gamma\beta}^{\dagger}\), where
\(\gamma,\beta\in\Ibar\) and \(\gamma\ne\beta\). 
Define
\begin{align}
  C_{\Ibar}
  &:={\lim}_{\varepsilon_\beta\to0\;(\beta\in\Ibar)}
  \Tr_{\F_{\mathrm{int}}}
  W_{\Ibar,\Ibar}^{\mathrm{off}}(\varepsilon)
  \notag\\
  &= {\lim}_{\varepsilon_\beta\to0\;(\beta\in\Ibar)}
  \prod_{\beta\in\Ibar}
  \prod_{\substack{\gamma\in\Ibar\\ \gamma\ne\beta}}
  \frac{1}{1-\dfrac{g_\gamma}{g_\beta}(1-\varepsilon_\beta)}
  \notag\\
  &=\prod_{\beta\in\Ibar}
  \prod_{\substack{\gamma\in\Ibar\\ \gamma\ne\beta}}
  \frac{1}{1-g_\gamma/g_\beta}.
  \label{eq:CIbar-def}
\end{align}
The trace is first evaluated in a region where the geometric series in the
second line converge.  The resulting rational function is then continued
meromorphically to \(\varepsilon_\beta=0\) for all \(\beta\in\Ibar\).
When \(|\Ibar|\ge2\), one must not set all \(\varepsilon_\beta\) equal to
zero before evaluating the trace, because both \(g_\gamma/g_\beta\) and its
reciprocal occur.  The same prescription is understood whenever
\(C_{\Ibar}\) is obtained from this trace later in the paper.

\subsection[The scaled L-operator and master T-operator]{The scaled
\texorpdfstring{$L$}{L}-operator and master \texorpdfstring{$T$}{T}-operator}
\label{subsec:scaled-master-T}

Define a diagonal matrix on the vector space $\C^M$ whose basis is labelled by $B$ by
\begin{equation}
  S_{\varepsilon}
  =\sum_{a\in I}E_{aa}+
    \sum_{\beta\in\Ibar}\sqrt{\varepsilon_\beta}\,E_{\beta\beta}.
  \label{eq:S-epsilon-scaling}
\end{equation}
It is convenient to write $S_{\varepsilon}=\sum_{c\in B}s_cE_{cc}$,
where
\begin{equation}
  s_c=1\quad(c\in I),
  \qquad
  s_\beta=\sqrt{\varepsilon_\beta}\quad(\beta\in\Ibar).
  \label{eq:scaling-sc}
\end{equation}
The scaled \(L\)-operator is
\begin{equation}
  \widetilde L^{(\Ibar)}(v)
  =\bigl(\Id_{\F_{B,\Ibar}}\otimes S_{\varepsilon}\bigr)
  L^{(\Ibar)}(v)
  \bigl(\Id_{\F_{B,\Ibar}}\otimes S_{\varepsilon}\bigr),
  \label{eq:scaled-local-L}
\end{equation}
so that the coefficient of \(E_{cd}\) in the scaled \(L\)-operator is
\begin{equation}
  [\widetilde L^{(\Ibar)}(v)]_{cd}
  =s_cs_d\left(v\delta_{cd}+
  \sum_{\beta\in\Ibar}\mathsf a_{d\beta}^{\dagger}\mathsf a_{c\beta}\right).
  \label{eq:scaled-component-detailed}
\end{equation}
One can show
\begin{equation}
  R(u)(S_{\varepsilon} \otimes S_{\varepsilon})=
  (S_{\varepsilon} \otimes S_{\varepsilon})R(u).
  \label{eq:R-commutes-S2S3}
\end{equation}
Thus the scaled \(L\)-operator satisfies the same RLL relation as the
unscaled one:
\begin{align}
\widetilde L^{(\Ibar)}_{12}(u)
\widetilde L^{(\Ibar)}_{13}(v)R_{23}(v-u)
=R_{23}(v-u)\widetilde L^{(\Ibar)}_{13}(v)\widetilde L^{(\Ibar)}_{12}(u).
\end{align}

Specializing Definition~\ref{def:finite-master} to
\(R=\Ibar\) and writing
\(\mathbf y_{\Ibar}=(y_\beta)_{\beta\in\Ibar}\), we obtain
\begin{equation}
  \mathcal T(u;\mathbf y_{\Ibar})
  =\Tr_{\F_{B,\Ibar}}
  \left[
  L^{(\Ibar)}_{0L}(u-\theta_L)\cdots
  L^{(\Ibar)}_{01}(u-\theta_1)W_{\Ibar}(\mathbf y_{\Ibar})
  \right].
  \label{eq:TB-master-detailed}
\end{equation}
Let \(\widetilde L^{(\Ibar)}_{0\ell}(v)\) denote the embedding of
\eqref{eq:scaled-local-L} into the \(\ell\)-th quantum space.  Then
\begin{align}
  \widetilde{\mathcal T}(u;\varepsilon)
  &:={S}_{\varepsilon}^{\otimes L}
  \mathcal T(u;\mathbf y(\varepsilon))
  {S}_{\varepsilon}^{\otimes L}
  \notag\\
  &=\Tr_{\F_{B,\Ibar}}
  \left[
  \widetilde L^{(\Ibar)}_{0L}(u-\theta_L)\cdots
  \widetilde L^{(\Ibar)}_{01}(u-\theta_1)
  W_{\Ibar}(\mathbf y(\varepsilon))
  \right].
  \label{eq:scaled-master-T}
\end{align}
For \(\mathbf c,\mathbf d\in B^L\), expand the matrix element of the
product of the scaled \(L\)-operators in Eq.~\eqref{eq:scaled-master-T}.
The generators \(\mathsf a_{c\beta},\mathsf a_{c\beta}^{\dagger}\) commute
with \(\mathsf a_{d\gamma},\mathsf a_{d\gamma}^{\dagger}\) whenever
\((c,\beta)\ne(d,\gamma)\), so the trace of each term factorizes.  Thus,
schematically, the matrix element has the form
\begin{equation}
  [\widetilde{\mathcal T}(u;\varepsilon)]_{\mathbf c,\mathbf d}
  =\sum_{\sigma} C_{\sigma}(u,\varepsilon)
  \prod_{\beta\in\Ibar}\prod_{c\in B}
  \Tr_{\F_{\{c\},\{\beta\}}}
  \left[
  U_{c\beta}^{(\sigma)}
  \left(\frac{g_c}{g_\beta}(1-\varepsilon_\beta)\right)^{\mathsf N_{c\beta}}
  \right].
  \label{eq:scaled-master-trace-factorization}
\end{equation}
Here the sum is finite, \(C_{\sigma}(u,\varepsilon)\) is a scalar
independent of the oscillator generators, and
\(U_{c\beta}^{(\sigma)}\) is a finite product of
\(\mathsf a_{c\beta}\) and \(\mathsf a_{c\beta}^{\dagger}\).

\begin{lemma}[Trace as \(\varepsilon\to0\)]

\label{lem:one-critical-trace}
Let $\mathsf A,\mathsf A^{\dagger}$ be a single bosonic oscillator,
$[\mathsf A,\mathsf A^{\dagger}]=1$, and let
$\mathsf N=\mathsf A^{\dagger}\mathsf A$.  Let $W$ be any finite product of
\[
  \sqrt\varepsilon \mathsf A^{\dagger},
  \qquad
  \sqrt\varepsilon \mathsf A.
\]
Suppose that this product contains $r$ copies of
$\sqrt\varepsilon \mathsf A^{\dagger}$ and $s$ copies of
$\sqrt\varepsilon \mathsf A$.  Then
\begin{align}
  \lim_{\varepsilon\to0}
  \varepsilon\Tr\left[(1-\varepsilon)^{\mathsf N} W\right]
  &=
  \lim_{\varepsilon\to0}
  \varepsilon\Tr\left[W(1-\varepsilon)^{\mathsf N}\right]
  \notag\\
  &=
  \begin{cases}
    r!,& r=s,\\
    0,& r\ne s.
  \end{cases}
  \label{eq:one-critical-word-limit}
\end{align}
\end{lemma}

\begin{lemma}[Regularity of the trace at \(\varepsilon=0\)]

\label{lem:regular-trace-away-from-one}
Let
\[
  \rho(\varepsilon)=\lambda(1-\varepsilon),
  \qquad \lambda\ne1.
\]
For any finite product \(U\) of \(\mathsf A\) and
\(\mathsf A^{\dagger}\), both
\[
  \Tr\left[\rho(\varepsilon)^{\mathsf N}U\right]
  \qquad\text{and}\qquad
  \Tr\left[U\rho(\varepsilon)^{\mathsf N}\right]
\]
are regular at \(\varepsilon=0\), with the traces understood by meromorphic
continuation or as formal traces in the occupation-number basis.
\end{lemma}

The proofs of Lemma~\ref{lem:one-critical-trace} and
Lemma~\ref{lem:regular-trace-away-from-one} are given in
\ref{app:direct-trace-asymptotics}. Their application to the product of
Schwinger \(L\)-operators is carried out in \ref{app:residue-details}.


\subsection{The degenerate Yangian \texorpdfstring{$L$}{L}-operator}
\label{subsec:degenerate-L-definition}

We specify the \(L\)-operator that enters the definition of \(Q_I(u)\)
in Section~\ref{sec:Q-contraction-theorem}. Its auxiliary space is
\(\F_{I,\Ibar}\), generated by
\(\mathsf a_{a\beta},\mathsf a_{a\beta}^{\dagger}\) with
\(a\in I\) and \(\beta\in\Ibar\), in the notation of
\eqref{eq:FJS-notation-fixed}. Define
\[
 \mathscr L_{I,\Ibar}(v)\in\End(\F_{I,\Ibar})\otimes\End(\C^M)
\]
by the following matrix elements on the quantum space \(\C^M\):

\begin{align}
  [\mathscr L_{I,\Ibar}(v)]_{ac}&=v\delta_{ac}+
  \sum_{\beta\in\Ibar}\mathsf a_{c\beta}^{\dagger}\mathsf a_{a\beta},
  &&a,c\in I,
  \label{eq:contracted-L-ac}
  \\
  [\mathscr L_{I,\Ibar}(v)]_{a\beta}&=\mathsf a_{a\beta},
  &&a\in I,
  \beta\in\Ibar,
  \label{eq:contracted-L-a-beta}
  \\
  [\mathscr L_{I,\Ibar}(v)]_{\beta a}&=\mathsf a_{a\beta}^{\dagger},
  &&\beta\in\Ibar,
  a\in I,
  \label{eq:contracted-L-beta-a}
  \\
  [\mathscr L_{I,\Ibar}(v)]_{\beta\gamma}&=\delta_{\beta\gamma},
  &&\beta,
  \gamma\in\Ibar.
  \label{eq:contracted-L-components}
\end{align}

These four formulas can also be written as

\begin{equation}
  \mathscr L_{I,\Ibar}(v)
  =
  \begin{pmatrix}
    \left(v\delta_{ac}+\displaystyle\sum_{\beta\in\Ibar}
    \mathsf a_{c\beta}^{\dagger}\mathsf a_{a\beta}\right)_{a,c\in I}
    & (\mathsf a_{a\beta})_{a\in I,\beta\in\Ibar}\\[1mm]
    (\mathsf a_{a\beta}^{\dagger})_{\beta\in\Ibar,a\in I}
    & (\delta_{\beta\gamma})_{\beta,\gamma\in\Ibar}
  \end{pmatrix}.
  \label{eq:contracted-L}
\end{equation}

In \eqref{eq:contracted-L}, the basis vectors \(e_a\) with
\(a\in I\) precede the vectors \(e_\beta\) with \(\beta\in\Ibar\).
The coefficient of \(v\) is \(\sum_{a\in I}E_{aa}\), whereas the
coefficient of \(v\) in \(L^{(\Ibar)}(v)\) is the identity on
\(\C^M\). The operator \(\mathscr L_{I,\Ibar}(v)\) agrees with the
degenerate Yangian \(L\)-operator of \cite{BFLMS2011}, after the changes
of convention explained in \ref{app:BFLMS-local-comparison}.

The calculation in \ref{app:residue-details} shows how the matrix
elements of \(\mathscr L_{I,\Ibar}\) occur in the product of Schwinger
\(L\)-operators when the terms contributing to the residues in
\eqref{eq:main-residue-formula} are evaluated.
Section~\ref{sec:HP-fixed-degree} derives the same \(L\)-operator by
restricting the Schwinger representation to fixed total occupation
numbers and then taking those numbers to infinity.

\section{The Q-operators \texorpdfstring{$Q_I$}{QI} from the contracted \texorpdfstring{$L$}{L}-operator}
\label{sec:Q-contraction-theorem}

We define \(Q_I(u)\) by taking the trace of a product of the
\(L\)-operators \(\mathscr L_{I,\Ibar}\) specified in
Subsection~\ref{subsec:degenerate-L-definition}. After defining the
number operators \(\mathsf M_i\) on the quantum space and the normalization
of the limit, we state the residue formula in
Proposition~\ref{prop:main-contraction}. The proof is contained in
\ref{app:residue-details}.
The normalization of \(Q_I(u)\) is compared with that of
\cite{BFLMS2011} in \ref{app:BFLMS-local-comparison}; the
residue construction follows the conventions of
\cite{AKLTZ2013,KLT2012}.
The construction with additional flavour labels is treated in
\ref{app:nontrivial-glI}, where the additional oscillators
with colour labels in \(I\) provide polynomial representations of
\(GL(V_I)\) in the \(L\)-operator.

\subsection{Definition of \texorpdfstring{$Q_I$}{QI}}

We use the diagonal operator \(D_{I,\Ibar}\) in
\eqref{eq:DAB-derived} as the boundary twist in the trace over
\(\F_{I,\Ibar}\).

\begin{definition}
The $Q$-operator attached to $I$ (cf. \cite{BFLMS2011}) 
is defined by
\begin{equation}
  Q_I(u)
  =\Tr_{\F_{I,\Ibar}}
  \left[
  \mathscr L_{I,\Ibar,L}(u-\theta_L)\cdots
  \mathscr L_{I,\Ibar,1}(u-\theta_1)
  D_{I,\Ibar}
  \right].
  \label{eq:Sch-Q-def}
\end{equation}

The operator \(\mathscr L_{I,\Ibar,\ell}(v)\) is obtained from the
matrix elements in
\eqref{eq:contracted-L-ac}--\eqref{eq:contracted-L-components}
by replacing \(E_{cd}\) with \(E_{cd}^{(\ell)}\), according to
\eqref{eq:embedded-local-L}. It acts on \(\F_{I,\Ibar}\) and the
\(\ell\)-th quantum space, and acts as the identity on the other quantum
spaces.

\end{definition}

Evaluating the trace in Eq.~\eqref{eq:Sch-Q-def} in the
occupation-number basis shows that it converges when
\(|g_a/g_\beta|<1\) for all \(a\in I\) and \(\beta\in\Ibar\).  Outside
this domain, each matrix element is defined by meromorphic continuation or
by the corresponding formal power series in the variables \(g_a/g_\beta\).

The oscillators $\mathsf a_{\gamma\beta}$ with
$\gamma,\beta\in\Ibar$ and $\gamma\ne\beta$ do not appear in
\eqref{eq:contracted-L}; tracing over them in the proof of
Proposition~\ref{prop:main-contraction} gives the factor $C_{\Ibar}$ defined
in \eqref{eq:CIbar-def}.

\subsection{Number operators on the quantum space}

For each colour label $i\in B$, define
\begin{equation}
  \mathsf M_i=\sum_{\ell=1}^L E_{ii}^{(\ell)}.
  \label{eq:Mbeta-count}
\end{equation}

On a basis vector \(e_{d_1}\otimes\cdots\otimes e_{d_L}\),
\(\mathsf M_i\) acts by the number of indices \(\ell\) for which
\(d_\ell=i\). The commutation relations needed for the residue formula
are
\begin{equation}
  [Q_I(u),\mathsf M_\beta]=0,
  \qquad \beta\in\Ibar.
  \label{eq:Q-commutes-Mbeta}
\end{equation}
\begin{equation}
  [\mathcal T(u;\mathbf y_{\Ibar}),\mathsf M_i]=0,
  \qquad i\in B.
  \label{eq:commute-Mbeta-detailed}
\end{equation}
It follows that \(Q_I(u)\) and \(\mathcal T(u;\mathbf y_{\Ibar})\)
preserve each simultaneous eigenspace of \(\mathsf M_\beta\),
\(\beta\in\Ibar\). On the eigenspace with eigenvalues
\(m_\beta\), the operator \((1-g_\beta y_\beta)^{\mathsf M_\beta}\)
acts by multiplication by \((1-g_\beta y_\beta)^{m_\beta}\).
The calculation of the commutators and their consequences for matrix
elements is given in \ref{app:colour-conservation}.

\subsection{Highest-pole coefficient}

Define
\begin{equation}
  \operatorname{HP}_{\Ibar}\mathcal T(u;\mathbf y_{\Ibar})
  =\lim_{\varepsilon_\beta\to0\;(\beta\in\Ibar)}
  \left(\prod_{\beta\in\Ibar}\varepsilon_\beta\right)
  S_{\varepsilon}^{\otimes L}
  \mathcal T(u;\mathbf y(\varepsilon))
  S_{\varepsilon}^{\otimes L}.
  \label{eq:HP-definition}
\end{equation}
Equation~\eqref{eq:commute-Mbeta-detailed} gives the equivalent
expression
\begin{equation}
  \operatorname{HP}_{\Ibar}\mathcal T(u;\mathbf y_{\Ibar})
  =\lim_{\varepsilon_\beta\to0\;(\beta\in\Ibar)}
  \left(
  \prod_{\beta\in\Ibar}\varepsilon_\beta^{\mathsf M_\beta+1}
  \right)
  \mathcal T(u;\mathbf y(\varepsilon)).
  \label{eq:HP-equivalent}
\end{equation}

The spectrum of \(\mathsf M_\beta\) on \(\Hq_L\) is
\(\{0,1,\ldots,L\}\).  We therefore define
\(\Gamma(\mathsf M_\beta+1)\) by its action on the eigenspaces:
\[
  \Gamma(\mathsf M_\beta+1)\big|_{\mathsf M_\beta=m}=m!.
\]
It is invertible on \(\Hq_L\), with inverse acting by \((m!)^{-1}\) on
the same eigenspace.

\begin{proposition}[Residue formula proved by a direct trace calculation in the occupation basis]
\label{prop:main-contraction}
For distinct eigenvalues $g_1,\ldots,g_M$ of the diagonal matrix $g$,
\begin{equation}
  \operatorname{HP}_{\Ibar}\mathcal T(u;\mathbf y_{\Ibar})
  =C_{\Ibar}\left(\prod_{\beta\in\Ibar}\Gamma(\mathsf M_\beta+1)\right)
  Q_I(u).
  \label{eq:main-contraction-formula}
\end{equation}
Equivalently, using the original variables \(y_\beta\),
\begin{align}
  Q_I(u)
  &=C_{\Ibar}^{-1}
  (-1)^{|\Ibar|}
  \left(\prod_{\beta\in\Ibar}g_\beta\right)
  \left(\prod_{\beta\in\Ibar}
  \Gamma(\mathsf M_\beta+1)^{-1}\right)
  \notag\\
  &\quad\times
  {\operatorname*{Res}}_{y_\beta=g_\beta^{-1},\,\beta\in\Ibar}
  \left[
  \prod_{\beta\in\Ibar}(1-g_\beta y_\beta)^{\mathsf M_\beta}
  \mathcal T(u;\mathbf y_{\Ibar})
  \prod_{\beta\in\Ibar}dy_\beta
  \right].
  \label{eq:main-residue-formula}
\end{align}
\end{proposition}

\begin{remark}
For a complex number $a$, $\operatorname*{Res}_{z=a}f(z)\,dz$ denotes
the coefficient of $(z-a)^{-1}$ in the Laurent expansion of $f(z)$.  In
\eqref{eq:main-residue-formula} the residue is the multiple residue in the
variables \(y_\beta\), \(\beta\in\Ibar\).  If
\(\Ibar=\{\beta_1,\ldots,\beta_r\}\), it means the iterated residue
\[
  \operatorname*{Res}_{y_{\beta_1}=g_{\beta_1}^{-1}}
  \cdots
  \operatorname*{Res}_{y_{\beta_r}=g_{\beta_r}^{-1}},
\]
or equivalently the coefficient of
\(\prod_{\beta\in\Ibar}(y_\beta-g_\beta^{-1})^{-1}\) in the Laurent expansion.
The order is irrelevant here because the variables \(y_\beta\) are independent.
The residue is operator-valued.  If \(\Hq_L[\mathbf m]\) denotes the
subspace on which \(\mathsf M_\beta=m_\beta\) for all
\(\beta\in\Ibar\), then
\((1-g_\beta y_\beta)^{\mathsf M_\beta}\) acts on
\(\Hq_L[\mathbf m]\) as the scalar
\((1-g_\beta y_\beta)^{m_\beta}\), and the residue is the corresponding
Laurent coefficient in \(\End(\Hq_L[\mathbf m])\).
\end{remark}

We also remark that 
some special cases of \eqref{eq:main-contraction-formula} are checked explicitly in \ref{app:short-chain-checks}.

The proof of Proposition~\ref{prop:main-contraction} is given in
\ref{app:residue-details}. The formulas above identify the selected
residues of \(\mathcal T\) with the operator \(Q_I(u)\) defined in
\eqref{eq:Sch-Q-def}, including the scalar \(C_{\Ibar}\) and the
operators \(\Gamma(\mathsf M_\beta+1)\).

\section[Schwinger realization and fixed occupation numbers]{From the Schwinger realization to the degenerate Yangian \texorpdfstring{$L$}{L}-operator by fixing occupation numbers}
\label{sec:HP-fixed-degree}

This section proves a statement which is related to, but different from, the residue formula of Proposition~\ref{prop:main-contraction}.  The residue formula concerns the trace over the full Fock representation and the limit $g_\beta y_\beta\to1$ for the selected flavour labels.  In contrast, the present section fixes the total occupation number for each such flavour label, sends these occupation numbers to infinity, and obtains the same \(L\)-operator \eqref{eq:contracted-L} from the Schwinger realization.  This gives the large-occupation-number limit of the $L$-operator written with Schwinger bosons.

For $M=2$, this limit is closely related to the Holstein--Primakoff
large-spin limit used in \cite{ShortcutQ}: after identifying the fixed
particle number $m$ with $2j$, it gives the same oscillator $L$-operator up
to elementary changes of convention.  In contrast to Ref.~\cite{ShortcutQ}, where the Holstein--Primakoff realization is used from the outset, here it is
derived from subspaces with fixed total particle numbers, and the construction
applies to arbitrary $M$ and $I\subset B$.

\subsection[Fixed particle-number sectors]{Sectors with fixed particle numbers and occupation-number coordinates}
\label{subsec:fixed-particle-coordinates}

For \(\beta\in\Ibar\), the total number of particles with
flavour label \(\beta\) is the operator \(\mathsf N_\beta\) defined in
\eqref{eq:flavour-number-operator}.  We fix a family of non-negative integers
\begin{equation}
  \mathbf m=(m_\beta)_{\beta\in\Ibar}.
\end{equation}
Whenever a limit in these integers is used, the notation
\begin{equation}
  m_\beta\to\infty\quad(\beta\in\Ibar)
  \label{eq:multi-m-limit-meaning}
\end{equation}
means that \(\min_{\beta\in\Ibar}m_\beta\to\infty\).  No ratio
\(m_\beta/m_\gamma\) is kept fixed.
The subspace on which the total number of particles with flavour label \(\beta\) is equal to \(m_\beta\), for each \(\beta\in\Ibar\), is
\begin{equation}
  \F_{\mathbf m}
  =
  \bigcap_{\beta\in\Ibar}
  \ker(\mathsf N_\beta-m_\beta)
  \subset \F_{B,\Ibar}.
  \label{eq:fixed-degree-sector}
\end{equation}
Equivalently,
\begin{equation}
  \F_{\mathbf m}
  \simeq
  \bigotimes_{\beta\in\Ibar}\Sym^{m_\beta}(V_B).
\end{equation}
This is a finite-dimensional subspace of the Fock representation generated by the Schwinger bosons.

We now describe this subspace by independent occupation numbers.  For a fixed \(\beta\in\Ibar\), the occupation number of the oscillator \(\mathsf a_{\beta\beta}\) is not independent once the total particle number \(m_\beta\) has been fixed.  Define
\begin{equation}
  \mathsf N^{\mathrm{off}}_\beta
  =
  \sum_{c\in B\setminus\{\beta\}}\mathsf a_{c\beta}^{\dagger}\mathsf a_{c\beta}.
  \label{eq:Noff-beta-HP}
\end{equation}
On \(\F_{\mathbf m}\), the identity \(\mathsf N_\beta=\mathsf N_{\beta\beta}+\mathsf N^{\mathrm{off}}_\beta\) gives
\begin{equation}
  \mathsf a_{\beta\beta}^{\dagger}\mathsf a_{\beta\beta}
  =
  m_\beta-\mathsf N^{\mathrm{off}}_\beta.
  \label{eq:abb-number-HP}
\end{equation}
Thus the independent occupation numbers are the numbers \(n_{c\beta}\) with \(c\ne\beta\).

Let \(\mathbf n=(n_{c\beta})\) denote a collection of non-negative integers indexed by pairs \((c,\beta)\) with \(\beta\in\Ibar\) and \(c\in B\setminus\{\beta\}\).  We impose the inequalities
\begin{equation}
  \sum_{c\in B\setminus\{\beta\}}n_{c\beta}\le m_\beta,
  \qquad \beta\in\Ibar.
  \label{eq:off-truncation-condition}
\end{equation}
For such \(\mathbf n\), put
\begin{equation}
  N^{\mathrm{off}}_\beta(\mathbf n)
  =
  \sum_{c\in B\setminus\{\beta\}}n_{c\beta}.
  \label{eq:Noff-on-basis}
\end{equation}
The symbol \(N^{\mathrm{off}}_\beta(\mathbf n)\) denotes an integer.  It is the eigenvalue of the number operator \(\mathsf N^{\mathrm{off}}_\beta\) on the basis vector labelled by \(\mathbf n\).  Thus \(m_\beta-N^{\mathrm{off}}_\beta(\mathbf n)\) is a non-negative integer by \eqref{eq:off-truncation-condition}.

Let \(|0\rangle_{\mathrm{off}}\) be the vector annihilated by all \(\mathsf a_{c\beta}\) with \(\beta\in\Ibar\) and \(c\ne\beta\).  For any \(\mathbf n\) satisfying \eqref{eq:off-truncation-condition}, define
\begin{equation}
  |\mathbf n\rangle_{\mathrm{off}}
  =
  \prod_{\beta\in\Ibar}
  \prod_{c\in B\setminus\{\beta\}}
  \frac{(\mathsf a_{c\beta}^{\dagger})^{n_{c\beta}}}{\sqrt{n_{c\beta}!}}
  |0\rangle_{\mathrm{off}}.
  \label{eq:off-basis-HP}
\end{equation}
The vectors \eqref{eq:off-basis-HP} form an orthonormal basis of their span.  We define
\begin{equation}
  \F_{\mathbf m}^{\mathrm{off}}
  =
  \operatorname{span}
  \left\{
  |\mathbf n\rangle_{\mathrm{off}}
  \;\middle|\;
  \mathbf n\text{ satisfies }\eqref{eq:off-truncation-condition}
  \right\}.
  \label{eq:Foff-definition-HP}
\end{equation}
The superscript \(\mathrm{off}\) indicates only that the occupation number of \(\mathsf a_{\beta\beta}\) has not been included among the independent variables.

Next define a vector in the fixed-particle-number subspace \(\F_{\mathbf m}\) by putting the dependent occupation number of \(\mathsf a_{\beta\beta}\) equal to \(m_\beta-N^{\mathrm{off}}_\beta(\mathbf n)\):
\begin{align}
  |\mathbf n;\mathbf m-\mathbf n\rangle
  &:=
  \prod_{\beta\in\Ibar}
  \frac{(\mathsf a_{\beta\beta}^{\dagger})^{m_\beta-N^{\mathrm{off}}_\beta(\mathbf n)}}
       {\sqrt{(m_\beta-N^{\mathrm{off}}_\beta(\mathbf n))!}}
  \notag\\
  &\quad\times
  \prod_{\beta\in\Ibar}
  \prod_{c\in B\setminus\{\beta\}}
  \frac{(\mathsf a_{c\beta}^{\dagger})^{n_{c\beta}}}{\sqrt{n_{c\beta}!}}
  |0\rangle_{B,\Ibar}.
  \label{eq:target-basis-HP}
\end{align}
The vectors \eqref{eq:target-basis-HP}, as \(\mathbf n\) runs over the solutions of \eqref{eq:off-truncation-condition}, form the occupation-number basis of \(\F_{\mathbf m}\).  Indeed, the fixed-particle-number condition is precisely
\begin{equation}
  n_{\beta\beta}+\sum_{c\in B\setminus\{\beta\}}n_{c\beta}=m_\beta,
  \qquad \beta\in\Ibar,
\end{equation}
so \(n_{\beta\beta}\) is uniquely determined by the numbers \(n_{c\beta}\), \(c\ne\beta\), and each basis vector of \(\F_{\mathbf m}\) is obtained in this way.

For one fixed \(\beta\), the number of choices of the integers \(n_{c\beta}\), \(c\ne\beta\), satisfying \eqref{eq:off-truncation-condition} is
\begin{equation}
  \#\left\{
  (n_{c\beta})_{c\ne\beta}
  \;\middle|\;
  n_{c\beta}\ge0,
  \sum_{c\ne\beta}n_{c\beta}\le m_\beta
  \right\}
  =
  \binom{m_\beta+|B|-1}{|B|-1}.
  \label{eq:dimension-one-beta-HP}
\end{equation}
The equality in \eqref{eq:dimension-one-beta-HP} follows by adding
\begin{equation}
  n_{\beta\beta}=m_\beta-
  \sum_{c\ne\beta}n_{c\beta},
\end{equation}
which gives a bijection with the non-negative solutions of \(\sum_{c\in B}n_{c\beta}=m_\beta\).  Hence
\begin{equation}
  \dim\F_{\mathbf m}^{\mathrm{off}}
  =
  \dim\F_{\mathbf m}
  =
  \prod_{\beta\in\Ibar}
  \binom{m_\beta+|B|-1}{|B|-1}.
  \label{eq:dimension-Foff-Fm-HP}
\end{equation}

Since the vector spaces \(\F_{\mathbf m}^{\mathrm{off}}\) and \(\F_{\mathbf m}\) have the same dimension by \eqref{eq:dimension-Foff-Fm-HP}, a linear isomorphism between them exists.  We use the simplest one: it sends the orthonormal basis vector \eqref{eq:off-basis-HP} to the corresponding orthonormal basis vector \eqref{eq:target-basis-HP}.  Thus we define
\begin{equation}
  \Upsilon_{\mathbf m}:\F_{\mathbf m}^{\mathrm{off}}\longrightarrow \F_{\mathbf m}
  \label{eq:Upsilon-def-HP}
\end{equation}
by
\begin{equation}
  \Upsilon_{\mathbf m}|\mathbf n\rangle_{\mathrm{off}}
  =
  |\mathbf n;\mathbf m-\mathbf n\rangle.
  \label{eq:Upsilon-on-basis-HP}
\end{equation}
 Its inverse is determined on the basis \eqref{eq:target-basis-HP} by
\begin{equation}
  \Upsilon_{\mathbf m}^{-1}|\mathbf n;\mathbf m-\mathbf n\rangle
  =
  |\mathbf n\rangle_{\mathrm{off}}.
  \label{eq:Upsilon-inverse-on-basis-HP}
\end{equation}

For later use, let \(\mathbf e_{c\beta}\) denote the multi-index which has value one at the pair \((c,\beta)\) and value zero at all other pairs.  Thus \(\mathbf n+\mathbf e_{c\beta}\) is obtained from \(\mathbf n\) by increasing \(n_{c\beta}\) by one, and \(\mathbf n-\mathbf e_{c\beta}\) is defined only when \(n_{c\beta}\ge1\).

We shall use the following identities, which are checked directly on the basis \eqref{eq:off-basis-HP}.  If \(c,d\ne\beta\), then
\begin{equation}
  \Upsilon_{\mathbf m}^{-1}\mathsf a_{c\beta}^{\dagger}\mathsf a_{d\beta}\Upsilon_{\mathbf m}
  =
  \mathsf a_{c\beta}^{\dagger}\mathsf a_{d\beta}.
  \label{eq:HP-cd-not-beta}
\end{equation}
For \(c\ne\beta\),
\begin{equation}
  \Upsilon_{\mathbf m}^{-1}\mathsf a_{c\beta}^{\dagger}\mathsf a_{\beta\beta}\Upsilon_{\mathbf m}
  =
  \mathsf a_{c\beta}^{\dagger}
  \sqrt{m_\beta-\mathsf N^{\mathrm{off}}_\beta}.
  \label{eq:HP-cbeta}
\end{equation}
Indeed, applying the left-hand side to \(|\mathbf n\rangle_{\mathrm{off}}\) gives
\begin{equation}
  \sqrt{m_\beta-N^{\mathrm{off}}_\beta(\mathbf n)}
  \sqrt{n_{c\beta}+1}
  |\mathbf n+\mathbf e_{c\beta}\rangle_{\mathrm{off}},
\end{equation}
which is the action of the right-hand side of \eqref{eq:HP-cbeta}.  Similarly, if \(d\ne\beta\), then
\begin{equation}
  \Upsilon_{\mathbf m}^{-1}\mathsf a_{\beta\beta}^{\dagger}\mathsf a_{d\beta}\Upsilon_{\mathbf m}
  =
  \sqrt{m_\beta-\mathsf N^{\mathrm{off}}_\beta}
  \,\mathsf a_{d\beta}.
  \label{eq:HP-betad}
\end{equation}
To check the order of the two factors in \eqref{eq:HP-betad}, apply the right-hand side to \(|\mathbf n\rangle_{\mathrm{off}}\):
\begin{align}
  \sqrt{m_\beta-\mathsf N^{\mathrm{off}}_\beta}
  \,\mathsf a_{d\beta}|\mathbf n\rangle_{\mathrm{off}}
  &=
  \sqrt{n_{d\beta}}
  \sqrt{m_\beta-(N^{\mathrm{off}}_\beta(\mathbf n)-1)}
  |\mathbf n-\mathbf e_{d\beta}\rangle_{\mathrm{off}}
  \notag\\
  &=
  \sqrt{n_{d\beta}}
  \sqrt{m_\beta-N^{\mathrm{off}}_\beta(\mathbf n)+1}
  |\mathbf n-\mathbf e_{d\beta}\rangle_{\mathrm{off}}.
\end{align}
This is also the result of applying \(\mathsf a_{\beta\beta}^{\dagger}\mathsf a_{d\beta}\) to the basis vector \eqref{eq:target-basis-HP} and then using \eqref{eq:Upsilon-inverse-on-basis-HP}.  Finally,
\begin{equation}
  \Upsilon_{\mathbf m}^{-1}\mathsf a_{\beta\beta}^{\dagger}\mathsf a_{\beta\beta}\Upsilon_{\mathbf m}
  =
  m_\beta-\mathsf N^{\mathrm{off}}_\beta.
  \label{eq:HP-betabeta}
\end{equation}
Equations \eqref{eq:HP-cd-not-beta}--\eqref{eq:HP-betabeta} are the formulas used below.  
They are the Holstein--Primakoff-type formulas, written here directly in the occupation-number basis and with the vector spaces specified explicitly.

\subsection[Direct contraction of the L-operator]{Direct contraction of the \texorpdfstring{$L$}{L}-operator}
We now apply the identification \eqref{eq:Upsilon-def-HP} to the
\(L\)-operator \eqref{eq:genuine-local-L}.  Define
\begin{equation}
  S_{\mathbf m}
  =\sum_{a\in I}E_{aa}+
    \sum_{\beta\in\Ibar}m_\beta^{-1/2}E_{\beta\beta},
  \label{eq:S-m-def}
\end{equation}
where \(m_\beta\in\Z_{>0}\) for every \(\beta\in\Ibar\).
It is convenient to denote its diagonal matrix elements by
\begin{equation}
  s_c=
  \begin{cases}
  1, & c\in I,\\[1mm]
  m_c^{-1/2}, & c\in\Ibar.
  \end{cases}
  \label{eq:s-c-def-HP}
\end{equation}

The rescaled \(L\)-operator transported to
\(\F_{\mathbf m}^{\mathrm{off}}\) is
\begin{equation}
  \widetilde L^{(\mathbf m)}(v)
  :=
  (\Upsilon_{\mathbf m}^{-1}\otimes \Id_{\C^M})
  (\Id_{\F_{\mathbf m}}\otimes S_{\mathbf m})
  L^{(\Ibar)}(v)
  (\Id_{\F_{\mathbf m}}\otimes S_{\mathbf m})
  (\Upsilon_{\mathbf m}\otimes \Id_{\C^M}).
  \label{eq:L-m-scaled-HP}
\end{equation}
Here \(S_{\mathbf m}\) acts on the \(\C^M\) on which the row and column
labels of the \(L\)-operator are carried, while
\(\Upsilon_{\mathbf m}\) identifies the two occupation-number spaces.
Hence the matrix element of \(\widetilde L^{(\mathbf m)}(v)\) with row
label \(c\) and column label \(d\) is
\begin{equation}
  [\widetilde L^{(\mathbf m)}(v)]_{cd}
  =
  s_cs_d\,
  \Upsilon_{\mathbf m}^{-1}
  [L^{(\Ibar)}(v)]_{cd}
  \Upsilon_{\mathbf m}.
  \label{eq:L-m-component-HP-first}
\end{equation}
Using \eqref{eq:local-L-component}, this becomes
\begin{equation}
  [\widetilde L^{(\mathbf m)}(v)]_{cd}
  =
  s_cs_d\left(
  v\delta_{cd}
  +
  \sum_{\alpha\in\Ibar}
  \Upsilon_{\mathbf m}^{-1}
  \mathsf a_{d\alpha}^{\dagger}\mathsf a_{c\alpha}
  \Upsilon_{\mathbf m}
  \right).
  \label{eq:L-m-component-HP}
\end{equation}
We compute the four types of matrix elements.  The arrows below are understood
on each basis vector \(|\mathbf n\rangle_{\mathrm{off}}\) whose occupation
numbers are fixed while the integers \(m_\beta\) tend to infinity in the sense
of \eqref{eq:multi-m-limit-meaning}.

First let \(a,c\in I\).  Then \(s_a=s_c=1\).  Since \(a,c\ne\beta\) for any
\(\beta\in\Ibar\), we obtain
\begin{align}
  [\widetilde L^{(\mathbf m)}(v)]_{ac}
  &=v\delta_{ac}+
  \sum_{\beta\in\Ibar}
  \Upsilon_{\mathbf m}^{-1}
  \mathsf a_{c\beta}^{\dagger}\mathsf a_{a\beta}
  \Upsilon_{\mathbf m}
  \notag\\
  &=v\delta_{ac}+\sum_{\beta\in\Ibar}
  \mathsf a_{c\beta}^{\dagger}\mathsf a_{a\beta}
  \qquad[\text{by }\eqref{eq:HP-cd-not-beta}].
  \label{eq:HP-I-I-block}
\end{align}
This matrix element does not depend on \(\mathbf m\).

Next let \(a\in I\) and \(\beta\in\Ibar\).  We have
\begin{align}
  [\widetilde L^{(\mathbf m)}(v)]_{a\beta}
  &=m_\beta^{-1/2}
  \sum_{\alpha\in\Ibar}
  \Upsilon_{\mathbf m}^{-1}
  \mathsf a_{\beta\alpha}^{\dagger}\mathsf a_{a\alpha}
  \Upsilon_{\mathbf m}
  \notag\\
  &=m_\beta^{-1/2}
  \Upsilon_{\mathbf m}^{-1}
  \mathsf a_{\beta\beta}^{\dagger}\mathsf a_{a\beta}
  \Upsilon_{\mathbf m}
  \notag\\
  &\quad+
  m_\beta^{-1/2}
  \sum_{\substack{\alpha\in\Ibar\\ \alpha\ne\beta}}
  \Upsilon_{\mathbf m}^{-1}
  \mathsf a_{\beta\alpha}^{\dagger}\mathsf a_{a\alpha}
  \Upsilon_{\mathbf m}.
  \label{eq:HP-I-Ibar-start}
\end{align}
By \eqref{eq:HP-betad}, the first term is
\begin{equation}
  \sqrt{1-\frac{\mathsf N_\beta^{\mathrm{off}}}{m_\beta}}\,
  \mathsf a_{a\beta}.
\end{equation}
For \(\alpha\ne\beta\), \eqref{eq:HP-cd-not-beta} gives
\begin{equation}
  \Upsilon_{\mathbf m}^{-1}
  \mathsf a_{\beta\alpha}^{\dagger}\mathsf a_{a\alpha}
  \Upsilon_{\mathbf m}
  =
  \mathsf a_{\beta\alpha}^{\dagger}\mathsf a_{a\alpha}.
\end{equation}
Thus the second term in \eqref{eq:HP-I-Ibar-start} has the explicit factor
\(m_\beta^{-1/2}\) and tends to zero.  Hence
\begin{equation}
  [\widetilde L^{(\mathbf m)}(v)]_{a\beta}
  \longrightarrow
  \mathsf a_{a\beta}.
  \label{eq:HP-I-Ibar-limit}
\end{equation}

Let \(\beta\in\Ibar\) and \(a\in I\).  Similarly,
\begin{align}
  [\widetilde L^{(\mathbf m)}(v)]_{\beta a}
  &=m_\beta^{-1/2}
  \sum_{\alpha\in\Ibar}
  \Upsilon_{\mathbf m}^{-1}
  \mathsf a_{a\alpha}^{\dagger}\mathsf a_{\beta\alpha}
  \Upsilon_{\mathbf m}
  \notag\\
  &=m_\beta^{-1/2}
  \Upsilon_{\mathbf m}^{-1}
  \mathsf a_{a\beta}^{\dagger}\mathsf a_{\beta\beta}
  \Upsilon_{\mathbf m}
  \notag\\
  &\quad+
  m_\beta^{-1/2}
  \sum_{\substack{\alpha\in\Ibar\\ \alpha\ne\beta}}
  \Upsilon_{\mathbf m}^{-1}
  \mathsf a_{a\alpha}^{\dagger}\mathsf a_{\beta\alpha}
  \Upsilon_{\mathbf m}.
  \label{eq:HP-Ibar-I-start}
\end{align}
By \eqref{eq:HP-cbeta}, the first term is
\begin{equation}
  \mathsf a_{a\beta}^{\dagger}
  \sqrt{1-\frac{\mathsf N_\beta^{\mathrm{off}}}{m_\beta}}.
\end{equation}
For \(\alpha\ne\beta\), \eqref{eq:HP-cd-not-beta} gives
\begin{equation}
  \Upsilon_{\mathbf m}^{-1}
  \mathsf a_{a\alpha}^{\dagger}\mathsf a_{\beta\alpha}
  \Upsilon_{\mathbf m}
  =
  \mathsf a_{a\alpha}^{\dagger}\mathsf a_{\beta\alpha}.
\end{equation}
Thus the second term in \eqref{eq:HP-Ibar-I-start} has the explicit factor
\(m_\beta^{-1/2}\) and tends to zero.  Hence
\begin{equation}
  [\widetilde L^{(\mathbf m)}(v)]_{\beta a}
  \longrightarrow
  \mathsf a_{a\beta}^{\dagger}.
  \label{eq:HP-Ibar-I-limit}
\end{equation}

Now let \(\beta\in\Ibar\).  For the \((\beta,\beta)\)-matrix element,
\begin{align}
  [\widetilde L^{(\mathbf m)}(v)]_{\beta\beta}
  &=
  m_\beta^{-1}
  \left(
  v+
  \sum_{\alpha\in\Ibar}
  \Upsilon_{\mathbf m}^{-1}
  \mathsf a_{\beta\alpha}^{\dagger}\mathsf a_{\beta\alpha}
  \Upsilon_{\mathbf m}
  \right)
  \notag\\
  &=m_\beta^{-1}v
  +m_\beta^{-1}(m_\beta-\mathsf N_\beta^{\mathrm{off}})
  \notag\\
  &\quad+
  m_\beta^{-1}
  \sum_{\substack{\alpha\in\Ibar\\ \alpha\ne\beta}}
  \mathsf a_{\beta\alpha}^{\dagger}\mathsf a_{\beta\alpha}
  \qquad[\text{by }\eqref{eq:HP-betabeta}\text{ and }\eqref{eq:HP-cd-not-beta}]
  \notag\\
  &=1+\frac{v-\mathsf N_\beta^{\mathrm{off}}}{m_\beta}
  +m_\beta^{-1}
  \sum_{\substack{\alpha\in\Ibar\\ \alpha\ne\beta}}
  \mathsf a_{\beta\alpha}^{\dagger}\mathsf a_{\beta\alpha}.
  \label{eq:HP-diagonal-Ibar}
\end{align}
Thus
\begin{equation}
  [\widetilde L^{(\mathbf m)}(v)]_{\beta\beta}
  \longrightarrow
  1.
  \label{eq:HP-diagonal-Ibar-limit}
\end{equation}

Finally let \(\beta,\gamma\in\Ibar\) with \(\beta\ne\gamma\).  Then
\begin{align}
  [\widetilde L^{(\mathbf m)}(v)]_{\beta\gamma}
  &=(m_\beta m_\gamma)^{-1/2}
  \sum_{\alpha\in\Ibar}
  \Upsilon_{\mathbf m}^{-1}
  \mathsf a_{\gamma\alpha}^{\dagger}\mathsf a_{\beta\alpha}
  \Upsilon_{\mathbf m}
  \notag\\
  &=(m_\beta m_\gamma)^{-1/2}
  \Upsilon_{\mathbf m}^{-1}
  \mathsf a_{\gamma\beta}^{\dagger}\mathsf a_{\beta\beta}
  \Upsilon_{\mathbf m}
  \notag\\
  &\quad+(m_\beta m_\gamma)^{-1/2}
  \Upsilon_{\mathbf m}^{-1}
  \mathsf a_{\gamma\gamma}^{\dagger}\mathsf a_{\beta\gamma}
  \Upsilon_{\mathbf m}
  \notag\\
  &\quad+(m_\beta m_\gamma)^{-1/2}
  \sum_{\substack{\alpha\in\Ibar\\ \alpha\ne\beta,\gamma}}
  \Upsilon_{\mathbf m}^{-1}
  \mathsf a_{\gamma\alpha}^{\dagger}\mathsf a_{\beta\alpha}
  \Upsilon_{\mathbf m}.
  \label{eq:HP-offdiag-Ibar-start}
\end{align}
Using \eqref{eq:HP-cbeta} in the first term gives
\begin{equation}
  (m_\beta m_\gamma)^{-1/2}
  \mathsf a_{\gamma\beta}^{\dagger}
  \sqrt{m_\beta-\mathsf N_\beta^{\mathrm{off}}}
  =
  m_\gamma^{-1/2}
  \mathsf a_{\gamma\beta}^{\dagger}
  \sqrt{1-\frac{\mathsf N_\beta^{\mathrm{off}}}{m_\beta}},
\end{equation}
which tends to zero.  Using \eqref{eq:HP-betad} in the second term gives
\begin{equation}
  (m_\beta m_\gamma)^{-1/2}
  \sqrt{m_\gamma-\mathsf N_\gamma^{\mathrm{off}}}\,
  \mathsf a_{\beta\gamma}
  =
  m_\beta^{-1/2}
  \sqrt{1-\frac{\mathsf N_\gamma^{\mathrm{off}}}{m_\gamma}}\,
  \mathsf a_{\beta\gamma},
\end{equation}
which also tends to zero.  For the remaining terms, \eqref{eq:HP-cd-not-beta}
gives expressions independent of \(m_\beta\) and \(m_\gamma\), multiplied by
\((m_\beta m_\gamma)^{-1/2}\).  Hence
\begin{equation}
  [\widetilde L^{(\mathbf m)}(v)]_{\beta\gamma}
  \longrightarrow
  0
  \qquad(\beta\ne\gamma).
  \label{eq:HP-offdiag-Ibar-limit}
\end{equation}

Combining \eqref{eq:HP-I-I-block}, \eqref{eq:HP-I-Ibar-limit},
\eqref{eq:HP-Ibar-I-limit}, \eqref{eq:HP-diagonal-Ibar-limit}, and
\eqref{eq:HP-offdiag-Ibar-limit}, we obtain the limit of matrix elements
\begin{equation}
  [\widetilde L^{(\mathbf m)}(v)]_{cd}
  \longrightarrow
  [\mathscr L_{I,\Ibar}(v)]_{cd}
  \qquad(c,d\in B),
  \label{eq:HP-local-L-limit}
\end{equation}
where \(\mathscr L_{I,\Ibar}(v)\) is the \(L\)-operator in \eqref{eq:contracted-L}.
The RLL relation is preserved by the same limit.  For finite
\(\mathbf m\), the restriction of \(L^{(\Ibar)}(v)\) to \(\F_{\mathbf m}\)
satisfies \eqref{eq:RLL-local-unscaled}.  Since \(S_{\mathbf m}\) is diagonal,
\[
 R(u)(S_{\mathbf m}\otimes S_{\mathbf m})
 =(S_{\mathbf m}\otimes S_{\mathbf m})R(u).
\]
Thus the rescaling by \(S_{\mathbf m}\) and the conjugation by
\(\Upsilon_{\mathbf m}\) in \eqref{eq:L-m-scaled-HP} preserve the RLL
relation.  Each matrix element of the RLL relation is a finite sum, so the limit
can be taken term by term.  Using \eqref{eq:HP-local-L-limit} gives
\begin{equation}
  \mathscr L_{I,\Ibar,12}(u)\mathscr L_{I,\Ibar,13}(v)R_{23}(v-u)
  =
  R_{23}(v-u)\mathscr L_{I,\Ibar,13}(v)\mathscr L_{I,\Ibar,12}(u).
  \label{eq:RLL-contracted-HP}
\end{equation}

\subsection[Fixed-particle-number traces]{Traces over the fixed-particle-number sectors}
\label{subsec:fixed-particle-traces}

We now pass from the $L$-operator limit \eqref{eq:HP-local-L-limit} to
traces.  The restriction of the diagonal boundary twist to
$\F_{\mathbf m}$ is
\begin{equation}
  \widehat g_{\mathbf m}
  =
  \prod_{\beta\in\Ibar}\prod_{c\in B}g_c^{\mathsf N_{c\beta}}.
  \label{eq:ghat-m-HP}
\end{equation}
For each $\beta\in\Ibar$, the condition
$\sum_{c\in B}\mathsf N_{c\beta}=m_\beta$ gives
\begin{align}
  \prod_{c\in B}g_c^{\mathsf N_{c\beta}}
  &=g_\beta^{m_\beta}
  \prod_{\substack{c\in B\\ c\ne\beta}}
  \left(\frac{g_c}{g_\beta}\right)^{\mathsf N_{c\beta}}.
  \label{eq:twist-factor-fixed-particle-beta}
\end{align}
By \eqref{eq:HP-cd-not-beta},
$\Upsilon_{\mathbf m}^{-1}\mathsf N_{c\beta}\Upsilon_{\mathbf m}
=\mathsf N_{c\beta}$ for $c\ne\beta$.  After removing the scalar
$\prod_{\beta\in\Ibar}g_\beta^{m_\beta}$, we define
$D_{\mathbf m}$ on $\F_{\mathbf m}^{\mathrm{off}}$ by
\begin{equation}
  D_{\mathbf m}
  :=
  \prod_{\beta\in\Ibar}
  \prod_{\substack{c\in B\\ c\ne\beta}}
  \left(\frac{g_c}{g_\beta}\right)^{\mathsf N_{c\beta}}.
  \label{eq:Dm-HP}
\end{equation}
Since
$B\setminus\{\beta\}=I\sqcup(\Ibar\setminus\{\beta\})$,
\begin{equation}
  D_{\mathbf m}
  =D_{I,\Ibar}\,W_{\Ibar,\Ibar}^{\mathrm{off}}(0).
  \label{eq:Dm-split-HP}
\end{equation}

Let
\begin{equation}
  \widetilde M_{\mathbf m}(u)
  =
  \widetilde L^{(\mathbf m)}_{L}(u-\theta_L)\cdots
  \widetilde L^{(\mathbf m)}_{1}(u-\theta_1)
  \label{eq:Mtilde-m-HP}
\end{equation}
and define
\begin{equation}
  \widetilde T_{\mathbf m}(u)
  =
  \Tr_{\F_{\mathbf m}^{\mathrm{off}}}
  \left[\widetilde M_{\mathbf m}(u)D_{\mathbf m}\right].
  \label{eq:T-m-fixed-degree}
\end{equation}

To state the limit, replace the ratios $g_c/g_\beta$ by independent
variables $\rho_{c\beta}$ and set
\begin{equation}
  D_{\mathbf m}(\boldsymbol\rho)
  =
  \prod_{\beta\in\Ibar}
  \prod_{\substack{c\in B\\ c\ne\beta}}
  \rho_{c\beta}^{\mathsf N_{c\beta}}.
  \label{eq:Dm-rho-HP}
\end{equation}
Write $\widetilde T_{\mathbf m}(u;\boldsymbol\rho)$ for
\eqref{eq:T-m-fixed-degree} with $D_{\mathbf m}$ replaced by
$D_{\mathbf m}(\boldsymbol\rho)$.  Likewise, let
$Q_I(u;\boldsymbol\rho)$ be obtained from \eqref{eq:Sch-Q-def} by replacing
$g_a/g_\beta$ with $\rho_{a\beta}$, and set
\[
  C_{\Ibar}(\boldsymbol\rho)
  =
  \prod_{\beta\in\Ibar}
  \prod_{\substack{\gamma\in\Ibar\\ \gamma\ne\beta}}
  \frac{1}{1-\rho_{\gamma\beta}}.
\]

For non-negative integers
$\boldsymbol\nu=(\nu_{c\beta})$, put
\begin{equation}
  \boldsymbol\rho^{\boldsymbol\nu}
  :=
  \prod_{\beta\in\Ibar}
  \prod_{\substack{c\in B\\ c\ne\beta}}
  \rho_{c\beta}^{\nu_{c\beta}},
  \label{eq:rho-monomial-HP}
\end{equation}
and denote extraction of its coefficient by
$[\boldsymbol\rho^{\boldsymbol\nu}]$.  Let
$|\boldsymbol\nu\rangle_{\mathrm{off}}$ be the vector
\eqref{eq:off-basis-HP} with $\mathbf n=\boldsymbol\nu$.  It belongs to
$\F_{\mathbf m}^{\mathrm{off}}$ if and only if
\begin{equation}
  \sum_{c\in B\setminus\{\beta\}}\nu_{c\beta}\le m_\beta,
  \qquad \beta\in\Ibar.
  \label{eq:truncation-in-trace-limit}
\end{equation}
For such $\mathbf m$, the diagonality of
$D_{\mathbf m}(\boldsymbol\rho)$ gives the identity in
$\End(\Hq_L)$
\begin{equation}
  [\boldsymbol\rho^{\boldsymbol\nu}]
  \widetilde T_{\mathbf m}(u;\boldsymbol\rho)
  =
  {}_{\mathrm{off}}\!\langle\boldsymbol\nu|
  \widetilde M_{\mathbf m}(u)
  |\boldsymbol\nu\rangle_{\mathrm{off}}.
  \label{eq:fixed-coefficient-as-matrix-element}
\end{equation}
For fixed $\boldsymbol\nu$, \eqref{eq:Mtilde-m-HP} is a finite product,
so the limits of the matrix elements in \eqref{eq:HP-local-L-limit} may
be taken successively.  By the definitions of
$Q_I(u;\boldsymbol\rho)$ and $C_{\Ibar}(\boldsymbol\rho)$, this gives
\begin{equation}
  \lim_{m_\beta\to\infty\;(\beta\in\Ibar)}
  [\boldsymbol\rho^{\boldsymbol\nu}]
  \widetilde T_{\mathbf m}(u;\boldsymbol\rho)
  =
  [\boldsymbol\rho^{\boldsymbol\nu}]
  \left(C_{\Ibar}(\boldsymbol\rho)
  Q_I(u;\boldsymbol\rho)\right).
  \label{eq:fixed-coefficient-limit}
\end{equation}
Since $\boldsymbol\nu$ is arbitrary, every coefficient of
$\widetilde T_{\mathbf m}(u;\boldsymbol\rho)$ tends to the corresponding
coefficient of $C_{\Ibar}(\boldsymbol\rho)Q_I(u;\boldsymbol\rho)$:
\[
  \widetilde T_{\mathbf m}(u;\boldsymbol\rho)
  \longrightarrow
  C_{\Ibar}(\boldsymbol\rho)Q_I(u;\boldsymbol\rho)
  \qquad
  \left(\min_{\beta\in\Ibar}m_\beta\to\infty\right).
\]
After taking the limit coefficient by coefficient, we set
$\rho_{c\beta}=g_c/g_\beta$ in the resulting rational functions.  This does
not assert a numerical limit at fixed $g_1,\ldots,g_M$.

The next subsection identifies $\widetilde T_{\mathbf m}(u)$ as a scaled
coefficient of the master $T$-operator in the variables $y_\beta$.

\subsection{Relation with the original master \texorpdfstring{$T$}{T}-operator and the residue point}
\label{subsec:fixed-particle-trace-from-master}

The trace in \eqref{eq:T-m-fixed-degree} comes from the master \(T\)-operator
by coefficient extraction in the variables \(y_\beta\).  This coefficient
extraction is different from the residue operation at \(g_\beta y_\beta=1\).
The present subsection explains the relation precisely.

Let \(R\) be a finite flavour set with \(B\subset R\).  We use the master
\(T\)-operator \(\mathcal T(u;\mathbf y_R)\) defined in
\eqref{eq:finite-master-def}.  For each \(\alpha\in R\), the operator \(\mathsf N_\alpha\) commutes with the
monodromy matrix \(M^{(R)}(u)\) in \eqref{eq:Schwinger-monodromy}.  Indeed,
\begin{equation}
  [\mathsf N_\alpha,\mathsf a_{c\gamma}^{\dagger}]
  =\delta_{\alpha\gamma}\mathsf a_{c\gamma}^{\dagger},
  \qquad
  [\mathsf N_\alpha,\mathsf a_{d\gamma}]
  =-\delta_{\alpha\gamma}\mathsf a_{d\gamma},
\end{equation}
and hence
\begin{align}
  [\mathsf N_\alpha,\mathsf a_{c\gamma}^{\dagger}\mathsf a_{d\gamma}]
  &=[\mathsf N_\alpha,\mathsf a_{c\gamma}^{\dagger}]\mathsf a_{d\gamma}
    +\mathsf a_{c\gamma}^{\dagger}[\mathsf N_\alpha,\mathsf a_{d\gamma}]
  \notag\\
  &=0.
  \label{eq:Nalpha-commutes-bilinears}
\end{align}
Each matrix element of \(L^{(R)}(v)\) is a scalar plus a linear combination of
these bilinears.  Therefore
\begin{equation}
  [\mathsf N_\alpha,M^{(R)}(u)]=0,
  \qquad \alpha\in R.
  \label{eq:Nalpha-commutes-monodromy}
\end{equation}
It follows that \(\F_{B,R}\) decomposes as
\begin{equation}
  \F_{B,R}
  =
  \bigoplus_{\mathbf n}\F_{B,R}[\mathbf n],
  \qquad
  \F_{B,R}[\mathbf n]
  :=
  \bigcap_{\alpha\in R}\ker(\mathsf N_\alpha-n_\alpha),
  \label{eq:full-flavour-degree-decomp}
\end{equation}
where \(\mathbf n=(n_\alpha)_{\alpha\in R}\) runs over all families of
non-negative integers indexed by \(R\).  
Since \(M^{(R)}(u)\), \(\widehat g\), and
\(\widehat h(\mathbf y_R)\) commute with all \(\mathsf N_\alpha\),
the operator under the trace in \eqref{eq:finite-master-def} preserves
\(\F_{B,R}[\mathbf n]\otimes\Hq_L\) for every \(\mathbf n\).

For a family \(\mathbf m=(m_\beta)_{\beta\in\Ibar}\) of non-negative integers,
write
\begin{equation}
  y^{\mathbf m}:=\prod_{\beta\in\Ibar}y_\beta^{m_\beta}.
\end{equation}
If
\begin{equation}
  F=
  \sum_{\substack{\mathbf n=(n_\beta)_{\beta\in\Ibar}\\
                  n_\beta\in\Z_{\ge0}}}
  F_{\mathbf n}\prod_{\beta\in\Ibar}y_\beta^{n_\beta},
\end{equation}
we set
\begin{equation}
  [y^{\mathbf m}]_{\Ibar}F:=F_{\mathbf m}.
  \label{eq:coef-y-m-definition-Ibar}
\end{equation}
Similarly, if \(S\subset R\), the notation \([y^{\mathbf n}]_S\) means
coefficient extraction in the variables \(y_\alpha\), \(\alpha\in S\), for a
family \(\mathbf n=(n_\alpha)_{\alpha\in S}\).  In particular,
\([y^{\mathbf0}]_{R\setminus\Ibar}\) extracts the part in which all flavour
labels outside \(\Ibar\) have total particle number zero.

Define
\begin{equation}
  T_{\mathbf m}(u)
  :=[y^{\mathbf m}]_{\Ibar}[y^{\mathbf0}]_{R\setminus\Ibar}
  \mathcal T(u;\mathbf y_R).
  \label{eq:fixed-trace-from-master-y}
\end{equation}
To identify this coefficient, note that
\([y^{\mathbf0}]_{R\setminus\Ibar}\) restricts the trace to the subspace
on which
\begin{equation}
  \mathsf N_\alpha=0,
  \qquad \alpha\in R\setminus\Ibar.
\end{equation}
For such an \(\alpha\), all oscillators with flavour label \(\alpha\)
have occupation number zero on this subspace.  Hence
\(\mathsf a_{d\alpha}\) annihilates every vector in it, and therefore
\begin{equation}
  \mathsf a_{c\alpha}^{\dagger}\mathsf a_{d\alpha}=0
  \qquad\hbox{on the subspace }\mathsf N_\alpha=0.
\end{equation}
Therefore the operators \(e_{cd}^{(R)}\) reduce to sums over the selected flavour labels:
\begin{equation}
  e_{cd}^{(R)}\big|_{\mathsf N_\alpha=0\ (\alpha\notin\Ibar)}
  =
  \sum_{\beta\in\Ibar}\mathsf a_{c\beta}^{\dagger}\mathsf a_{d\beta}
  =e_{cd}^{(\Ibar)}.
  \label{eq:eR-reduces-to-eIbar}
\end{equation}
The operation \([y^{\mathbf m}]_{\Ibar}\) then fixes
\begin{equation}
  \mathsf N_\beta=m_\beta,
  \qquad \beta\in\Ibar.
\end{equation}
Consequently,
\begin{equation}
  T_{\mathbf m}(u)
  =
  \Tr_{\F_{\mathbf m}}
  \left[
  M^{(\Ibar)}(u)
  \prod_{\beta\in\Ibar}\prod_{c\in B}g_c^{\mathsf N_{c\beta}}
  \right],
  \label{eq:coefficient-is-fixed-degree-trace}
\end{equation}
where \(M^{(\Ibar)}(u)\) is the monodromy matrix in
\eqref{eq:Schwinger-monodromy} with the flavour set equal to \(\Ibar\).

Using \eqref{eq:twist-factor-fixed-particle-beta}, the scalar factor coming from the boundary twist on \(\F_{\mathbf m}\) is \(\prod_{\beta\in\Ibar}g_\beta^{m_\beta}\).  Removing it gives the normalized coefficient
\begin{equation}
  T_{\mathbf m}^{\rm norm}(u)
  =
  \left(\prod_{\beta\in\Ibar}g_\beta^{-m_\beta}\right)
  [y^{\mathbf m}]_{\Ibar}[y^{\mathbf0}]_{R\setminus\Ibar}
  \mathcal T(u;\mathbf y_R).
  \label{eq:normalized-fixed-trace-from-master}
\end{equation}
If the scaling matrix \(S_{\mathbf m}\) in \eqref{eq:S-m-def} is inserted on
all quantum sites, then the trace in \eqref{eq:T-m-fixed-degree} is
\begin{equation}
  \widetilde T_{\mathbf m}(u)
  =
  \left(\prod_{\beta\in\Ibar}g_\beta^{-m_\beta}\right)
  S_{\mathbf m}^{\otimes L}
  \left(
  [y^{\mathbf m}]_{\Ibar}[y^{\mathbf0}]_{R\setminus\Ibar}
  \mathcal T(u;\mathbf y_R)
  \right)
  S_{\mathbf m}^{\otimes L}.
  \label{eq:scaled-fixed-trace-from-master}
\end{equation}
The equality with \eqref{eq:T-m-fixed-degree} follows from the definition of
\(\widetilde L^{(\mathbf m)}(v)\) in \eqref{eq:L-m-scaled-HP} and from the
invariance of trace under the change of basis \(\Upsilon_{\mathbf m}\).
Equations \eqref{eq:fixed-trace-from-master-y},
\eqref{eq:normalized-fixed-trace-from-master}, and
\eqref{eq:scaled-fixed-trace-from-master} express the traces over \(\mathcal F_{\mathbf m}\) as coefficients in the expansion of the master \(T\)-operator in the variables \(y_\beta\).
They are not residue formulas.  The residue formula of
Proposition~\ref{prop:main-contraction} concerns the behaviour of the whole
generating series near \(g_\beta y_\beta=1\), whereas here the particle numbers
\(m_\beta\) are fixed first and are sent to infinity only afterwards.

Thus this section proves the following three statements.  First, the \(L\)-operator
\eqref{eq:contracted-L} is obtained from the Schwinger \(L\)-operator by the
large-occupation-number limit of the matrix elements in
\eqref{eq:HP-local-L-limit}.  Second, the RLL relation is preserved in
that limit, as shown in \eqref{eq:RLL-contracted-HP}.  Third, the trace over the
subspace \(\F_{\mathbf m}\) is obtained from the master \(T\)-operator by the
coefficient extraction displayed in \eqref{eq:scaled-fixed-trace-from-master}.

\section{\texorpdfstring{$QQ$}{QQ}-relations from the mKP bilinear identity}
\label{sec:QQ-relations}

The \(QQ\)-relations are functional relations for Baxter
\(Q\)-operators, or for their eigenvalues, and may be viewed as
generalizations of quantum Wronskian conditions.  They appear in various
forms in the literature; early operator-valued examples can be found in
\cite{BLZ1999,BazhanovHibberdKhoroshkin2002}.  In this section we reformulate, in the conventions of the present paper, the derivation of the operator-valued \(QQ\)-relation following
\cite{AKLTZ2013,KLT2012}.

We first recall the normalization relating the trace \(Q_I(u)\) in
\eqref{eq:Sch-Q-def} to the oscillator construction of
\cite{BFLMS2011}; the precise comparison is given in
\ref{app:BFLMS-local-comparison}.  We then derive the \(QQ\)-relation from
the mKP bilinear identity for the master \(T\)-operator by taking residues in
the variables \(y_\beta\), \(\beta\in\Ibar\), at
\(y_\beta=g_\beta^{-1}\).

\subsection{Diagonal factor and normalized trace}

We use the diagonal factor \(D_{I,\Ibar}\) on \(\F_{I,\Ibar}\) defined in
\eqref{eq:DAB-derived}.  In the sense of the Fock-space trace fixed in
Section~\ref{sec:conventions},
\begin{align}
  \chi_{I,\Ibar}
  &=\Tr_{\F_{I,\Ibar}}D_{I,\Ibar}
  \notag\\
  &=\prod_{a\in I}\prod_{\beta\in\Ibar}
    \sum_{n\ge0}\left(\frac{g_a}{g_\beta}\right)^n
  =\prod_{a\in I}\prod_{\beta\in\Ibar}
    \frac1{1-g_a/g_\beta}.
  \label{eq:chi-I-Ibar}
\end{align}
The normalized version of \eqref{eq:Sch-Q-def} is
\begin{equation}
  \widehat Q_I(u)=\chi_{I,\Ibar}^{-1}Q_I(u).
  \label{eq:Qhat-normalized-def}
\end{equation}
This normalization is stated only for comparison with conventions used in papers on oscillator construction of $Q$-operators.

\subsection{\texorpdfstring{$QQ$}{QQ}-relations}
The next proposition states the $QQ$-relation for the residue-defined
operators \(\mathscr Q_J^{\mathrm{res}}(u)\).  Its proof uses Theorem~\ref{thm:CBR-mKP-route}, namely the mKP bilinear identity for the master \(T\)-operator.  In the specialization below, we use the coefficient of \(z^{-1}\) in the Laurent expansion in \(z^{-1}\).  The residues are then taken only with respect to the Miwa variables at the points \(y_\beta=g_\beta^{-1}\).

For a finite subset \(S\subset B\), set
\begin{equation}
  [y]=\left(y,\frac{y^2}{2},\frac{y^3}{3},\ldots\right),
  \qquad
  \wt(\mathbf y_S)=\sum_{\mu\in S}[y_\mu],
  \qquad
  \mathbf y_S=(y_\mu)_{\mu\in S}.
  \label{eq:Miwa-vector-for-subsets}
\end{equation}
By \eqref{eq:finite-master-is-Miwa-specialization},
\begin{equation}
  \mathcal T(u;\mathbf y_S)=\mathcal T(u,\wt(\mathbf y_S)).
  \label{eq:T-yS-from-Miwa-times}
\end{equation}
Here the flavour set is indicated by the subscript on \(\mathbf y_S\).

For any subset \(J\subset B\), put
\begin{equation}
  \overline J=B\setminus J.
\end{equation}
Define
\begin{align}
  \mathscr Q_J^{\mathrm{res}}(u)
  &:={\operatorname*{Res}}_{y_\beta=g_\beta^{-1},\,\beta\in\overline J}
  \left[
  \prod_{\beta\in\overline J}(1-g_\beta y_\beta)^{\mathsf M_\beta}
  \mathcal T(u;\mathbf y_{\overline J})
  \prod_{\beta\in\overline J}dy_\beta
  \right].
  \label{eq:residue-normalized-Q}
\end{align}
The residue in \eqref{eq:residue-normalized-Q} is understood in the same
operator-valued multiple-residue sense as in \eqref{eq:main-residue-formula}.
If \(\overline J=\emptyset\), the residue symbol and the product are omitted.
By Proposition~\ref{prop:main-contraction},
\begin{equation}
  \mathscr Q_J^{\mathrm{res}}(u)
  =\mathcal N_J Q_J(u),
  \qquad
  \mathcal N_J=(-1)^{|\overline J|}
  \left(\prod_{\beta\in\overline J}g_\beta^{-1}\right)
  C_{\overline J}
  \prod_{\beta\in\overline J}
  \Gamma(\mathsf M_\beta+1).
  \label{eq:Qres-to-Q-normalization}
\end{equation}
The factor \((-1)^{|\overline J|}\prod_{\beta\in\overline J}g_\beta^{-1}\) in
\(\mathcal N_J\) comes from \(dy_\beta=-g_\beta^{-1}d\varepsilon_\beta\), where
\(\varepsilon_\beta=1-g_\beta y_\beta\).

\begin{proposition}[The $QQ$-relation from the mKP bilinear identity]
\label{prop:QQ-relations}
Let \(I\subset B\), and let \(a,b\notin I\) be distinct colour labels.  Then
\begin{align}
  &(g_a-g_b)
  \mathscr Q_I^{\mathrm{res}}(u+1)
  \mathscr Q_{I\cup\{a,b\}}^{\mathrm{res}}(u)
  \notag\\
  &\qquad=
  g_a\,\mathscr Q_{I\cup\{a\}}^{\mathrm{res}}(u+1)
  \mathscr Q_{I\cup\{b\}}^{\mathrm{res}}(u)
  -g_b\,\mathscr Q_{I\cup\{b\}}^{\mathrm{res}}(u+1)
  \mathscr Q_{I\cup\{a\}}^{\mathrm{res}}(u).
  \label{eq:QQ-relation-residue-normalized}
\end{align}
\end{proposition}

\begin{proof}
Put
\begin{equation}
  \tau_u(\wt)=\mathcal T(u,\wt).
\end{equation}
In \eqref{eq:mKP-bilinear-paper}, replace the first spectral parameter by
\(u+1\), put the second one equal to \(u\), and choose
\begin{equation}
  \wt=\ws+[y_a]+[y_b],
  \qquad
  \wt'=\ws,
  \label{eq:Hirota-specialization-times-QQ}
\end{equation}
where \(y_a\) and \(y_b\) are independent variables.  Then
\begin{equation}
  e^{\xi([y_a]+[y_b],z)}
  =\exp\left(\sum_{k\ge1}\frac{y_a^k+y_b^k}{k}z^k\right)
  =\frac{1}{(1-y_a z)(1-y_b z)},
  \label{eq:Hirota-exponential-two-variables}
\end{equation}
where the last expression is expanded in non-negative powers of \(z\).

We write \(w=z^{-1}\), so that the shifts \(\pm[z^{-1}]\) in
\eqref{eq:mKP-bilinear-paper} become \(\pm[w]\).
 Let
\begin{equation}
  F(w)=
  \tau_{u+1}\bigl(\ws+[y_a]+[y_b]-[w]\bigr)
  \tau_u\bigl(\ws+[w]\bigr).
  \label{eq:Fw-def-before-expansion}
\end{equation}
Since \([w]=(w,w^2/2,w^3/3,\ldots)\), substitution of
\(\ws+[w]\) or \(\ws+[y_a]+[y_b]-[w]\) into each Schur function gives a
formal series in non-negative powers of \(w\).  Therefore, coefficient by
coefficient in the Schur function expansion of \(\tau\),
\begin{equation}
  F(w)=\sum_{n\ge0}F_nw^n.
  \label{eq:Fw-series-QQ-proof}
\end{equation}
Equation \eqref{eq:mKP-bilinear-paper} implies that the
coefficient of \(z^{-1}\) in
\begin{equation}
  \frac{z}{(1-y_a z)(1-y_b z)}F(z^{-1})
  \label{eq:coefficient-source-QQ-proof}
\end{equation}
is zero.  Expanding this expression gives
\begin{align}
  \frac{z}{(1-y_a z)(1-y_b z)}F(z^{-1})
  &=\sum_{r,s,n\ge0}y_a^r y_b^sF_n z^{r+s+1-n}.
\end{align}
The coefficient of \(z^{-1}\) is obtained from \(n=r+s+2\).  Hence
\begin{align}
  0
  &=\operatorname*{Coeff}_{z^{-1}}
  \left[
  \frac{z}{(1-y_a z)(1-y_b z)}F(z^{-1})
  \right]
  \notag\\
  &=\sum_{n\ge2}F_n\sum_{r+s=n-2}y_a^r y_b^s
  \notag\\
  &=\frac{1}{y_a-y_b}
  \sum_{n\ge2}F_n\left(y_a^{n-1}-y_b^{n-1}\right)
  \notag\\
  &=\frac{1}{y_a-y_b}
  \left[
  \frac{F(y_a)-F_0-y_aF_1}{y_a}
  -\frac{F(y_b)-F_0-y_bF_1}{y_b}
  \right]
  \notag\\
  &=\frac{y_bF(y_a)-y_aF(y_b)+(y_a-y_b)F_0}
  {y_a y_b (y_a-y_b)}.
  \label{eq:coefficient-calculation-QQ-proof}
\end{align}
Therefore the numerator in the last line of \eqref{eq:coefficient-calculation-QQ-proof} is zero.  Substituting the definition of
\(F\) gives
\begin{align}
  &(y_a-y_b)
  \tau_{u+1}(\ws+[y_a]+[y_b])\tau_u(\ws)
  \notag\\
  &\qquad=
  y_a\,\tau_{u+1}(\ws+[y_a])\tau_u(\ws+[y_b])
  -y_b\,\tau_{u+1}(\ws+[y_b])\tau_u(\ws+[y_a]).
  \label{eq:Hirota-specialization-for-QQ}
\end{align}
This is the only consequence of \eqref{eq:mKP-bilinear-paper} used below; see also \cite[Eq.~(21)]{Shigyo2013}.

Since \(a,b\notin I\), both labels belong to \(\Ibar=B\setminus I\).  We now set
\begin{equation}
  \ws=\sum_{\gamma\in\Ibar\setminus\{a,b\}}[y_\gamma].
  \label{eq:s-for-QQ-proof}
\end{equation}
All operators in \eqref{eq:Hirota-specialization-for-QQ} commute with the
operators \(\mathsf M_\mu\).  We work first on a simultaneous
eigenspace on which
\begin{equation}
  \mathsf M_\mu=m_\mu,
  \qquad
  \mu\in\Ibar,
  \label{eq:M-eigenvalues-for-QQ-proof}
\end{equation}
and then read the resulting equality on the direct sum of these eigenspaces.
Set
\begin{equation}
  \varepsilon_\mu=1-g_\mu y_\mu,
  \qquad \mu\in\Ibar.
\end{equation}
For a subset \(S\subset\Ibar\), Proposition~\ref{prop:main-contraction}
implies that, on the eigenspace \eqref{eq:M-eigenvalues-for-QQ-proof},
\(\mathcal T(v;\mathbf y_S)\) has a Laurent expansion of the form
\begin{equation}
  \mathcal T(v;\mathbf y_S)
  =\sum_{\substack{\mathbf n\in\Z^S\\ n_\mu\ge -m_\mu-1\text{ for all }\mu\in S}}
  T_{S,\mathbf n}(v)
  \prod_{\mu\in S}\varepsilon_\mu^{n_\mu}.
  \label{eq:TS-Laurent-expansion-for-QQ}
\end{equation}
Here \(\Z^S\) means the set of integer-valued functions on \(S\).  Formula
\eqref{eq:residue-normalized-Q} becomes, on the same eigenspace,
\begin{equation}
  \mathscr Q_{B\setminus S}^{\mathrm{res}}(v)
  =\left(\prod_{\mu\in S}(-g_\mu^{-1})\right)
  T_{S,\mathbf n^{\max}_S}(v),
  \qquad
  \mathbf n^{\max}_S(\mu)=-m_\mu-1.
  \label{eq:Qres-as-Laurent-coefficient}
\end{equation}

For an expression \(X\) depending on the variables \(y_\mu\), \(\mu\in\Ibar\), define
on the eigenspace \eqref{eq:M-eigenvalues-for-QQ-proof}
\begin{align}
  \mathcal R_{I;a,b}[X]
  &:={\operatorname*{Res}}_{y_\mu=g_\mu^{-1},\,\mu\in\Ibar}
  \left[
  \varepsilon_a^{m_a}\varepsilon_b^{m_b}
  \prod_{\gamma\in\Ibar\setminus\{a,b\}}\varepsilon_\gamma^{2m_\gamma+1}
  X\prod_{\mu\in\Ibar}dy_\mu
  \right].
  \label{eq:product-residue-operation-QQ-proof}
\end{align}
The power \(2m_\gamma+1\) is needed for \(\gamma\in\Ibar\setminus\{a,b\}\),
because the variable \(y_\gamma\) appears in both factors \(\tau_{u+1}(\cdots)\) and \(\tau_u(\cdots)\) in each term of \eqref{eq:Hirota-specialization-for-QQ}.  For each such \(\gamma\), we use the following residue calculation.  Write
\(\varepsilon=1-gy\) and suppose
\begin{equation}
  F(y)=\sum_{r\ge-m-1}F_r\varepsilon^r,
  \qquad
  G(y)=\sum_{s\ge-m-1}G_s\varepsilon^s.
  \label{eq:common-variable-expansion}
\end{equation}
Then
\begin{align}
  &\operatorname*{Res}_{y=g^{-1}}
  \left[\varepsilon^{2m+1}F(y)G(y)\,dy\right]
  \notag\\
  &\quad=
  \operatorname*{Res}_{\varepsilon=0}
  \left[
  \varepsilon^{2m+1}
  \left(\sum_{r\ge-m-1}F_r\varepsilon^r\right)
  \left(\sum_{s\ge-m-1}G_s\varepsilon^s\right)
  (-g^{-1}d\varepsilon)
  \right]
  \notag\\
  &\quad=-g^{-1}F_{-m-1}G_{-m-1}.
  \label{eq:product-residue-one-variable-explicit}
\end{align}
On the other hand,
\begin{align}
  \operatorname*{Res}_{y=g^{-1}}\left[\varepsilon^mF(y)\,dy\right]
  &=-g^{-1}F_{-m-1},
  \\
  \operatorname*{Res}_{y=g^{-1}}\left[\varepsilon^mG(y)\,dy\right]
  &=-g^{-1}G_{-m-1}.
  \label{eq:single-residue-one-variable-explicit}
\end{align}
Thus
\begin{align}
  &\operatorname*{Res}_{y=g^{-1}}
  \left[\varepsilon^{2m+1}F(y)G(y)\,dy\right]
  \notag\\
  &
  =(-g)
  \operatorname*{Res}_{y=g^{-1}}\left[\varepsilon^mF(y)\,dy\right]
  \operatorname*{Res}_{y=g^{-1}}\left[\varepsilon^mG(y)\,dy\right].
  \label{eq:residue-product-factor-minus-g}
\end{align}
Applying \eqref{eq:residue-product-factor-minus-g} independently to each
\(\gamma\in\Ibar\setminus\{a,b\}\) produces the scalar
\begin{equation}
  \prod_{\gamma\in\Ibar\setminus\{a,b\}}(-g_\gamma).
  \label{eq:common-product-factor-QQ}
\end{equation}
The scalar in \eqref{eq:common-product-factor-QQ} appears in the residue of each of the three terms in \eqref{eq:Hirota-specialization-for-QQ}.

We now evaluate \(\mathcal R_{I;a,b}\) on each term of
\eqref{eq:Hirota-specialization-for-QQ}.  The regular factors \(y_a-y_b\),
\(y_a\), and \(y_b\) are evaluated at \(y_a=g_a^{-1}\) and
\(y_b=g_b^{-1}\).  For the left-hand side of
\eqref{eq:Hirota-specialization-for-QQ},
\begin{align}
  \tau_{u+1}(\ws+[y_a]+[y_b])
  &=\mathcal T(u+1;\mathbf y_{\Ibar}),
  \\
  \tau_u(\ws)
  &=\mathcal T(u;\mathbf y_{\Ibar\setminus\{a,b\}}).
  \label{eq:first-product-Tsets-QQ-proof}
\end{align}
Using \eqref{eq:Qres-as-Laurent-coefficient} and
\eqref{eq:residue-product-factor-minus-g}, we get
\begin{align}
  &\mathcal R_{I;a,b}\!\left[
  (y_a-y_b)
  \tau_{u+1}(\ws+[y_a]+[y_b])\tau_u(\ws)
  \right]
  \notag\\
  &\quad=
  \left(\prod_{\gamma\in\Ibar\setminus\{a,b\}}(-g_\gamma)\right)
  (g_a^{-1}-g_b^{-1})
  \mathscr Q_I^{\mathrm{res}}(u+1)
  \mathscr Q_{I\cup\{a,b\}}^{\mathrm{res}}(u).
  \label{eq:first-product-residue-QQ-proof}
\end{align}
For the term \(y_a\tau_{u+1}(\ws+[y_a])\tau_u(\ws+[y_b])\) on the right-hand side of
\eqref{eq:Hirota-specialization-for-QQ},
\begin{align}
  \tau_{u+1}(\ws+[y_a])
  &=\mathcal T(u+1;\mathbf y_{\Ibar\setminus\{b\}}),
  \\
  \tau_u(\ws+[y_b])
  &=\mathcal T(u;\mathbf y_{\Ibar\setminus\{a\}}).
  \label{eq:second-product-Tsets-QQ-proof}
\end{align}
Therefore
\begin{align}
  &\mathcal R_{I;a,b}\!\left[
  y_a\,
  \tau_{u+1}(\ws+[y_a])\tau_u(\ws+[y_b])
  \right]
  \notag\\
  &\quad=
  \left(\prod_{\gamma\in\Ibar\setminus\{a,b\}}(-g_\gamma)\right)
  g_a^{-1}
  \mathscr Q_{I\cup\{b\}}^{\mathrm{res}}(u+1)
  \mathscr Q_{I\cup\{a\}}^{\mathrm{res}}(u).
  \label{eq:second-product-residue-QQ-proof}
\end{align}
For the term \(y_b\tau_{u+1}(\ws+[y_b])\tau_u(\ws+[y_a])\) on the right-hand side of
\eqref{eq:Hirota-specialization-for-QQ},
\begin{align}
  \tau_{u+1}(\ws+[y_b])
  &=\mathcal T(u+1;\mathbf y_{\Ibar\setminus\{a\}}),
  \\
  \tau_u(\ws+[y_a])
  &=\mathcal T(u;\mathbf y_{\Ibar\setminus\{b\}}).
  \label{eq:third-product-Tsets-QQ-proof}
\end{align}
Therefore
\begin{align}
  &\mathcal R_{I;a,b}\!\left[
  y_b\,
  \tau_{u+1}(\ws+[y_b])\tau_u(\ws+[y_a])
  \right]
  \notag\\
  &\quad=
  \left(\prod_{\gamma\in\Ibar\setminus\{a,b\}}(-g_\gamma)\right)
  g_b^{-1}
  \mathscr Q_{I\cup\{a\}}^{\mathrm{res}}(u+1)
  \mathscr Q_{I\cup\{b\}}^{\mathrm{res}}(u).
  \label{eq:third-product-residue-QQ-proof}
\end{align}
Apply \(\mathcal R_{I;a,b}\) to both sides of
\eqref{eq:Hirota-specialization-for-QQ}.  Substituting
\eqref{eq:first-product-residue-QQ-proof}--\eqref{eq:third-product-residue-QQ-proof}
and cancelling the scalar \eqref{eq:common-product-factor-QQ} gives
\begin{align}
  &(g_a^{-1}-g_b^{-1})
  \mathscr Q_I^{\mathrm{res}}(u+1)
  \mathscr Q_{I\cup\{a,b\}}^{\mathrm{res}}(u)
  \notag\\
  &\qquad=
  g_a^{-1}
  \mathscr Q_{I\cup\{b\}}^{\mathrm{res}}(u+1)
  \mathscr Q_{I\cup\{a\}}^{\mathrm{res}}(u)
  -g_b^{-1}
  \mathscr Q_{I\cup\{a\}}^{\mathrm{res}}(u+1)
  \mathscr Q_{I\cup\{b\}}^{\mathrm{res}}(u).
  \label{eq:QQ-before-multiplying-g}
\end{align}
Multiplying \eqref{eq:QQ-before-multiplying-g} by \(-g_ag_b\) gives
\eqref{eq:QQ-relation-residue-normalized}.
\end{proof}

\begin{remark}
The normalized operators \(\widehat Q_J(u)\) defined by
\eqref{eq:Qhat-normalized-def} satisfy the same relation as
\eqref{eq:QQ-relation-residue-normalized}, with each
\(\mathscr Q_J^{\mathrm{res}}\) replaced by \(\widehat Q_J\).  Indeed,
using \eqref{eq:Qres-to-Q-normalization}, one checks that
\begin{equation}
  \frac{\mathcal N_I\chi_{I,\overline I}\,
        \mathcal N_{I\cup\{a,b\}}\chi_{I\cup\{a,b\},\overline{I\cup\{a,b\}}}}
       {\mathcal N_{I\cup\{a\}}\chi_{I\cup\{a\},\overline{I\cup\{a\}}}\,
        \mathcal N_{I\cup\{b\}}\chi_{I\cup\{b\},\overline{I\cup\{b\}}}}
  =1.
  \label{eq:QQ-Qhat-normalization-ratio}
\end{equation}
For the unnormalized operators \(Q_J(u)\) in \eqref{eq:Sch-Q-def}, the same
substitution gives
\begin{align}
  &(g_a-g_b)
  (1-g_b/g_a)^{-1}(1-g_a/g_b)^{-1}
  Q_I(u+1)Q_{I\cup\{a,b\}}(u)
  \notag\\
  &\qquad=
  g_a\,Q_{I\cup\{a\}}(u+1)Q_{I\cup\{b\}}(u)
  -g_b\,Q_{I\cup\{b\}}(u+1)Q_{I\cup\{a\}}(u).
  \label{eq:QQ-relation-Q-normalization}
\end{align}
Thus the extra scalar factor in \eqref{eq:QQ-relation-Q-normalization} is only
a consequence of the normalization of \(Q_J(u)\).
\end{remark}

\section{Outlook}

In this paper, we defined the master $T$-operator as a trace over the Fock representation of Schwinger bosons and clarified the relation between two existing constructions of Baxter $Q$-operators.  
Although we demonstrated our method only for the simple toy model of a rational $\gl(M)$ spin chain, we expect that it can be applied and extended to more complicated models and problems.  
Among the various possible directions, our tentative goals include
the following.

The first direction is a Lie-algebraic extension, in particular to
superalgebras.  In view of the discussion in Ref.~\cite{KLT2012}, the
extension to $\gl(M|N)$ should be straightforward.  Degenerate Yangian $L$-operators in oscillator form and the corresponding $Q$-operators are known \cite{FLMS2011}, and the conventional master $T$-operator with its mKP interpretation has also been developed \cite{TsuboiZabrodinZotov2015}; what remains here is their realization within the present Schwinger-oscillator trace framework.  For $\mathfrak{osp}(M|2s)$, if the corresponding $L$-operator admits a representation in terms of Schwinger oscillators, a similar construction should also be possible.  
Recent constructions of degenerate Yangian \(L\)-operators in oscillator form of orthosymplectic type \cite{FrassekTsymbaliuk2024} provide concrete objects with which to compare the present approach.  Moreover, transfer-matrix eigenvalues and Baxter
$Q$-functions for $\mathfrak{osp}(M|2s)$ can be obtained by a certain folding of the corresponding results for $\gl(M|N)$ \cite{Tsuboi2024}.  We therefore expect that an analogous relation can be established at the operator level for quantities defined by traces, such as the $Q$-operators and $T$-operators discussed in this paper.

Second, the conventional master $T$-operator for trigonometric $R$-matrices is defined by a Schur function expansion and interpreted as a tau-function in Ref.~\cite{Zabrodin2013Trig}.  Its direct realization as a trace over a $q$-deformed Schwinger-boson representation should also be possible because $q$-deformations of the Schwinger-boson realization are known.  A relevant $q$-oscillator realization for quantum superalgebras is described in Ref.~\cite{FSV1991}. 
In particular, the large-occupation-number limit of
\(q\)-deformed Schwinger oscillators should give the \(L\)-operators that
enter the construction of \(Q\)-operators for \(U_q(\gl(M|N))\).
We expect that this viewpoint can also reformulate the corresponding
\(L\)-operators for super spin chains obtained in our previous papers
\cite{BazhanovTsuboi2008,Tsuboi2014,Tsuboi2019}. 
In representation-theoretic terms, the oscillator representations
underlying these \(L\)-operators may be regarded as superalgebraic
analogues of the representations appearing in
\cite{BLZ1999,BazhanovHibberdKhoroshkin2002,Kojima2008}.  They are also related to asymptotic representations of quantum affine (super)algebras
\cite{HernandezJimbo2012,Zhang2017Asymptotic}.
The relation between \(L\)-operators constructed from such representations and the universal \(R\)-matrix is explained, 
for example, in \cite{BoosGohmannKlumperNirovRazumov2010}.

Third, in the trigonometric setting, transfer matrices constructed from solutions of the tetrahedron equation may also be interpreted as a kind of master $T$-operator.  
Indeed, the $q$-oscillator solution of the tetrahedron equation and the transfer matrices derived from it in Ref.~\cite{BazhanovSergeev2006} have a mathematical structure similar to that of our master $T$-operator constructed from Schwinger bosons.

These three topics are closely interrelated, and one of the motivations
of this paper is to establish the methodology needed to address them.
We will discuss the details elsewhere.

\section*{Acknowledgements}
We thank the referee for useful comments.

\paragraph{\bf Declaration of generative AI use}
During the preparation of this work the author used ChatGPT in order to write an initial draft of the manuscript based on the author's original ideas through an iterative process of direct edits and prompt instructions. After using this tool/service, the author reviewed and edited the content as needed and takes full responsibility for the content of the published article.

\appendix

\section{Trace formulas for a single bosonic oscillator}
\label{app:direct-trace-asymptotics}

This appendix proves Lemmas~\ref{lem:one-critical-trace} and
\ref{lem:regular-trace-away-from-one}.  Both results follow from the same
trace formula.

Let \(\mathsf A,\mathsf A^{\dagger}\) be a single bosonic oscillator,
\([\mathsf A,\mathsf A^{\dagger}]=1\), and let
\(\mathsf N=\mathsf A^{\dagger}\mathsf A\).  Consider a finite product
\(U\) containing \(r\) factors \(\mathsf A^{\dagger}\) and \(s\) factors
\(\mathsf A\).  Repeated use of
\(\mathsf A\mathsf A^{\dagger}=\mathsf A^{\dagger}\mathsf A+1\) gives
\begin{equation}
  U=
  \sum_{j=0}^{\min\{r,s\}}
  c_j(\mathsf A^{\dagger})^{r-j}\mathsf A^{s-j},
  \qquad c_0=1.
  \label{eq:normal-order-one-product}
\end{equation}
The coefficients \(c_j\) depend on the order of the factors in \(U\).

For \(|\rho|<1\), summing the diagonal matrix elements in the
occupation-number basis gives
\begin{align}
  \Tr\left[
  \rho^{\mathsf N}(\mathsf A^{\dagger})^p\mathsf A^q
  \right]
  &=\delta_{pq}
  \sum_{n\ge p}n(n-1)\cdots(n-p+1)\rho^n
  \notag\\
  &=\delta_{pq}\,p!\,
  \frac{\rho^p}{(1-\rho)^{p+1}}.
  \label{eq:regular-trace-normal-monomial}
\end{align}
Combining this identity with \eqref{eq:normal-order-one-product}, we obtain
\begin{equation}
  \Tr\left[\rho^{\mathsf N}U\right]
  =
  \begin{cases}
  \displaystyle
  \sum_{j=0}^{r}
  c_j(r-j)!\,
  \dfrac{\rho^{r-j}}{(1-\rho)^{r-j+1}},
  & r=s,\\[4mm]
  0,& r\ne s.
  \end{cases}
  \label{eq:trace-general-one-product}
\end{equation}
For \(|\rho|<1\), these are ordinary convergent trace identities.  The
right-hand sides extend meromorphically in \(\rho\); their Taylor expansions
at \(\rho=0\) are the formal traces obtained by summing the diagonal matrix
elements.

The position of the diagonal factor can be changed without invoking
cyclicity.  Indeed,
\([\mathsf N,U]=(r-s)U\).  If \(r\ne s\), all diagonal matrix elements of
both \(\rho^{\mathsf N}U\) and \(U\rho^{\mathsf N}\) are zero.  If
\(r=s\), then \(U\) commutes with \(\rho^{\mathsf N}\).  Hence, in either
case,
\begin{equation}
  \Tr\left[\rho^{\mathsf N}U\right]
  =\Tr\left[U\rho^{\mathsf N}\right].
  \label{eq:trace-two-orders-one-product}
\end{equation}
This equality follows directly from diagonal matrix elements, not from a
cyclicity property of the formal trace.

\subsection*{Proof of Lemma~\ref{lem:one-critical-trace}}

Let \(U\) be obtained from \(W\) by removing the scalar factor
\(\sqrt\varepsilon\) from each creation and annihilation operator.  Then
\begin{equation}
  W=\varepsilon^{(r+s)/2}U.
  \label{eq:scaled-product-unscaled-product}
\end{equation}
If \(r\ne s\), both traces vanish by
\eqref{eq:trace-general-one-product} and
\eqref{eq:trace-two-orders-one-product}.  Suppose that \(r=s\).  Setting
\(\rho=1-\varepsilon\) in \eqref{eq:trace-general-one-product}, we find
\begin{align}
  \varepsilon
  \Tr\left[(1-\varepsilon)^{\mathsf N}W\right]
  &=
  \sum_{j=0}^{r}
  c_j(r-j)!(1-\varepsilon)^{r-j}\varepsilon^j
  \notag\\
  &\longrightarrow r!.
  \label{eq:scaled-product-limit}
\end{align}
As \(\varepsilon\to0\), only the term \(j=0\) survives, and \(c_0=1\).
Equation~\eqref{eq:trace-two-orders-one-product} gives the same limit when
\((1-\varepsilon)^{\mathsf N}\) is placed on the right of \(W\).  This
proves Lemma~\ref{lem:one-critical-trace}.

\subsection*{Proof of Lemma~\ref{lem:regular-trace-away-from-one}}

Equation~\eqref{eq:trace-general-one-product} expresses the first trace as
a finite sum of rational functions whose denominators are powers of
\(1-\rho(\varepsilon)\).  Since
\(1-\rho(0)=1-\lambda\ne0\), every term is regular at
\(\varepsilon=0\).  The second trace has the same diagonal matrix elements
by \eqref{eq:trace-two-orders-one-product}.

We shall also use the traces of powers of the number operator.  For
\(q\ge0\),
\begin{align}
 A_{q}:= \Tr\left[\rho^{\mathsf N}\mathsf N^q\right]
  =\sum_{n\ge0}n^q\rho^n
  =\left(\rho\frac{d}{d\rho}\right)^q\frac{1}{1-\rho},
  \label{eq:Aq-trace}
\end{align}
which is a rational function with possible pole only at \(\rho=1\).
The first few cases are
\begin{align}
  A_0&=\Tr(\rho^{\mathsf N})=\frac1{1-\rho},
  \label{eq:A0-trace}\\
  A_1&=\Tr(\rho^{\mathsf N}\mathsf N)=\frac{\rho}{(1-\rho)^2},
  \label{eq:A1-trace}\\
  A_2&=\Tr(\rho^{\mathsf N}\mathsf N^2)=\frac{\rho(1+\rho)}{(1-\rho)^3},
  \label{eq:A2-trace}\\
  A_3&=\Tr(\rho^{\mathsf N}\mathsf N^3)=\frac{\rho(1+4\rho+\rho^2)}{(1-\rho)^4}.
  \label{eq:A3-trace}
\end{align}

\clearpage

\section{Proof of the residue formula}
\label{app:residue-details}

This appendix gives the proof of Proposition~\ref{prop:main-contraction}.
We use the definitions in Sections~\ref{sec:local-contraction} and
\ref{sec:Q-contraction-theorem}. In particular,
\(\widetilde L^{(\Ibar)}\), \(W_{\Ibar}\), and
\(\mathscr L_{I,\Ibar}\) are given in
\eqref{eq:scaled-local-L}, \eqref{eq:Wcrit-simple}, and
\eqref{eq:contracted-L-ac}--\eqref{eq:contracted-L}.

We first expand the matrix elements of \(\widetilde L^{(\Ibar)}\) and
show which terms can be omitted from the normalized trace. We then
factor the products of the retained matrix elements, keeping the order
of the operators associated with each oscillator. The conservation
relations are proved in \ref{app:colour-conservation}. In
\ref{app:full-residue-proof}, we first factor the trace over
\(\F_{B,\Ibar}\) into a product of three traces. We then apply
Lemma~\ref{lem:one-critical-trace} and
Lemma~\ref{lem:regular-trace-away-from-one} to evaluate the limit in
\eqref{eq:HP-definition}.

For the calculations of matrix elements, we use the following notation.

Let
\begin{equation}
  |\mathbf d\rangle=|d_1,\ldots,d_L\rangle
  =e_{d_1}\otimes\cdots\otimes e_{d_L},
  \qquad
  |\mathbf c\rangle=|c_1,\ldots,c_L\rangle
  =e_{c_1}\otimes\cdots\otimes e_{c_L},
\end{equation}
be basis vectors of the quantum space \(\Hq_L=(\C^M)^{\otimes L}\), where
\(\mathbf d=(d_1,\ldots,d_L)\) and \(\mathbf c=(c_1,\ldots,c_L)\) have entries in
\(B\).  For every \(i\in B\), set
\begin{equation}
  m_i(\mathbf d)=\#\{\ell\mid d_\ell=i\},
  \qquad
  m_i(\mathbf c)=\#\{\ell\mid c_\ell=i\}.
  \label{eq:m-beta-words}
\end{equation}
Then
\begin{equation}
  \mathsf M_i|\mathbf d\rangle
  =m_i(\mathbf d)|\mathbf d\rangle,
  \qquad
  \mathsf M_i|\mathbf c\rangle
  =m_i(\mathbf c)|\mathbf c\rangle,
  \qquad i\in B.
  \label{eq:Mbeta-on-word-basis}
\end{equation}

\begin{equation}
  v_\ell=u-\theta_\ell.
\end{equation}

\subsection{Matrix elements of the scaled \texorpdfstring{$L$}{L}-operator}
\label{app:local-matrix-details}
\label{subsec:leading-scaled-local-matrix}

To apply Lemma~\ref{lem:one-critical-trace} to a term of the trace
\eqref{eq:scaled-master-trace-factorization}, we write each occurrence
of \(\mathsf A_\beta\) or \(\mathsf A_\beta^{\dagger}\) together
with its factor \(\sqrt{\varepsilon_\beta}\). For example,
\[
 \varepsilon_\beta\mathsf A_\beta^{\dagger}\mathsf A_\beta
 =(\sqrt{\varepsilon_\beta}\mathsf A_\beta^{\dagger})
  (\sqrt{\varepsilon_\beta}\mathsf A_\beta).
\]
Lemma~\ref{lem:regular-trace-away-from-one} applies to the generators
\(\mathsf a_{c\beta}\) with \(c\ne\beta\). The following
expansions keep the two types of factors explicit. Their use in the
product over all sites is justified in \ref{app:remainder-details}.

If \(a,c\in I\), then \(s_a=s_c=1\) and, by \eqref{eq:scaled-component-detailed},
\begin{align}
  [\widetilde L^{(\Ibar)}(v)]_{ac}
  &=v\delta_{ac}+
  \sum_{\beta\in\Ibar}\mathsf a_{c\beta}^{\dagger}\mathsf a_{a\beta}.
  \label{eq:AA-block-detailed}
\end{align}
All terms in \eqref{eq:AA-block-detailed} are retained in the leading
part defined below.

If \(a\in I\) and \(\beta\in\Ibar\), then
\begin{align}
  [\widetilde L^{(\Ibar)}(v)]_{a\beta}
  &=\sqrt{\varepsilon_\beta}
  \sum_{\gamma\in\Ibar}\mathsf a_{\beta\gamma}^{\dagger}\mathsf a_{a\gamma}
  \notag\\
  &=(\sqrt{\varepsilon_\beta}\mathsf A_\beta^{\dagger})\mathsf a_{a\beta}
    +\sqrt{\varepsilon_\beta}
    \sum_{\substack{\gamma\in\Ibar\\ \gamma\ne\beta}}
    \mathsf a_{\beta\gamma}^{\dagger}\mathsf a_{a\gamma}.
  \label{eq:a-beta-block-detailed}
\end{align}
The first term contains
\(\sqrt{\varepsilon_\beta}\mathsf A_\beta^{\dagger}\).  In the second
term, \(\gamma\ne\beta\), and the oscillators have labels
\((\beta,\gamma)\) and \((a,\gamma)\).  The factors in
\eqref{eq:Wcrit-simple} involving their number operators are
\[
 \left(\frac{g_\beta}{g_\gamma}(1-\varepsilon_\gamma)\right)^{\mathsf N_{\beta\gamma}},
 \qquad
 \left(\frac{g_a}{g_\gamma}(1-\varepsilon_\gamma)\right)^{\mathsf N_{a\gamma}}.
\]
The bases of these powers tend to \(g_\beta/g_\gamma\) and
\(g_a/g_\gamma\), and neither limit is \(1\).  By
Lemma~\ref{lem:regular-trace-away-from-one}, the traces involving these
oscillators are regular as \(\varepsilon_\gamma\to0\); the explicit factor
\(\sqrt{\varepsilon_\beta}\) therefore prevents the second term from
contributing to the highest-order pole.

If \(\beta\in\Ibar\) and \(a\in I\), then
\begin{align}
  [\widetilde L^{(\Ibar)}(v)]_{\beta a}
  &=\sqrt{\varepsilon_\beta}
  \sum_{\gamma\in\Ibar}\mathsf a_{a\gamma}^{\dagger}\mathsf a_{\beta\gamma}
  \notag\\
  &=\mathsf a_{a\beta}^{\dagger}(\sqrt{\varepsilon_\beta}\mathsf A_\beta)
    +\sqrt{\varepsilon_\beta}
    \sum_{\substack{\gamma\in\Ibar\\ \gamma\ne\beta}}
    \mathsf a_{a\gamma}^{\dagger}\mathsf a_{\beta\gamma}.
  \label{eq:beta-a-block-detailed}
\end{align}
The first term contains
\(\sqrt{\varepsilon_\beta}\mathsf A_\beta\).  In the second term,
\(\gamma\ne\beta\), and the oscillators have labels \((a,\gamma)\) and
\((\beta,\gamma)\).  The factors in \eqref{eq:Wcrit-simple} involving their
number operators are
\[
 \left(\frac{g_a}{g_\gamma}(1-\varepsilon_\gamma)\right)^{\mathsf N_{a\gamma}},
 \qquad
 \left(\frac{g_\beta}{g_\gamma}(1-\varepsilon_\gamma)\right)^{\mathsf N_{\beta\gamma}}.
\]
The bases of these powers tend to \(g_a/g_\gamma\) and
\(g_\beta/g_\gamma\), and neither limit is \(1\).  By
Lemma~\ref{lem:regular-trace-away-from-one}, the traces involving these
oscillators are regular as \(\varepsilon_\gamma\to0\); the explicit factor
\(\sqrt{\varepsilon_\beta}\) therefore prevents the second term from
contributing to the highest-order pole.

If \(\beta\in\Ibar\), then
\begin{align}
  [\widetilde L^{(\Ibar)}(v)]_{\beta\beta}
  &=\varepsilon_\beta\left(v+
  \sum_{\delta\in\Ibar}\mathsf a_{\beta\delta}^{\dagger}\mathsf a_{\beta\delta}
  \right)
  \notag\\
  &=\varepsilon_\beta \mathsf A_\beta^{\dagger}\mathsf A_\beta
    +\varepsilon_\beta v
    +\varepsilon_\beta
    \sum_{\substack{\delta\in\Ibar\\ \delta\ne\beta}}
    \mathsf a_{\beta\delta}^{\dagger}\mathsf a_{\beta\delta}.
  \label{eq:beta-beta-block-detailed}
\end{align}
The first term is
\[
  (\sqrt{\varepsilon_\beta}\mathsf A_\beta^{\dagger})
  (\sqrt{\varepsilon_\beta}\mathsf A_\beta),
\]
and is therefore of the form covered by
Lemma~\ref{lem:one-critical-trace}.  For the second term, the explicit
factor \(\varepsilon_\beta\) remains after this lemma is applied to the
trace over \(\F_{\{\beta\},\{\beta\}}\).  The same is true for each term
in the sum, while Lemma~\ref{lem:regular-trace-away-from-one} shows that
the trace over \(\F_{\{\beta\},\{\delta\}}\), \(\delta\ne\beta\), is
regular.  Hence only the first term can contribute to the highest-order
pole.

Finally, if \(\beta,\gamma\in\Ibar\) and \(\beta\ne\gamma\), then
\begin{equation}
  [\widetilde L^{(\Ibar)}(v)]_{\beta\gamma}
  =\sqrt{\varepsilon_\beta\varepsilon_\gamma}
  \sum_{\delta\in\Ibar}\mathsf a_{\gamma\delta}^{\dagger}\mathsf a_{\beta\delta}.
  \label{eq:beta-gamma-offdiag-block}
\end{equation}
For \(\delta=\beta\) this contains
\begin{equation}
  \sqrt{\varepsilon_\gamma}\,
  \mathsf a_{\gamma\beta}^{\dagger}(\sqrt{\varepsilon_\beta}\mathsf A_\beta).
\end{equation}
After Lemma~\ref{lem:one-critical-trace} is applied to the trace over
\(\F_{\{\beta\},\{\beta\}}\),
Lemma~\ref{lem:regular-trace-away-from-one} shows that the trace over
\(\F_{\{\gamma\},\{\beta\}}\) is regular; the factor
\(\sqrt{\varepsilon_\gamma}\) therefore remains.  For
\(\delta=\gamma\), the corresponding term is
\begin{equation}
  \sqrt{\varepsilon_\beta}
  (\sqrt{\varepsilon_\gamma}\mathsf A_\gamma^{\dagger})\mathsf a_{\beta\gamma}.
\end{equation}

Lemma~\ref{lem:one-critical-trace} applies to the trace over
\(\F_{\{\gamma\},\{\gamma\}}\), and
Lemma~\ref{lem:regular-trace-away-from-one} applies to the trace over
\(\F_{\{\beta\},\{\gamma\}}\). The factor
\(\sqrt{\varepsilon_\beta}\) remains after these traces have been
evaluated.

  If
\(\delta\ne\beta,\gamma\), the traces over
\(\F_{\{\gamma\},\{\delta\}}\) and
\(\F_{\{\beta\},\{\delta\}}\) are regular by
Lemma~\ref{lem:regular-trace-away-from-one}, while the factor
\(\sqrt{\varepsilon_\beta\varepsilon_\gamma}\) remains.  Thus every
term in \eqref{eq:beta-gamma-offdiag-block} still contains a positive
power of at least one \(\varepsilon_\alpha\) after the traces are
evaluated and does not contribute to the highest-order pole.

\subsection{Terms that do not contribute to the normalized trace}
\label{app:remainder-details}

The calculations in \ref{app:local-matrix-details} distinguish the
following terms in the matrix elements of
\(\widetilde L^{(\Ibar)}(v)\). We denote the matrix formed from them
by \(\widetilde L^{\rm lead}(v)\); the difference from
\(\widetilde L^{(\Ibar)}(v)\) will be denoted by \(\Delta L(v)\).
We shall verify below that replacing
\(\widetilde L^{(\Ibar)}\) by \(\widetilde L^{\rm lead}\) in the
product over all sites does not change the normalized limit in
\eqref{eq:HP-definition}.

\begin{align}
  [\widetilde L^{\rm lead}(v)]_{ac}
  &=v\delta_{ac}+\sum_{\beta\in\Ibar}\mathsf a_{c\beta}^{\dagger}\mathsf a_{a\beta},
  &&a,c\in I,
  \label{eq:lead-AA}
  \\
  [\widetilde L^{\rm lead}(v)]_{a\beta}
  &=(\sqrt{\varepsilon_\beta}\mathsf A_\beta^{\dagger})\mathsf a_{a\beta},
  &&a\in I,
  \beta\in\Ibar,
  \label{eq:lead-Abar}
  \\
  [\widetilde L^{\rm lead}(v)]_{\beta a}
  &=\mathsf a_{a\beta}^{\dagger}(\sqrt{\varepsilon_\beta}\mathsf A_\beta),
  &&\beta\in\Ibar,
  a\in I,
  \label{eq:lead-barA}
  \\
  [\widetilde L^{\rm lead}(v)]_{\beta\gamma}
  &=\delta_{\beta\gamma}\varepsilon_\beta \mathsf A_\beta^{\dagger}\mathsf A_\beta,
  &&\beta,
  \gamma\in\Ibar.
  \label{eq:lead-barbar}
\end{align}

Set
\begin{equation}
  [\Delta L(v)]_{cd}:=
  [\widetilde L^{(\Ibar)}(v)]_{cd}
  -[\widetilde L^{\rm lead}(v)]_{cd}.
  \label{eq:Delta-L-definition}
\end{equation}

Then \eqref{eq:AA-block-detailed}--\eqref{eq:beta-gamma-offdiag-block} and
\eqref{eq:lead-AA}--\eqref{eq:lead-barbar} give
\begin{align}
  [\Delta L(v)]_{ac}&=0,
  &&a,c\in I,
  \label{eq:Delta-L-ac}
  \\
  [\Delta L(v)]_{a\beta}
  &=\sqrt{\varepsilon_\beta}
    \sum_{\substack{\gamma\in\Ibar\\ \gamma\ne\beta}}
    \mathsf a_{\beta\gamma}^{\dagger}\mathsf a_{a\gamma},
  &&a\in I,\ \beta\in\Ibar,
  \label{eq:Delta-L-a-beta}
  \\
  [\Delta L(v)]_{\beta a}
  &=\sqrt{\varepsilon_\beta}
    \sum_{\substack{\gamma\in\Ibar\\ \gamma\ne\beta}}
    \mathsf a_{a\gamma}^{\dagger}\mathsf a_{\beta\gamma},
  &&\beta\in\Ibar,\ a\in I,
  \label{eq:Delta-L-beta-a}
  \\
  [\Delta L(v)]_{\beta\beta}
  &=\varepsilon_\beta v
    +\varepsilon_\beta
    \sum_{\substack{\delta\in\Ibar\\ \delta\ne\beta}}
    \mathsf a_{\beta\delta}^{\dagger}\mathsf a_{\beta\delta},
  &&\beta\in\Ibar,
  \label{eq:Delta-L-beta-beta}
  \\
  [\Delta L(v)]_{\beta\gamma}
  &=\sqrt{\varepsilon_\beta\varepsilon_\gamma}
  \sum_{\delta\in\Ibar}
  \mathsf a_{\gamma\delta}^{\dagger}\mathsf a_{\beta\delta},
  &&\beta,\gamma\in\Ibar,\ \beta\ne\gamma.
  \label{eq:Delta-L-beta-gamma}
\end{align}

Substituting \(\widetilde L^{(\Ibar)}=\widetilde L^{\rm lead}+\Delta L\)
into
\[
 [\widetilde L^{(\Ibar)}(v_L)]_{c_Ld_L}\cdots
 [\widetilde L^{(\Ibar)}(v_1)]_{c_1d_1}
\]
gives the product of the leading matrix elements and terms containing
at least one factor \([\Delta L(v_\ell)]_{c_\ell d_\ell}\).

 Every summand in
Eqs.~\eqref{eq:Delta-L-a-beta}--\eqref{eq:Delta-L-beta-gamma} is a product
of \(\sqrt{\varepsilon_\beta}\mathsf A_\beta\),
\(\sqrt{\varepsilon_\beta}\mathsf A_\beta^{\dagger}\), and generators
\(\mathsf a_{c\beta},\mathsf a_{c\beta}^{\dagger}\) with \(c\ne\beta\),
with an additional factor \(\sqrt{\varepsilon_\alpha}\) or
\(\varepsilon_\alpha\) for at least one \(\alpha\).  These additional
factors remain when the matrix elements of several \(L\)-operators are
multiplied.
Lemmas~\ref{lem:one-critical-trace} and
\ref{lem:regular-trace-away-from-one}, applied to the factors in
Eq.~\eqref{eq:scaled-master-trace-factorization}, then show that a summand
containing at least one factor \([\Delta L(v_\ell)]_{c_\ell d_\ell}\)
does not contribute to the highest-order pole of the master
\(T\)-operator.  Hence only the matrix elements
\eqref{eq:lead-AA}--\eqref{eq:lead-barbar} need be retained.

In particular, the difference between the two products satisfies
\begin{align}
 &\lim_{\varepsilon_\beta\to0\;(\beta\in\Ibar)}
 \left(\prod_{\beta\in\Ibar}\varepsilon_\beta\right)
 \Tr_{\F_{B,\Ibar}}\Bigl[
 \bigl(
 \widetilde L^{(\Ibar)}_{0L}(v_L)\cdots
 \widetilde L^{(\Ibar)}_{01}(v_1)
 \notag\\[-1mm]
 &\hspace{37mm}{}-
 \widetilde L^{\rm lead}_{0L}(v_L)\cdots
 \widetilde L^{\rm lead}_{01}(v_1)
 \bigr) W_{\Ibar}(\mathbf y(\varepsilon))\Bigr]=0.
 \label{eq:leading-replacement-in-trace}
\end{align}

\subsection{Factorization of the products of matrix elements}
\label{app:ordered-factorization}

We compare \(\widetilde L^{\rm lead}\), given in
\eqref{eq:lead-AA}--\eqref{eq:lead-barbar}, with the operator
\(\mathscr L_{I,\Ibar}\) defined in
\eqref{eq:contracted-L-ac}--\eqref{eq:contracted-L-components}.
The factors involving \(\mathsf A_\beta\) and
\(\mathsf A_\beta^{\dagger}\) are denoted by
\(\mathcal K_{cd}^{(\varepsilon)}\) and are defined by

\begin{align}
  \mathcal K_{ac}^{(\varepsilon)}&=1,
  &&a,c\in I,
  \label{eq:K-AA}
  \\
  \mathcal K_{a\beta}^{(\varepsilon)}&=\sqrt{\varepsilon_\beta}\mathsf A_\beta^{\dagger},
  &&a\in I,
  \beta\in\Ibar,
  \label{eq:K-Abar}
  \\
  \mathcal K_{\beta a}^{(\varepsilon)}&=\sqrt{\varepsilon_\beta}\mathsf A_\beta,
  &&\beta\in\Ibar,
  a\in I,
  \label{eq:K-barA}
  \\
  \mathcal K_{\beta\gamma}^{(\varepsilon)}&=\delta_{\beta\gamma}\varepsilon_\beta \mathsf A_\beta^{\dagger}\mathsf A_\beta,
  &&\beta,
  \gamma\in\Ibar.
  \label{eq:K-barbar}
\end{align}
Then \eqref{eq:lead-AA}--\eqref{eq:lead-barbar} and
\eqref{eq:contracted-L-ac}--\eqref{eq:contracted-L-components} give
\begin{equation}
  [\widetilde L^{\rm lead}(v)]_{cd}
  =\mathcal K_{cd}^{(\varepsilon)}
  [\mathscr L_{I,\Ibar}(v)]_{cd}.
  \label{eq:lead-KL-factorization}
\end{equation}
The operator \(\mathcal K_{cd}^{(\varepsilon)}\) is built from
\(\mathsf A_\beta\) and \(\mathsf A_\beta^{\dagger}\),
\(\beta\in\Ibar\), whereas
\([\mathscr L_{I,\Ibar}(w)]_{ef}\) is built from
\(\mathsf a_{a\beta}\) and \(\mathsf a_{a\beta}^{\dagger}\) with
\(a\in I\) and \(\beta\in\Ibar\), or is a scalar.  Since
\(I\cap\Ibar=\varnothing\), the canonical commutation
relations give
\begin{equation}
  [\mathsf A_\beta,\mathsf a_{a\gamma}]=[\mathsf A_\beta,\mathsf a_{a\gamma}^{\dagger}]
  =[\mathsf A_\beta^{\dagger},\mathsf a_{a\gamma}]
  =[\mathsf A_\beta^{\dagger},\mathsf a_{a\gamma}^{\dagger}]=0
  \quad(a\in I,\ \beta,\gamma\in\Ibar).
  \label{eq:K-commutes-L-basic}
\end{equation}
Consequently
\begin{equation}
  [\mathcal K_{cd}^{(\varepsilon)},[\mathscr L_{I,\Ibar}(w)]_{ef}]=0
  \quad(c,d,e,f\in B).
  \label{eq:K-commutes-L}
\end{equation}

For the basis vectors \(|\mathbf c\rangle\) and
\(|\mathbf d\rangle\) specified at the beginning of this appendix,
using \eqref{eq:lead-KL-factorization} and \eqref{eq:K-commutes-L}, we obtain
\begin{align}
&[\widetilde L^{\rm lead}(v_L)]_{c_Ld_L}\cdots
  [\widetilde L^{\rm lead}(v_1)]_{c_1d_1}
\notag\\*
&=
  \mathcal K_{c_Ld_L}^{(\varepsilon)}\cdots
  \mathcal K_{c_1d_1}^{(\varepsilon)}
  [\mathscr L_{I,\Ibar}(v_L)]_{c_Ld_L}\cdots
  [\mathscr L_{I,\Ibar}(v_1)]_{c_1d_1}
  \label{eq:leading-product-factorization}\\*
&=\mathcal K_{c_Ld_L}^{(\varepsilon)}\cdots
  \mathcal K_{c_1d_1}^{(\varepsilon)}
  \left[
  \mathscr L_{I,\Ibar,L}(v_L)\cdots
  \mathscr L_{I,\Ibar,1}(v_1)
  \right]_{\mathbf c,\mathbf d}.
  \label{eq:contracted-monodromy-matrix-element-from-L}
\end{align}

\subsection{Conservation of the number of spins of each colour}
\label{app:colour-conservation}
The operators \(\mathsf M_i\) are defined in \eqref{eq:Mbeta-count}.
The vectors \(|\mathbf c\rangle\), \(|\mathbf d\rangle\) and the numbers
\(m_i(\mathbf c),m_i(\mathbf d)\) are defined at the beginning of this appendix.

We prove \eqref{eq:Q-commutes-Mbeta} and \eqref{eq:commute-Mbeta-detailed}.  For
\(\beta\in\Ibar\), set
\[
  \mathsf N^I_\beta:=\sum_{a\in I}\mathsf N_{a\beta}.
\]
Equation \eqref{eq:contracted-L} gives
\[
 \left[
 -\mathsf N^I_\beta\otimes\Id
 +\Id\otimes E_{\beta\beta},
 \mathscr L_{I,\Ibar}(v)
 \right]=0,
 \qquad \beta\in\Ibar.
\]
Therefore
\[
 \left[
 -\mathsf N^I_\beta\otimes\Id
 +\Id\otimes\mathsf M_\beta,
 \mathscr L_{I,\Ibar,L}(u-\theta_L)\cdots
 \mathscr L_{I,\Ibar,1}(u-\theta_1)
 \right]=0.
\]
Since \(D_{I,\Ibar}\) commutes with \(\mathsf N^I_\beta\), taking the trace
of the preceding relation gives
Eq.~\eqref{eq:Q-commutes-Mbeta}.
Taking the matrix element of Eq.~\eqref{eq:Q-commutes-Mbeta} between
\(|\mathbf c\rangle\) and \(|\mathbf d\rangle\), and using
Eqs.~\eqref{eq:m-beta-words}, \eqref{eq:Mbeta-on-word-basis}, gives
\begin{align*}
  0
  &=\langle\mathbf c|[Q_I(u),\mathsf M_\beta]|\mathbf d\rangle \\
  &=\bigl(m_\beta(\mathbf d)-m_\beta(\mathbf c)\bigr)
    [Q_I(u)]_{\mathbf c,\mathbf d}.
\end{align*}
Thus, imposing this relation for every \(\beta\in\Ibar\), one obtains
\begin{equation}
  [Q_I(u)]_{\mathbf c,\mathbf d}=0
  \quad\text{unless}\quad
  m_\beta(\mathbf c)=m_\beta(\mathbf d)
  \quad\text{for all }\beta\in\Ibar.
  \label{eq:Q-preserves-Mbeta}
\end{equation}

For \(L^{(\Ibar)}(v)\), fix \(i\in B\) and set
\[
  \mathsf N^B_i:=\sum_{\gamma\in\Ibar}\mathsf N_{i\gamma}.
\]
Equation \eqref{eq:local-L-component} gives
\[
 \left[
 \mathsf N^B_i\otimes\Id
 +\Id\otimes E_{ii},
 L^{(\Ibar)}(v)
 \right]=0,
 \qquad i\in B.
\]
The same argument, applied to the product of the \(L^{(\Ibar)}\)-operators
and to the trace defining \(\mathcal T(u;\mathbf y_{\Ibar})\), gives
Eq.~\eqref{eq:commute-Mbeta-detailed}.

\subsection{Proof of Proposition~\ref{prop:main-contraction}}
\label{app:full-residue-proof}
\begin{proof}

We prove \eqref{eq:main-contraction-formula} by evaluating its matrix
elements between the vectors \(|\mathbf c\rangle\) and
\(|\mathbf d\rangle\) introduced at the beginning of this appendix.
The definition \eqref{eq:HP-definition} gives

\begin{align}
&\left[\operatorname{HP}_{\Ibar}\mathcal T(u;\mathbf y_{\Ibar})\right]_{\mathbf c,\mathbf d}
\notag\\
&=
\lim_{\varepsilon_\beta\to0\;(\beta\in\Ibar)}
\left(\prod_{\beta\in\Ibar}\varepsilon_\beta\right)
\Tr_{\F_{B,\Ibar}}
\left[
[\widetilde L^{(\Ibar)}(v_L)]_{c_Ld_L}
\cdots
[\widetilde L^{(\Ibar)}(v_1)]_{c_1d_1}
W_{\Ibar}(\mathbf y(\varepsilon))
\right].
\label{eq:HP-matrix-element-start}
\end{align}
\ref{app:remainder-details} shows that every
term in Eq.~\eqref{eq:HP-matrix-element-start} containing a factor
\([\Delta L(v_\ell)]_{c_\ell d_\ell}\) does not contribute to the limit.
Hence only the matrix elements
\eqref{eq:lead-AA}--\eqref{eq:lead-barbar} need be retained.

If for some site \(\ell\) the labels \(c_\ell,d_\ell\) are distinct
and both lie in \(\Ibar\), then
\([\widetilde L^{\rm lead}(v_\ell)]_{c_\ell d_\ell}=0\) and
\([\mathscr L_{I,\Ibar}(v_\ell)]_{c_\ell d_\ell}=0\).
Replacing each \([\widetilde L^{(\Ibar)}(v_j)]_{c_jd_j}\) in
\eqref{eq:HP-matrix-element-start} by
\([\widetilde L^{\rm lead}(v_j)]_{c_jd_j}\), as justified by
\eqref{eq:leading-replacement-in-trace}, therefore gives zero for the
\((\mathbf c,\mathbf d)\) matrix element on the left hand side of
\eqref{eq:main-contraction-formula}. The definition
\eqref{eq:Sch-Q-def} also gives \([Q_I(u)]_{\mathbf c,\mathbf d}=0\).
We may therefore assume below that no site has two distinct labels
\(c_\ell,d_\ell\in\Ibar\).

Substituting \eqref{eq:leading-product-factorization} and
\eqref{eq:WB-decomposition} into \eqref{eq:HP-matrix-element-start} gives
\begin{align}
&\left[\operatorname{HP}_{\Ibar}\mathcal T(u;\mathbf y_{\Ibar})\right]_{\mathbf c,\mathbf d}
\notag\\
&=\lim_{\varepsilon_\beta\to0\;(\beta\in\Ibar)}
\left(\prod_{\beta\in\Ibar}\varepsilon_\beta\right)
\Tr_{\F_{B,\Ibar}}\Bigl[
\mathcal K_{c_Ld_L}^{(\varepsilon)}\cdots
\mathcal K_{c_1d_1}^{(\varepsilon)}
\notag\\
&\hspace{25mm}{}\times
\left[\mathscr L_{I,\Ibar,L}(v_L)\cdots
\mathscr L_{I,\Ibar,1}(v_1)\right]_{\mathbf c,\mathbf d}
\notag\\
&\hspace{25mm}{}\times W_{\mathrm{sing}}(\varepsilon)
W_{I,\Ibar}(\varepsilon)
W_{\Ibar,\Ibar}^{\mathrm{off}}(\varepsilon)\Bigr].
\label{eq:HP-after-leading-substitution}
\end{align}
The boundary operators in \eqref{eq:HP-after-leading-substitution}
still depend on \(\varepsilon\).
Let \(\F_{\mathrm{cr}}\) be the Fock space generated by
\(\mathsf A_\beta,\mathsf A_\beta^\dagger\), \(\beta\in\Ibar\).
Together with \(\F_{I,\Ibar}\) and \(\F_{\mathrm{int}}\), defined in
\eqref{eq:FJS-notation-fixed} and before \eqref{eq:CIbar-def}, this gives
\[
 \F_{B,\Ibar}\simeq
 \F_{\mathrm{cr}}\otimes\F_{I,\Ibar}\otimes\F_{\mathrm{int}}.
\]
The product \(\mathcal K_{c_Ld_L}^{(\varepsilon)}\cdots
\mathcal K_{c_1d_1}^{(\varepsilon)}W_{\mathrm{sing}}(\varepsilon)\)
acts on \(\F_{\mathrm{cr}}\). The matrix element of the product of
\(\mathscr L\)-operators, multiplied by \(W_{I,\Ibar}(\varepsilon)\),
acts on \(\F_{I,\Ibar}\). The operator
\(W_{\Ibar,\Ibar}^{\mathrm{off}}(\varepsilon)\) acts on
\(\F_{\mathrm{int}}\). We therefore factor the trace over
\(\F_{B,\Ibar}\) into a product of three traces, taken over
\(\F_{\mathrm{cr}}\), \(\F_{I,\Ibar}\), and \(\F_{\mathrm{int}}\),
respectively:
\begin{align}
&\left[\operatorname{HP}_{\Ibar}\mathcal T(u;\mathbf y_{\Ibar})\right]_{\mathbf c,\mathbf d}
\notag\\
&=\lim_{\varepsilon_\beta\to0\;(\beta\in\Ibar)}
\Biggl\{\left(\prod_{\beta\in\Ibar}\varepsilon_\beta\right)
\Tr_{\F_{\mathrm{cr}}}\!\left[
\mathcal K_{c_Ld_L}^{(\varepsilon)}\cdots
\mathcal K_{c_1d_1}^{(\varepsilon)}W_{\mathrm{sing}}(\varepsilon)\right]
\notag\\
&\qquad\times\Tr_{\F_{I,\Ibar}}\!\left[
\left[\mathscr L_{I,\Ibar,L}(v_L)\cdots
\mathscr L_{I,\Ibar,1}(v_1)\right]_{\mathbf c,\mathbf d}
W_{I,\Ibar}(\varepsilon)\right]
\notag\\
&\qquad\times\Tr_{\F_{\mathrm{int}}}
\!\left[W_{\Ibar,\Ibar}^{\mathrm{off}}(\varepsilon)\right]\Biggr\}.
\label{eq:three-trace-factorization}
\end{align}

The first trace in \eqref{eq:three-trace-factorization} is evaluated by
Lemma~\ref{lem:one-critical-trace}.  In the ordered product
\(\mathcal K_{c_Ld_L}^{(\varepsilon)}\cdots
\mathcal K_{c_1d_1}^{(\varepsilon)}\), the sites with \(d_\ell=\beta\)
give \(m_\beta(\mathbf d)\) copies of
\(\sqrt{\varepsilon_\beta}\mathsf A_\beta^\dagger\), while the sites with
\(c_\ell=\beta\) give \(m_\beta(\mathbf c)\) copies of
\(\sqrt{\varepsilon_\beta}\mathsf A_\beta\). 
Hence
\begin{align}
&\lim_{\varepsilon_\beta\to0\;(\beta\in\Ibar)}
\prod_{\beta\in\Ibar}\varepsilon_\beta
\Tr_{\F_{\mathrm{cr}}}
\left[
\mathcal K_{c_Ld_L}^{(\varepsilon)}\cdots
\mathcal K_{c_1d_1}^{(\varepsilon)}
W_{\mathrm{sing}}(\varepsilon)
\right]
\notag\\
&\qquad=
\prod_{\beta\in\Ibar}
\delta_{m_\beta(\mathbf c),m_\beta(\mathbf d)}
\Gamma(m_\beta(\mathbf d)+1).
\label{eq:critical-trace-factor-result}
\end{align}

Each matrix element of \(\mathscr L_{I,\Ibar}(v)\) is a polynomial
in \(\mathsf a_{a\beta},\mathsf a_{a\beta}^{\dagger}\), with
\(a\in I\) and \(\beta\in\Ibar\). Expanding the finite product in
the trace over \(\F_{I,\Ibar}\) in
\eqref{eq:three-trace-factorization} therefore gives a finite sum of
products of traces to which Lemma~\ref{lem:regular-trace-away-from-one}
applies, with \(\lambda=g_a/g_\beta\ne1\).
These traces are regular at \(\varepsilon_\beta=0\).
Since \(W_{I,\Ibar}(0)=D_{I,\Ibar}\), we obtain
\begin{align}
&\lim_{\varepsilon_\beta\to0\;(\beta\in\Ibar)}
\Tr_{\F_{I,\Ibar}}\!\left[
\left[\mathscr L_{I,\Ibar,L}(v_L)\cdots
\mathscr L_{I,\Ibar,1}(v_1)\right]_{\mathbf c,\mathbf d}
W_{I,\Ibar}(\varepsilon)\right]\notag\\
&\qquad=\Tr_{\F_{I,\Ibar}}\!\left[
\left[\mathscr L_{I,\Ibar,L}(v_L)\cdots
\mathscr L_{I,\Ibar,1}(v_1)\right]_{\mathbf c,\mathbf d}
D_{I,\Ibar}\right]\notag\\
&\qquad=[Q_I(u)]_{\mathbf c,\mathbf d},
\label{eq:remaining-trace-is-Q}
\end{align}
where the last equality follows from \eqref{eq:Sch-Q-def}.

The limit of the trace over \(\F_{\mathrm{int}}\) in
\eqref{eq:three-trace-factorization} equals \(C_{\Ibar}\) by
\eqref{eq:CIbar-def}. The traces over \(\F_{I,\Ibar}\) and
\(\F_{\mathrm{int}}\) are interpreted by meromorphic continuation,
as specified after \eqref{eq:CIbar-def}.

All three factors in \eqref{eq:three-trace-factorization} now have
finite limits: the first by Lemma~\ref{lem:one-critical-trace}, and
the other two by Lemma~\ref{lem:regular-trace-away-from-one}.

Substituting \eqref{eq:critical-trace-factor-result},
\eqref{eq:remaining-trace-is-Q}, and \eqref{eq:CIbar-def} into
\eqref{eq:three-trace-factorization}, we obtain
\begin{align}
&\left[\operatorname{HP}_{\Ibar}\mathcal T(u;\mathbf y_{\Ibar})\right]_{\mathbf c,\mathbf d}
\notag\\
&=
C_{\Ibar}
\left(
\prod_{\beta\in\Ibar}
\delta_{m_\beta(\mathbf c),m_\beta(\mathbf d)}
\Gamma(m_\beta(\mathbf d)+1)
\right)
\left[Q_I(u)\right]_{\mathbf c,\mathbf d}.
\label{eq:matrix-element-final-direct}
\end{align}
By \eqref{eq:Q-preserves-Mbeta}, the matrix element of
\(Q_I(u)\) is already zero unless all the Kronecker deltas in
\eqref{eq:matrix-element-final-direct} are equal to one.  Thus \eqref{eq:matrix-element-final-direct} proves
\eqref{eq:main-contraction-formula} on every basis vector.

To express the normalized limit in \eqref{eq:main-contraction-formula}
as a residue in the original variable \(y_\beta\), use
\begin{equation}
  \varepsilon_\beta=1-g_\beta y_\beta,
  \qquad
  y_\beta=\frac{1-\varepsilon_\beta}{g_\beta},
  \qquad
  dy_\beta=-\frac{1}{g_\beta}d\varepsilon_\beta.
  \label{eq:y-epsilon-residue-change}
\end{equation}
Suppose that \(F(y_\beta)\) has a pole of order at most \(m+1\) at \(y_\beta=g_\beta^{-1}\).  Write
\begin{equation}
  F\left(\frac{1-\varepsilon_\beta}{g_\beta}\right)
  =\sum_{k=-m-1}^{\infty}f_k\varepsilon_\beta^k.
\end{equation}
Then
\begin{align}
&\operatorname*{Res}_{y_\beta=g_\beta^{-1}}
  \left[(1-g_\beta y_\beta)^mF(y_\beta)dy_\beta\right]
\notag\\
&=\operatorname*{Res}_{\varepsilon_\beta=0}
  \left[\varepsilon_\beta^m
  F\left(\frac{1-\varepsilon_\beta}{g_\beta}\right)
  \left(-\frac{d\varepsilon_\beta}{g_\beta}\right)
  \right]
\notag\\
&=-g_\beta^{-1}f_{-m-1}
\notag\\
&=-g_\beta^{-1}
  \lim_{\varepsilon_\beta\to0}
  \varepsilon_\beta^{m+1}
  F\left(\frac{1-\varepsilon_\beta}{g_\beta}\right).
  \label{eq:y-residue-to-limit-detailed}
\end{align}
Equivalently,
\begin{equation}
  \lim_{\varepsilon_\beta\to0}
  \varepsilon_\beta^{m+1}
  F\left(\frac{1-\varepsilon_\beta}{g_\beta}\right)
  =-g_\beta
  \operatorname*{Res}_{y_\beta=g_\beta^{-1}}
  \left[(1-g_\beta y_\beta)^mF(y_\beta)dy_\beta\right].
  \label{eq:y-limit-to-residue-detailed}
\end{equation}
Applying \eqref{eq:y-limit-to-residue-detailed} for each \(\beta\in\Ibar\), and then replacing \(m\) by the diagonal operator \(\mathsf M_\beta\) on each simultaneous eigenspace of the operators \(\mathsf M_\beta\), gives \eqref{eq:main-residue-formula}.

\end{proof}

\section{Checks in special cases}
\label{app:short-chain-checks}

This appendix checks Proposition~\ref{prop:main-contraction} in several
special cases.

\subsection{Zero sites}

For \(L=0\), the ordered product of \(L\)-operators is the identity.
Definition~\ref{def:finite-master} gives
\begin{equation}
  \left.\mathcal T(u;\mathbf y_{\Ibar})\right|_{L=0}
  =\prod_{\beta\in\Ibar}\prod_{c\in B}\frac1{1-g_cy_\beta}.
  \label{eq:zero-site-masterT-yIbar}
\end{equation}
Substitute
\(y_\beta=(1-\varepsilon_\beta)/g_\beta\).  Since
\(\mathsf M_\beta=0\) for \(L=0\), the highest-pole coefficient is
\begin{align}
  \operatorname{HP}_{\Ibar}
  \left.\mathcal T(u;\mathbf y_{\Ibar})\right|_{L=0}
  &=\lim_{\varepsilon_\beta\to0\,(\beta\in\Ibar)}
    \left(\prod_{\beta\in\Ibar}\varepsilon_\beta\right)
    \prod_{\beta\in\Ibar}\prod_{c\in B}
    \frac1{1-g_c(1-\varepsilon_\beta)/g_\beta}
  \notag\\
  &=\prod_{\beta\in\Ibar}
    \prod_{\substack{c\in B\\ c\ne\beta}}
    \frac1{1-g_c/g_\beta}.
  \label{eq:zero-site-HP-left}
\end{align}
This is the left-hand side of Proposition~\ref{prop:main-contraction}.
Independently, the definition of \(Q_I(u)\) gives
\begin{align}
  Q_I(u)\big|_{L=0}
  &=\Tr_{\F_{I,\Ibar}}D_{I,\Ibar}
  =\prod_{a\in I}\prod_{\beta\in\Ibar}\frac1{1-g_a/g_\beta} .
  \label{eq:zero-site-Q-right}
\end{align}
Multiplying \eqref{eq:zero-site-Q-right} by the meromorphically
continued scalar \(C_{\Ibar}\) in \eqref{eq:CIbar-def} gives exactly
\eqref{eq:zero-site-HP-left}.  This proves
the zero-site case.

\subsection{One site}

For $L=1$, put $v_1=u-\theta_1$, and let $c,d\in B$.  We first evaluate
all traces in the definition \eqref{eq:TB-master-detailed}; the residues are
then taken from the resulting rational functions of
$\mathbf y_{\Ibar}=(y_\gamma)_{\gamma\in\Ibar}$.
By \eqref{eq:local-L-component},
\begin{align}
  [\mathcal T(u;\mathbf y_{\Ibar})]_{c,d}
  &=v_1\delta_{cd}\Tr_{\F_{B,\Ibar}}W_{\Ibar}(\mathbf y_{\Ibar})
  +\sum_{\gamma\in\Ibar}
  \Tr_{\F_{B,\Ibar}}
  \left[\mathsf a_{d\gamma}^{\dagger}\mathsf a_{c\gamma}
  W_{\Ibar}(\mathbf y_{\Ibar})\right].
  \label{eq:one-site-master-before-trace}
\end{align}
For each $i\in B$ and $\delta\in\Ibar$, the oscillator
$\mathsf a_{i\delta},\mathsf a_{i\delta}^{\dagger}$ contributes
$\sum_{n\ge0}(g_i y_\delta)^n$.  Hence
\begin{align}
  \Tr_{\F_{B,\Ibar}}W_{\Ibar}(\mathbf y_{\Ibar})
  &=\prod_{\delta\in\Ibar}\prod_{i\in B}
  \sum_{n\ge0}(g_i y_\delta)^n
  \notag\\
  &=\prod_{\delta\in\Ibar}\prod_{i\in B}
  \frac{1}{1-g_i y_\delta}.
  \label{eq:one-site-trace-identity}
\end{align}
Fix $\gamma\in\Ibar$.  If $c\ne d$, the operator
$\mathsf a_{d\gamma}^{\dagger}\mathsf a_{c\gamma}$ changes two occupation
numbers, so all of its diagonal matrix elements vanish.  If $c=d$, the
trace factorizes as
\begin{align}
&\Tr_{\F_{B,\Ibar}}
  \left[\mathsf a_{c\gamma}^{\dagger}\mathsf a_{c\gamma}
  W_{\Ibar}(\mathbf y_{\Ibar})\right]
\notag\\
&=\left(\sum_{n\ge0}n(g_cy_\gamma)^n\right)
  \prod_{\substack{\delta\in\Ibar,\ i\in B\\(i,\delta)\ne(c,\gamma)}}
  \left(\sum_{n\ge0}(g_i y_\delta)^n\right)
\notag\\
&=\frac{g_cy_\gamma}{(1-g_cy_\gamma)^2}
  \prod_{\substack{\delta\in\Ibar,\ i\in B\\(i,\delta)\ne(c,\gamma)}}
  \frac{1}{1-g_i y_\delta}
  \qquad\text{[by \eqref{eq:A0-trace} and \eqref{eq:A1-trace}]}
\notag\\
&=\frac{g_cy_\gamma}{1-g_cy_\gamma}
  \prod_{\delta\in\Ibar}\prod_{i\in B}
  \frac{1}{1-g_i y_\delta}.
\end{align}
Consequently, for $c,d\in B$ and $\gamma\in\Ibar$,
\begin{align}
&\Tr_{\F_{B,\Ibar}}
  \left[\mathsf a_{d\gamma}^{\dagger}\mathsf a_{c\gamma}
  W_{\Ibar}(\mathbf y_{\Ibar})\right]
\notag\\
&\qquad=\delta_{cd}\,
  \frac{g_cy_\gamma}{1-g_cy_\gamma}
  \prod_{\delta\in\Ibar}\prod_{i\in B}
  \frac{1}{1-g_i y_\delta}.
  \label{eq:one-site-bilinear-trace}
\end{align}
Substitution into \eqref{eq:one-site-master-before-trace} gives
\begin{equation}
  [\mathcal T(u;\mathbf y_{\Ibar})]_{c,d}
  =\delta_{cd}
  \left(\prod_{\delta\in\Ibar}\prod_{i\in B}
  \frac{1}{1-g_i y_\delta}\right)
  \left(v_1+\sum_{\gamma\in\Ibar}
  \frac{g_cy_\gamma}{1-g_cy_\gamma}\right),
  \qquad c,d\in B.
  \label{eq:one-site-master-explicit}
\end{equation}

We now apply the multiple residue in \eqref{eq:main-residue-formula} to
\eqref{eq:one-site-master-explicit}.  Let \(a\in I\).  Then
\(\mathsf M_\gamma e_a=0\) for any \(\gamma\in\Ibar\).  For fixed
\(\gamma\in\Ibar\), the pole at \(y_\gamma=g_\gamma^{-1}\) comes from
the factor \((1-g_\gamma y_\gamma)^{-1}\) in
\eqref{eq:one-site-master-explicit}.  We have
\[
  \operatorname*{Res}_{y_\gamma=g_\gamma^{-1}}
  \frac{dy_\gamma}{1-g_\gamma y_\gamma}=-g_\gamma^{-1}.
\]
The remaining factors depending on \(y_\gamma\), including
\((1-g_i y_\gamma)^{-1}\) with \(i\ne\gamma\) and
\[
  \frac{g_a y_\gamma}{1-g_a y_\gamma},
\]
are regular at \(y_\gamma=g_\gamma^{-1}\).

 It follows that
\begin{align}
&{\operatorname*{Res}}_{y_\gamma=g_\gamma^{-1},\,\gamma\in\Ibar}
\left[
[\mathcal T(u;\mathbf y_{\Ibar})]_{a,a}
\prod_{\gamma\in\Ibar}dy_\gamma
\right]
\notag\\
&\quad=(-1)^{|\Ibar|}
\left(\prod_{\gamma\in\Ibar}g_\gamma^{-1}\right)
C_{\Ibar}\chi_{I,\Ibar}
\left(
 v_1+\sum_{\gamma\in\Ibar}
 \frac{g_a/g_\gamma}{1-g_a/g_\gamma}
\right).
\label{eq:one-site-residue-I}
\end{align}
Here the factors indexed by $a\in I$ give $\chi_{I,\Ibar}$ by
\eqref{eq:chi-I-Ibar}, while the factors indexed by distinct
$\beta,\gamma\in\Ibar$ give $C_{\Ibar}$ by \eqref{eq:CIbar-def}.

Let $\beta\in\Ibar$.  Since
$\mathsf M_\gamma e_\beta=\delta_{\gamma\beta}e_\beta$, the expression inside the residue in
\eqref{eq:main-residue-formula} contains $1-g_\beta y_\beta$.  This cancels
the factor $(1-g_\beta y_\beta)^{-1}$ in the first product in
\eqref{eq:one-site-master-explicit}.  The term $v_1$ and all terms in the
sum with $\gamma\ne\beta$ are then regular in $y_\beta$.  The term with
$\gamma=\beta$ gives
\[
  \operatorname*{Res}_{y_\beta=g_\beta^{-1}}
  \frac{g_\beta y_\beta}{1-g_\beta y_\beta}\,dy_\beta
  =-g_\beta^{-1}.
\]
Taking the residues in all variables yields
\begin{align}
&{\operatorname*{Res}}_{y_\gamma=g_\gamma^{-1},\,\gamma\in\Ibar}
\left[
(1-g_\beta y_\beta)
[\mathcal T(u;\mathbf y_{\Ibar})]_{\beta,\beta}
\prod_{\gamma\in\Ibar}dy_\gamma
\right]
\notag\\
&\quad=(-1)^{|\Ibar|}
\left(\prod_{\gamma\in\Ibar}g_\gamma^{-1}\right)
C_{\Ibar}\chi_{I,\Ibar}.
\label{eq:one-site-residue-Ibar}
\end{align}
All off-diagonal matrix elements in \eqref{eq:one-site-master-explicit}
vanish, and therefore so do their residues.

We next compute $Q_I(u)$, which appears on the right-hand side of
\eqref{eq:main-contraction-formula}, directly from its definition
\eqref{eq:Sch-Q-def}.  Let $a,c\in I$ and $\gamma\in\Ibar$.  If $a\ne c$,
then $\mathsf a_{c\gamma}^{\dagger}\mathsf a_{a\gamma}$ has no diagonal
matrix elements.  If $a=c$, then
\begin{align}
&\Tr_{\F_{I,\Ibar}}
  \left[\mathsf a_{a\gamma}^{\dagger}\mathsf a_{a\gamma}
  D_{I,\Ibar}\right]
\notag\\
&=\left(\sum_{n\ge0}n(g_a/g_\gamma)^n\right)
  \prod_{\substack{b\in I,\ \delta\in\Ibar\\(b,\delta)\ne(a,\gamma)}}
  \left(\sum_{n\ge0}(g_b/g_\delta)^n\right)
\notag\\
&=\frac{g_a/g_\gamma}{(1-g_a/g_\gamma)^2}
  \prod_{\substack{b\in I,\ \delta\in\Ibar\\(b,\delta)\ne(a,\gamma)}}
  \frac{1}{1-g_b/g_\delta}
  \qquad\text{[by \eqref{eq:A0-trace} and \eqref{eq:A1-trace}]}
\notag\\
&=\chi_{I,\Ibar}
  \frac{g_a/g_\gamma}{1-g_a/g_\gamma}.
\end{align}
Thus, for $a,c\in I$ and $\gamma\in\Ibar$,
\begin{equation}
  \Tr_{\F_{I,\Ibar}}
  \left[\mathsf a_{c\gamma}^{\dagger}\mathsf a_{a\gamma}
  D_{I,\Ibar}\right]
  =\delta_{ac}\chi_{I,\Ibar}
  \frac{g_a/g_\gamma}{1-g_a/g_\gamma}.
  \label{eq:one-site-Q-bilinear-evaluated}
\end{equation}
For $a\in I$ and $\beta\in\Ibar$, the operators
$\mathsf a_{a\beta}$ and $\mathsf a_{a\beta}^{\dagger}$ change the
occupation number, and hence
\[
  \Tr_{\F_{I,\Ibar}}(\mathsf a_{a\beta}D_{I,\Ibar})
  =\Tr_{\F_{I,\Ibar}}(\mathsf a_{a\beta}^{\dagger}D_{I,\Ibar})=0.
\]
Using the matrix elements
\eqref{eq:contracted-L-ac}--\eqref{eq:contracted-L-components}, we obtain,
for $c,d\in B$,
\begin{equation}
  [Q_I(u)]_{c,d}
  =
  \begin{cases}
  \displaystyle
  \chi_{I,\Ibar}
  \left(v_1+\sum_{\gamma\in\Ibar}
  \frac{g_c/g_\gamma}{1-g_c/g_\gamma}\right),
  & c=d\in I,\\[3mm]
  \chi_{I,\Ibar},& c=d\in\Ibar,\\
  0,& c\ne d.
  \end{cases}
  \label{eq:one-site-Q-explicit}
\end{equation}

For one site, each eigenvalue of $\mathsf M_\gamma$ is zero or one, so
$\Gamma(\mathsf M_\gamma+1)=1$ on each basis vector $e_c$, $c\in B$.
Multiplying \eqref{eq:one-site-residue-I} and
\eqref{eq:one-site-residue-Ibar} by the prefactor in
\eqref{eq:main-residue-formula} gives \eqref{eq:one-site-Q-explicit}.
The off-diagonal matrix elements vanish in both calculations.  This proves
the one-site case.

\subsection{A further check for arbitrary chain length}

Fix \(\beta\in\Ibar\), write
\(\mathbf b_L=(\beta,\ldots,\beta)\in B^L\), and put
\(v_\ell=u-\theta_\ell\).  For arbitrary \(L\ge1\), we first compute
\([\mathcal T(u;\mathbf y_{\Ibar})]_{\mathbf b_L,\mathbf b_L}\)
before taking the highest-pole limit.  Equation
\eqref{eq:local-L-component} gives
\[
  [L^{(\Ibar)}(v_\ell)]_{\beta\beta}
  =v_\ell+\sum_{\gamma\in\Ibar}\mathsf N_{\beta\gamma}.
\]
Consequently, Definition~\ref{def:finite-master} gives the exact matrix
element
\[
  [\mathcal T(u;\mathbf y_{\Ibar})]_{\mathbf b_L,\mathbf b_L}
  =\Tr_{\F_{B,\Ibar}}
  \left[
   \prod_{\ell=1}^{L}
   \left(v_\ell+\sum_{\gamma\in\Ibar}\mathsf N_{\beta\gamma}\right)
   W_{\Ibar}(\mathbf y_{\Ibar})
  \right].
\]
The factors in the product commute.  Set
\[
  x_\gamma=g_\beta y_\gamma,
  \qquad
  \mathcal D_\beta
  =\sum_{\gamma\in\Ibar}x_\gamma\frac{\partial}{\partial x_\gamma},
\]
and define
\[
  Z_\beta(\mathbf y_{\Ibar})
  =\prod_{\gamma\in\Ibar}\frac{1}{1-g_\beta y_\gamma},
  \qquad
  Z_{\ne\beta}(\mathbf y_{\Ibar})
  =\prod_{\gamma\in\Ibar}
   \prod_{\substack{i\in B\\ i\ne\beta}}
   \frac{1}{1-g_i y_\gamma}.
\]
Summing the diagonal matrix elements in the occupation-number basis
therefore gives the exact
rational-function identity
\begin{equation}
  [\mathcal T(u;\mathbf y_{\Ibar})]_{\mathbf b_L,\mathbf b_L}
  =Z_{\ne\beta}(\mathbf y_{\Ibar})
   \left[\prod_{\ell=1}^{L}(v_\ell+\mathcal D_\beta)\right]
   Z_\beta(\mathbf y_{\Ibar}).
\end{equation}
This identity is first evaluated in the common convergence domain
\(|g_i y_\gamma|<1\), and the resulting rational function is then continued
meromorphically.  Here the variables \(x_\gamma\) are regarded as independent
while the derivatives are taken, and are then specialized to
\(x_\gamma=g_\beta y_\gamma\).

We now take the limit by substituting
\(y_\gamma=(1-\varepsilon_\gamma)/g_\gamma\).  Since
\(\mathsf M_\gamma|\mathbf b_L\rangle
=L\delta_{\gamma\beta}|\mathbf b_L\rangle\),
Eq.~\eqref{eq:HP-equivalent} gives
\begin{align*}
 & [\operatorname{HP}_{\Ibar}\mathcal T(u;\mathbf y_{\Ibar})]
   _{\mathbf b_L,\mathbf b_L}
 \\
 &\qquad=
 \lim_{\varepsilon_\gamma\to0\;(\gamma\in\Ibar)}
 \varepsilon_\beta^{L+1}
 \prod_{\substack{\gamma\in\Ibar\\ \gamma\ne\beta}}
 \varepsilon_\gamma\,
 [\mathcal T(u;\mathbf y(\varepsilon))]
 _{\mathbf b_L,\mathbf b_L}.
\end{align*}
Moreover,
\[
 Z_\beta(\mathbf y(\varepsilon))
 =\frac{1}{\varepsilon_\beta}
  \prod_{\substack{\gamma\in\Ibar\\ \gamma\ne\beta}}
  \frac{1}{1-\dfrac{g_\beta}{g_\gamma}
  (1-\varepsilon_\gamma)}.
\]
To obtain the leading term explicitly, put
\[
 E_\gamma:=x_\gamma\frac{\partial}{\partial x_\gamma},
 \qquad
 \mathcal D_\beta=E_\beta+\mathcal D_\beta^{(\ne\beta)},
 \qquad
 \mathcal D_\beta^{(\ne\beta)}
 :=\sum_{\substack{\gamma\in\Ibar\\ \gamma\ne\beta}}E_\gamma,
\]
and isolate the factor that becomes singular as \(x_\beta\to1\):
\[
 Z_\beta(\mathbf y(\varepsilon))
 =\frac{1}{1-x_\beta}\,G_\beta(\boldsymbol\varepsilon),
 \qquad
 G_\beta(\boldsymbol\varepsilon)
 :=\prod_{\substack{\gamma\in\Ibar\\ \gamma\ne\beta}}
 \frac{1}{1-\dfrac{g_\beta}{g_\gamma}(1-\varepsilon_\gamma)}.
\]
Here \(x_\beta=1-\varepsilon_\beta\), whereas
\(G_\beta\) is independent of \(\varepsilon_\beta\).  Since
\[
 E_\beta
 =x_\beta\frac{\partial}{\partial x_\beta}
 =-(1-\varepsilon_\beta)
   \frac{\partial}{\partial\varepsilon_\beta},
\]
we have, for every integer \(k\ge0\),
\begin{equation}
 E_\beta^k\frac{1}{1-x_\beta}
 =E_\beta^k\varepsilon_\beta^{-1}
 =k!\,\varepsilon_\beta^{-k-1}
  +O(\varepsilon_\beta^{-k}).
\end{equation}
This follows by induction on \(k\).

Next expand the polynomial in \(\mathcal D_\beta\):
\[
 \prod_{\ell=1}^{L}(v_\ell+\mathcal D_\beta)
 =\sum_{k=0}^{L}e_{L-k}(v_1,\ldots,v_L)
  \mathcal D_\beta^k,
\]
where \(e_j\) is the elementary symmetric polynomial and \(e_0=1\).
Since \(E_\beta\) and \(\mathcal D_\beta^{(\ne\beta)}\) commute,
\begin{align*}
 \mathcal D_\beta^k
 \left(\frac{G_\beta}{1-x_\beta}\right)
 & =\sum_{j=0}^{k}\binom{k}{j}
 \left(E_\beta^j\frac{1}{1-x_\beta}\right)
 \left((\mathcal D_\beta^{(\ne\beta)})^{k-j}G_\beta\right).
\end{align*}
For fixed \(\varepsilon_\gamma\) with \(\gamma\ne\beta\), the second
factor is finite as \(\varepsilon_\beta\to0\).  By the preceding estimate, the summand
with index \(j\) has pole order at most \(j+1\).  Consequently, among all
terms in the preceding two expansions, a pole of order \(L+1\) can arise
only when
\[
 k=L,\qquad j=L.
\]
Its coefficient is \(L!G_\beta\), because the coefficient of
\(\mathcal D_\beta^L\) in the polynomial is \(e_0=1\).  Every term with
\(k<L\), every term containing at least one \(v_\ell\), and every term in
which at least one Euler derivative acts on \(G_\beta\), has pole order at
most \(L\).  Therefore
\[
 \left[\prod_{\ell=1}^{L}(v_\ell+\mathcal D_\beta)\right]
 Z_\beta(\mathbf y(\varepsilon))
 =\frac{L!}{\varepsilon_\beta^{L+1}}
  \prod_{\substack{\gamma\in\Ibar\\ \gamma\ne\beta}}
  \frac{1}{1-\dfrac{g_\beta}{g_\gamma}
  (1-\varepsilon_\gamma)}
  +O(\varepsilon_\beta^{-L}).
\]
In particular, after multiplication by \(\varepsilon_\beta^{L+1}\),
all lower-order poles vanish as \(\varepsilon_\beta\to0\).

It follows that
\begin{align*}
 & [\operatorname{HP}_{\Ibar}\mathcal T(u;\mathbf y_{\Ibar})]
   _{\mathbf b_L,\mathbf b_L}
 \\
 &\qquad=L!
 \prod_{\delta\in\Ibar}
 \prod_{\substack{i\in B\\ i\ne\delta}}
 \frac{1}{1-g_i/g_\delta}
 =L!\,C_{\Ibar}\chi_{I,\Ibar}.
\end{align*}

Finally, since
\([\mathscr L_{I,\Ibar}(v_\ell)]_{\beta\beta}=1\), the definition
\eqref{eq:Sch-Q-def} gives
\[
 [Q_I(u)]_{\mathbf b_L,\mathbf b_L}
 =\Tr_{\F_{I,\Ibar}}D_{I,\Ibar}
 =\chi_{I,\Ibar}.
\]
Therefore
\[
 [\operatorname{HP}_{\Ibar}\mathcal T(u;\mathbf y_{\Ibar})]
 _{\mathbf b_L,\mathbf b_L}
 =C_{\Ibar}\,L!\,[Q_I(u)]_{\mathbf b_L,\mathbf b_L}.
\]
This verifies Proposition~\ref{prop:main-contraction} for this diagonal
matrix element.

\section{The minimal example M=2}
\label{app:minimal-M2}

Take \(M=2\), \(I=\{2\}\), and \(\Ibar=\{1\}\).  Put
\begin{equation}
  \rho=\frac{g_2}{g_1},
  \qquad
  \mathsf N=\mathsf a_{21}^{\dagger}\mathsf a_{21}.
\end{equation}
With respect to the ordered basis \((e_1,e_2)\) of \(\C^2\),
\begin{equation}
  \mathscr L_{I,\Ibar}(v)=
  \begin{pmatrix}
    1 & \mathsf a_{21}^{\dagger}\\
    \mathsf a_{21} & v+\mathsf N
  \end{pmatrix},
  \qquad
  D_{I,\Ibar}=\rho^{\mathsf N}.
\end{equation}
All traces in this example are taken over the Fock representation
\(\F_{I,\Ibar}\).  The quantities \(A_q\) are defined by
\eqref{eq:Aq-trace}; the explicit formulas for \(A_0,A_1,A_2,A_3\) are
\eqref{eq:A0-trace}--\eqref{eq:A3-trace}.
For one site,
\begin{equation}
  Q_I(u)=
  \begin{pmatrix}
    A_0 & 0\\
    0 & v_1A_0+A_1
  \end{pmatrix}.
\end{equation}
For two sites, we write matrix elements of this operator as
\([Q_I(u)]_{\mathbf c,\mathbf d}\).  For example,
\begin{align}
  [Q_I(u)]_{11,11}&=A_0,
  &
  [Q_I(u)]_{12,12}&=v_2A_0+A_1,
  \\
  [Q_I(u)]_{21,12}
  &=\Tr_{\F_{I,\Ibar}}
    (\mathsf a_{21}^{\dagger}\mathsf a_{21}\rho^{\mathsf N})
  \notag\\
  &=\frac{\rho}{(1-\rho)^2},
  \\
  [Q_I(u)]_{12,21}
  &=\Tr_{\F_{I,\Ibar}}
    (\mathsf a_{21}\mathsf a_{21}^{\dagger}\rho^{\mathsf N})
  \notag\\
  &=\Tr_{\F_{I,\Ibar}}
    ((\mathsf N+1)\rho^{\mathsf N})
  \notag\\
  &=\frac1{(1-\rho)^2}.
\end{align}
Here we use \eqref{eq:A0-trace}, \eqref{eq:A1-trace}, and
\[
  \mathsf a_{21}\mathsf a_{21}^{\dagger}
  =\mathsf a_{21}^{\dagger}\mathsf a_{21}+1.
\]
For three sites, we obtain 
\begin{align}
  [Q_I(u)]_{222,222}
  &=\Tr_{\F_{I,\Ibar}}
  \bigl[(v_3+\mathsf N)(v_2+\mathsf N)
  (v_1+\mathsf N)\rho^{\mathsf N}\bigr]
  \notag\\
  &=v_1v_2v_3A_0
  +(v_1v_2+v_1v_3+v_2v_3)A_1
  \notag\\
  &\quad +(v_1+v_2+v_3)A_2+A_3.
\end{align}

\section[Comparison with the degenerate Yangian L-operator]{Comparison with the degenerate Yangian \texorpdfstring{$L$}{L}-operator of 
\cite{BFLMS2011}}
\label{app:BFLMS-local-comparison}

This appendix compares the matrix elements
\eqref{eq:contracted-L-ac}--\eqref{eq:contracted-L-components} with the
degenerate Yangian \(L\)-operator of Bazhanov, Frassek, Lukowski,
Meneghelli and Staudacher \cite{BFLMS2011}.  We write
\(B=I\sqcup\Ibar\), and use \(a,c,d,e\in I\), and 
\(\beta,\gamma\in\Ibar\).  Put \(p=|I|\).

Let \(\mathcal E=\End(\C^M)\).  In the convention of Ref.~\cite{BFLMS2011}, it has the form:
\begin{equation}
  L_I^{\mathrm{BFLMS}}(z)
  =\sum_{c,d\in B}E_{cd}\otimes
  \ell_{cd}^{\mathrm{BFLMS}}(z)
  \in \mathcal E\otimes\mathcal A_{\mathrm{BFLMS}},
  \label{eq:BFLMS-L-expanded}
\end{equation}
where $z \in \mathbb{C}$. 
The associative algebra \(\mathcal A_{\mathrm{BFLMS}}\) contains generators
\(\mathcal J_{ac}\), \(a,c\in I\), satisfying 
the relation of $\gl(V_{I})$: 
\begin{equation}
 [\mathcal J_{ac},\mathcal J_{de}]
 =\delta_{cd}\mathcal J_{ae}-\delta_{ea}\mathcal J_{dc},
 \label{eq:BFLMS-J-relation}
\end{equation}
and the generators
\(\mathbf b_{\beta a}\), \(\mathbf b_{a\beta}^{\dagger}\) of a bosonic oscillator algebra satisfying
\begin{equation}
  [\mathbf b_{\beta a},\mathbf b_{c\gamma}^{\dagger}]
  =\delta_{\beta\gamma}\delta_{ac},
  \qquad
  [\mathbf b_{\beta a},\mathbf b_{\gamma c}]
  =[\mathbf b_{a\beta}^{\dagger},\mathbf b_{c\gamma}^{\dagger}]=0.
  \label{eq:BFLMS-osc-rel}
\end{equation}
The generators in \eqref{eq:BFLMS-J-relation} commute with the 
generators of the oscillator algebra.  The non-zero matrix elements in \eqref{eq:BFLMS-L-expanded}
are
\begin{align}
  \ell_{ac}^{\mathrm{BFLMS}}(z)
  &=z\delta_{ac}+\mathcal J_{ca}
    -\sum_{\beta\in\Ibar}
    \left(\mathbf b_{a\beta}^{\dagger}\mathbf b_{\beta c}
    +\frac12\delta_{ac}\right),
  &&a,c\in I,
  \label{eq:BFLMS-L-II}
  \\
  \ell_{a\beta}^{\mathrm{BFLMS}}(z)
  &=\mathbf b_{a\beta}^{\dagger},
  &&a\in I,\ \beta\in\Ibar,
  \label{eq:BFLMS-L-I-Ibar}
  \\
  \ell_{\beta a}^{\mathrm{BFLMS}}(z)
  &=-\mathbf b_{\beta a},
  &&\beta\in\Ibar,\ a\in I,
  \label{eq:BFLMS-L-Ibar-I}
  \\
  \ell_{\beta\gamma}^{\mathrm{BFLMS}}(z)
  &=\delta_{\beta\gamma},
  &&\beta,\gamma\in\Ibar.
  \label{eq:BFLMS-L-IbarIbar}
\end{align}

Let \(\mathcal A_{I,\Ibar}^{\mathcal J}\) be the associative algebra generated
by elements \(\mathcal J_{ac}\) satisfying \eqref{eq:BFLMS-J-relation} and
by \(\mathsf a_{a\beta},\mathsf a_{a\beta}^{\dagger}\) satisfying the
relations of Section~\ref{sec:conventions}.  Define
\(\Phi:\mathcal A_{\mathrm{BFLMS}}\to
\mathcal A_{I,\Ibar}^{\mathcal J}\) by
\begin{equation}
  \Phi(\mathcal J_{ac})=\mathcal J_{ac},
  \qquad
  \Phi(\mathbf b_{\beta a})=\mathsf a_{a\beta}^{\dagger},
  \qquad
  \Phi(\mathbf b_{a\beta}^{\dagger})=-\mathsf a_{a\beta}.
  \label{eq:BFLMS-to-ours-algebra-map}
\end{equation}
The commutator of the oscillator algebra is preserved because
\begin{equation}
  [\Phi(\mathbf b_{\beta a}),\Phi(\mathbf b_{c\gamma}^{\dagger})]
  =[\mathsf a_{a\beta}^{\dagger},-\mathsf a_{c\gamma}]
  =\delta_{ac}\delta_{\beta\gamma}.
  \label{eq:BFLMS-to-ours-commutator-check}
\end{equation}
Thus \(\Phi\) is a homomorphism of associative algebras.

Define
\begin{equation}
  d_c=1\quad(c\in I),
  \qquad
  d_c=-1\quad(c\in\Ibar),
  \label{eq:D-factors-BFLMS-comparison}
\end{equation}
and
\begin{equation}
  \Sigma_I=\sum_{a\in I}E_{aa}-\sum_{\beta\in\Ibar}E_{\beta\beta},
  \qquad
  \Sigma_I E_{cd}\Sigma_I^{-1}=d_cd_dE_{cd}.
  \label{eq:D-operator-BFLMS-comparison}
\end{equation}
Set
\begin{equation}
  \mathscr L_{I,\Ibar}^{\mathcal J}(v)
  :=\sum_{c,d\in B}d_cd_d\,
  \Phi\!\left(
  \ell_{cd}^{\mathrm{BFLMS}}
  \left(v-\frac{|\Ibar|}{2}\right)
  \right)\otimes E_{cd}.
  \label{eq:direct-BFLMS-ours-relation}
\end{equation}
For \(a,c\in I\), the relation
\(\mathsf a_{a\beta}\mathsf a_{c\beta}^{\dagger}
=\mathsf a_{c\beta}^{\dagger}\mathsf a_{a\beta}+\delta_{ac}\) gives
\begin{align}
  &d_ad_c\,
  \Phi\!\left(
  \ell_{ac}^{\mathrm{BFLMS}}
  \left(v-\frac{|\Ibar|}{2}\right)
  \right)
  \notag\\
  &\quad=
  \left(v-\frac{|\Ibar|}{2}\right)\delta_{ac}
  +\mathcal J_{ca}
  -\sum_{\beta\in\Ibar}
  \left(-\mathsf a_{a\beta}\mathsf a_{c\beta}^{\dagger}
  +\frac12\delta_{ac}\right)
  \notag\\
  &\quad=v\delta_{ac}+\mathcal J_{ca}
  +\sum_{\beta\in\Ibar}
  \mathsf a_{c\beta}^{\dagger}\mathsf a_{a\beta}.
  \label{eq:BFLMS-check-II-block}
\end{align}
The remaining matrix elements are
\begin{align}
  d_ad_\beta\Phi(\ell_{a\beta}^{\mathrm{BFLMS}}(z))
  &=\mathsf a_{a\beta},
  \label{eq:BFLMS-check-I-Ibar-block}\\
  d_\beta d_a\Phi(\ell_{\beta a}^{\mathrm{BFLMS}}(z))
  &=\mathsf a_{a\beta}^{\dagger},
  \label{eq:BFLMS-check-Ibar-I-block}\\
  d_\beta d_\gamma\Phi(\ell_{\beta\gamma}^{\mathrm{BFLMS}}(z))
  &=\delta_{\beta\gamma}.
  \label{eq:BFLMS-check-Ibar-Ibar-block}
\end{align}
When \(\mathcal J_{ac}\) acts by zero, these formulas reproduce
\eqref{eq:contracted-L-ac}--\eqref{eq:contracted-L-components}.

We next transform the RLL relation.  
Let \(\mathcal X\) and \(\mathcal Y\) be 
associative algebras, and take 
\(x,x' \in \mathcal X\) and \(y,y' \in \mathcal Y\). 
 Define a map 
\begin{equation}
  f_{\mathcal X,\mathcal Y}:\mathcal X\otimes\mathcal Y
  \longrightarrow\mathcal Y\otimes\mathcal X,
  \qquad
  f_{\mathcal X,\mathcal Y}(x\otimes y)=y\otimes x.
  \label{eq:two-factor-flip}
\end{equation}
This is an algebra isomorphism, since
\begin{equation}
 f_{\mathcal X,\mathcal Y}
 ((x\otimes y)(x'\otimes y'))
 =yy'\otimes xx'
 =f_{\mathcal X,\mathcal Y}(x\otimes y)
  f_{\mathcal X,\mathcal Y}(x'\otimes y').
 \label{eq:two-factor-flip-product}
\end{equation}
The RLL relation of \cite{BFLMS2011} is an equality in
\(\mathcal E\otimes\mathcal E\otimes\mathcal A_{\mathrm{BFLMS}}\):
\begin{equation}
 R_{12}(u-v)L_{I,13}^{\mathrm{BFLMS}}(u)
 L_{I,23}^{\mathrm{BFLMS}}(v)
 =L_{I,23}^{\mathrm{BFLMS}}(v)
 L_{I,13}^{\mathrm{BFLMS}}(u)R_{12}(u-v).
 \label{eq:BFLMS-RLL-tensor-form}
\end{equation}
Apply successively
\begin{equation}
 \tau_1=\operatorname{id}_{\mathcal E}\otimes
 f_{\mathcal E,\mathcal A_{\mathrm{BFLMS}}},
 \qquad
 \tau_2=f_{\mathcal E,\mathcal A_{\mathrm{BFLMS}}}
 \otimes\operatorname{id}_{\mathcal E}.
 \label{eq:two-successive-flips}
\end{equation}
Define
\begin{equation}
 \widetilde L_I^{\mathrm{BFLMS}}(z)
 :=f_{\mathcal E,\mathcal A_{\mathrm{BFLMS}}}
 (L_I^{\mathrm{BFLMS}}(z))
 =\sum_{c,d\in B}\ell_{cd}^{\mathrm{BFLMS}}(z)\otimes E_{cd}.
 \label{eq:BFLMS-L-flipped}
\end{equation}
The images of the three factors in \eqref{eq:BFLMS-RLL-tensor-form} are
\begin{align}
 (\tau_2\circ\tau_1)(R_{12}(u-v))&=R_{23}(u-v),
 \notag\\
 (\tau_2\circ\tau_1)(L_{I,13}^{\mathrm{BFLMS}}(u))
 &=\widetilde L_{I,12}^{\mathrm{BFLMS}}(u),
 \notag\\
 (\tau_2\circ\tau_1)(L_{I,23}^{\mathrm{BFLMS}}(v))
 &=\widetilde L_{I,13}^{\mathrm{BFLMS}}(v).
 \label{eq:images-under-two-flips}
\end{align}
Because \(\tau_1\) and \(\tau_2\) preserve multiplication,
\eqref{eq:BFLMS-RLL-tensor-form} becomes
\begin{equation}
 R_{23}(u-v)\widetilde L_{I,12}^{\mathrm{BFLMS}}(u)
 \widetilde L_{I,13}^{\mathrm{BFLMS}}(v)
 =\widetilde L_{I,13}^{\mathrm{BFLMS}}(v)
 \widetilde L_{I,12}^{\mathrm{BFLMS}}(u)R_{23}(u-v).
 \label{eq:BFLMS-RLL-after-flips}
\end{equation}
Apply \(\Phi\) to the first tensor factor, conjugate 
the second and the third tensor factors by \(\Sigma_I\), 
and shift both spectral parameters $u$ and $v$ 
by \(-|\Ibar|/2\).  These operations preserve
\eqref{eq:BFLMS-RLL-after-flips}.  For the conjugation, one uses
\begin{equation}
 (\Sigma_I\otimes\Sigma_I)R(w)
 (\Sigma_I^{-1}\otimes\Sigma_I^{-1})=R(w),
 \qquad R(w)=w\,\Id_{\mathbb{C}^{M} \otimes \mathbb{C}^{M}}+P,
 \label{eq:D-preserves-R}
\end{equation}
where $w \in \mathbb{C}$.
Consequently \(\mathscr L_{I,\Ibar}^{\mathcal J}\) satisfies
\begin{equation}
 R_{23}(u-v)\mathscr L_{I,\Ibar,12}^{\mathcal J}(u)
 \mathscr L_{I,\Ibar,13}^{\mathcal J}(v)
 =\mathscr L_{I,\Ibar,13}^{\mathcal J}(v)
 \mathscr L_{I,\Ibar,12}^{\mathcal J}(u)R_{23}(u-v).
 \label{eq:BFLMS-RLL-ours-first-form}
\end{equation}
Let \(w=v-u\).  Since
\begin{equation}
 R(w)R(-w)=(1-w^2)\Id_{\mathbb{C}^{M} \otimes \mathbb{C}^{M}},
 \label{eq:rational-R-unitarity}
\end{equation}
we may multiply \eqref{eq:BFLMS-RLL-ours-first-form} on the left and on the
right by \(R_{23}(w)\).  For generic \(w\), division by \(1-w^2\) gives
\begin{equation}
 \mathscr L_{I,\Ibar,12}^{\mathcal J}(u)
 \mathscr L_{I,\Ibar,13}^{\mathcal J}(v)R_{23}(v-u)
 =R_{23}(v-u)\mathscr L_{I,\Ibar,13}^{\mathcal J}(v)
 \mathscr L_{I,\Ibar,12}^{\mathcal J}(u).
 \label{eq:BFLMS-RLL-ours-second-form}
\end{equation}
Both sides are polynomial in \(w\), so the equality holds for all \(w\).
Finally, setting \(\mathcal J_{ac}=0\) gives
\eqref{eq:RLL-contracted-HP} for \(\mathscr L_{I,\Ibar}(v)\).

\section{Additional flavour labels and polynomial representations of \texorpdfstring{$GL(V_I)$}{GL(VI)}}
\label{app:nontrivial-glI}

Let \(S_{\mathrm{add}}\) be a finite set disjoint from \(\Ibar\), and set
\[
  R=\Ibar\sqcup S_{\mathrm{add}}
\]
throughout this appendix. For \(\mu\in S_{\mathrm{add}}\), write
\(y_\mu=\eta_\mu\) and keep \(\eta_\mu\) fixed as
\(y_\beta\to g_\beta^{-1}\), \(\beta\in\Ibar\). Put
\[
  \boldsymbol\eta=(\eta_\mu)_{\mu\in S_{\mathrm{add}}},
\]
and assume
\begin{equation}
  g_c\eta_\mu\ne1,
  \qquad c\in B,\quad \mu\in S_{\mathrm{add}}.
  \label{eq:additional-variables-regular}
\end{equation}
The oscillators with colour labels in \(I\) and flavour labels in
\(S_{\mathrm{add}}\) give a polynomial representation of \(GL(V_I)\).
After the contraction described in Section~\ref{sec:local-contraction},
this representation appears only in the matrix elements of the
contracted \(L\)-operator whose row and column labels both belong to
\(I\).
The Fock representation \(\F_{I,S_{\mathrm{add}}}\) and its number
operators are those defined in \eqref{eq:Fock-notation-fixed} and
\eqref{eq:number-operator-fixed}.  Put
\begin{equation}
  \mathcal J_{ac}^{\mathrm{add}}
  =\sum_{\mu\in S_{\mathrm{add}}}
  \mathsf a_{a\mu}^{\dagger}\mathsf a_{c\mu},
  \qquad a,c\in I.
  \label{eq:Jadd-def}
\end{equation}
By \eqref{eq:eij-commutator}, these operators satisfy 
the relations of $\gl(V_{I})$:
\begin{equation}
  [\mathcal J_{ac}^{\mathrm{add}},\mathcal J_{de}^{\mathrm{add}}]
  =\delta_{cd}\mathcal J_{ae}^{\mathrm{add}}
   -\delta_{ea}\mathcal J_{dc}^{\mathrm{add}},
  \qquad a,c,d,e\in I.
  \label{eq:Jadd-gl-relation}
\end{equation}

Start from the \(L\)-operator \eqref{eq:Schwinger-L} with flavour set
\(\Ibar\sqcup S_{\mathrm{add}}\), and apply the scaling
\eqref{eq:S-epsilon-scaling}.  Fix \(\mu\in S_{\mathrm{add}}\).  The term
with flavour label \(\mu\) in the \((c,d)\)-matrix element becomes
\begin{equation}
  s_cs_d\,\mathsf a_{d\mu}^{\dagger}\mathsf a_{c\mu}.
  \label{eq:additional-flavour-scaled-term}
\end{equation}
If \(c,d\in I\), then \eqref{eq:scaling-sc} gives \(s_cs_d=1\), and
summing \eqref{eq:additional-flavour-scaled-term} over
\(\mu\in S_{\mathrm{add}}\) gives \(\mathcal J_{dc}^{\mathrm{add}}\).
If \(c\in\Ibar\), the same equation supplies a factor
\(\sqrt{\varepsilon_c}\); if \(d\in\Ibar\), it supplies a factor
\(\sqrt{\varepsilon_d}\).  Thus at least one such factor is present unless
both row and column labels belong to \(I\).  For \(R=\Ibar\sqcup S_{\mathrm{add}}\),
Eqs.~\eqref{eq:full-twist-R} and \eqref{eq:Wcrit-simple} give
\[
  \widehat g\,\widehat h(\mathbf y_R)
  =W_{\Ibar}(\mathbf y_{\Ibar})
   \prod_{\mu\in S_{\mathrm{add}}}\prod_{i\in B}
   (g_i\eta_\mu)^{\mathsf N_{i\mu}}.
\]
Since \(g_i\eta_\mu\ne1\),
Lemma~\ref{lem:regular-trace-away-from-one} shows that the traces over
\(\F_{\{i\},\{\mu\}}\) are regular as
\(\varepsilon_\beta\to0\) for \(\beta\in\Ibar\). Hence these traces do not
cancel the factor \(\sqrt{\varepsilon_c}\) or
\(\sqrt{\varepsilon_d}\), and only the matrix elements with \(c,d\in I\)
can contribute to the coefficient of the highest-order pole.

The additional flavour labels therefore modify only the
\(I\times I\) block. The resulting \(L\)-operator is
\begin{equation}
  \mathscr L_{I,\Ibar}^{\mathrm{add}}(v)
  \in\End(\F_{I,\Ibar}\otimes\F_{I,S_{\mathrm{add}}})
     \otimes\End(\C^M),
  \label{eq:L-with-additional-flavours-space}
\end{equation}
with matrix elements
\begin{align}
  [\mathscr L_{I,\Ibar}^{\mathrm{add}}(v)]_{ac}
  &=[\mathscr L_{I,\Ibar}(v)]_{ac}
    +\mathcal J_{ca}^{\mathrm{add}},
  &&a,c\in I,
  \label{eq:L-with-additional-flavours}
  \\
  [\mathscr L_{I,\Ibar}^{\mathrm{add}}(v)]_{cd}
  &=[\mathscr L_{I,\Ibar}(v)]_{cd},
  &&c\in\Ibar\ \text{or}\ d\in\Ibar.
  \label{eq:L-with-additional-flavours-rest}
\end{align}
The matrix elements of $\mathscr L_{I,\Ibar}(v)$ on the right-hand sides are given by
\eqref{eq:contracted-L-ac}--\eqref{eq:contracted-L-components}.

Define
\begin{equation}
  D_{I,S_{\mathrm{add}}}(\boldsymbol\eta)
  =\prod_{a\in I}\prod_{\mu\in S_{\mathrm{add}}}
  (g_a\eta_\mu)^{\mathsf N_{a\mu}}
  \label{eq:D-additional}
\end{equation}
and
\begin{align}
  \mathcal X_I^{(S_{\mathrm{add}})}(u;\boldsymbol\eta)
  :=\Tr_{\F_{I,\Ibar}\otimes\F_{I,S_{\mathrm{add}}}}
  \Bigl[
  &\mathscr L_{I,\Ibar,L}^{\mathrm{add}}(u-\theta_L)\cdots
   \mathscr L_{I,\Ibar,1}^{\mathrm{add}}(u-\theta_1)
  \notag\\[-1mm]
  &\times D_{I,\Ibar}D_{I,S_{\mathrm{add}}}(\boldsymbol\eta)
  \Bigr].
  \label{eq:X-generating-eta}
\end{align}

Applying Theorem~\ref{thm:Howe-decomposition-paper} with
\(V=V_I\) and \(W=U_{S_{\mathrm{add}}}\) gives
\begin{equation}
  \F_{I,S_{\mathrm{add}}}
  \simeq
  \bigoplus_{\Lambda:\,\ell(\Lambda)\le
  \min(|I|,|S_{\mathrm{add}}|)}
  V_\Lambda^{GL(V_I)}\otimes
  V_\Lambda^{GL(U_{S_{\mathrm{add}}})}.
  \label{eq:Howe-additional-appendix}
\end{equation}
The notation for partitions and for the group and Lie-algebra actions
\(\pi_\Lambda^X\) is fixed in Section~\ref{subsec:Schur-functions}.  On the
summand indexed by \(\Lambda\), \(\mathcal J_{ac}^{\mathrm{add}}\) acts as
\(\pi_\Lambda^{V_I}(E_{ac})\) on \(V_\Lambda^{GL(V_I)}\).

Let \(g_I\in GL(V_I)\) and
\(h_{\mathrm{add}}\in GL(U_{S_{\mathrm{add}}})\) be the diagonal operators
specified by
\begin{equation}
  g_Ie_a=g_ae_a,
  \qquad
  h_{\mathrm{add}}f_\mu=\eta_\mu f_\mu.
  \label{eq:gI-hadd-def}
\end{equation}
By \eqref{eq:twist-restriction-Howe-summand}, the restriction of
\eqref{eq:D-additional} to the summand indexed by \(\Lambda\) is
\begin{equation}
  \pi_\Lambda^{V_I}(g_I)\otimes
  \pi_\Lambda^{U_{S_{\mathrm{add}}}}(h_{\mathrm{add}}).
  \label{eq:D-additional-on-Howe-summand}
\end{equation}
Equation \eqref{eq:GLr-character-is-Schur} gives
\begin{equation}
  \Tr_{V_\Lambda^{GL(U_{S_{\mathrm{add}}})}}
  \pi_\Lambda^{U_{S_{\mathrm{add}}}}(h_{\mathrm{add}})
  =s_\Lambda(\boldsymbol\eta).
  \label{eq:Schur-character-additional}
\end{equation}

In our conventions, the operator \(X_I(u;\Lambda)\) defined below
corresponds to the transfer matrix denoted \(T_I^{\{\Lambda\}}(u)\) in
Ref.~\cite{KLT2012}. 
For the trace considered here, define
\begin{align}
  X_I(u;\Lambda)
  :=\Tr_{\F_{I,\Ibar}\otimes V_\Lambda^{GL(V_I)}}
  \Bigl[
  &\mathscr L_{I,\Ibar,L}^{\Lambda}(u-\theta_L)\cdots
   \mathscr L_{I,\Ibar,1}^{\Lambda}(u-\theta_1)
  \notag\\[-1mm]
  &\times D_{I,\Ibar}\pi_\Lambda^{V_I}(g_I)
  \Bigr],
  \label{eq:X-Lambda-Sch}
\end{align}
where
\begin{align}
  [\mathscr L_{I,\Ibar}^{\Lambda}(v)]_{ac}
  &=[\mathscr L_{I,\Ibar}(v)]_{ac}
    +\pi_\Lambda^{V_I}(E_{ca}),
  &&a,c\in I,
  \label{eq:L-Lambda-components}
  \\
  [\mathscr L_{I,\Ibar}^{\Lambda}(v)]_{cd}
  &=[\mathscr L_{I,\Ibar}(v)]_{cd},
  &&c\in\Ibar\ \text{or}\ d\in\Ibar.
  \label{eq:L-Lambda-components-rest}
\end{align}
Equations \eqref{eq:Howe-additional-appendix}--
\eqref{eq:Schur-character-additional} imply
\begin{equation}
  \mathcal X_I^{(S_{\mathrm{add}})}(u;\boldsymbol\eta)
  =\sum_\Lambda s_\Lambda(\boldsymbol\eta)X_I(u;\Lambda).
  \label{eq:X-Schur-function-expansion}
\end{equation}
By the comparison in \ref{app:BFLMS-local-comparison},
\eqref{eq:X-Lambda-Sch} agrees, up to the spectral-parameter shift and the
scalar normalization specified there, with the operator denoted
\(X_I(z,\Lambda)\) in Ref.~\cite{BFLMS2011}.

Equation \eqref{eq:X-generating-eta} traces over
\(\F_{I,S_{\mathrm{add}}}\).  If one starts instead from the full Fock
representation with colour labels in \(B\), then
\begin{equation}
  \F_{B,S_{\mathrm{add}}}
  =\F_{I,S_{\mathrm{add}}}\otimes
   \F_{\Ibar,S_{\mathrm{add}}}.
  \label{eq:additional-Fock-factorization}
\end{equation}
The oscillator generators in \(\F_{\Ibar,S_{\mathrm{add}}}\) do not appear in
\eqref{eq:L-with-additional-flavours}--
\eqref{eq:L-with-additional-flavours-rest}.  Hence the trace over
\(\F_{\Ibar,S_{\mathrm{add}}}\) is
\begin{align}
  \Tr_{\F_{\Ibar,S_{\mathrm{add}}}}
  \prod_{\beta\in\Ibar}\prod_{\mu\in S_{\mathrm{add}}}
  (g_\beta\eta_\mu)^{\mathsf N_{\beta\mu}}
  &=\prod_{\beta\in\Ibar}\prod_{\mu\in S_{\mathrm{add}}}
  \frac{1}{1-g_\beta\eta_\mu}.
  \label{eq:additional-Ibar-scalar}
\end{align}
Thus, if the trace defining
\(\mathcal X_I^{(S_{\mathrm{add}})}(u;\boldsymbol\eta)\) is extended
from \(\F_{I,S_{\mathrm{add}}}\) to the full representation
\(\F_{B,S_{\mathrm{add}}}\), the trace over
\(\F_{\Ibar,S_{\mathrm{add}}}\) gives the scalar factor in
Eq.~\eqref{eq:additional-Ibar-scalar}.

\subsection*{The case \(\Lambda=0\)}
Let \(V_0=\C\) be the one-dimensional representation of
\(\gl(V_I)\) with highest weight zero.  Then
\(\pi_0^{V_I}(E_{ac})=0\), \(\pi_0^{V_I}(g_I)=1\), and
\begin{equation}
  X_I(u;0)=Q_I(u).
  \label{eq:X-zero-equals-Q}
\end{equation}
In Ref.~\cite{BFLMS2011}, the operator \(X_I^+(z,0)\) is defined instead
with the infinite-dimensional Verma module \(M_0\) of highest weight zero:
\begin{equation}
  X_I(z,0):\quad \F_{I,\Ibar}\otimes V_0,
  \qquad
  X_I^+(z,0):\quad \F_{I,\Ibar}\otimes M_0.
  \label{eq:finite-versus-Verma-spaces}
\end{equation}
The BGG resolution of \(V_0\) contains \(M_0\); hence replacing the trace over \(V_0\) by the trace over \(M_0\) is not merely a multiplication by an overall scalar factor.

We do not state a residue formula for
\eqref{eq:X-generating-eta}.  Extending
Proposition~\ref{prop:main-contraction} to this trace would require a
separate calculation with the additional oscillators.

\end{document}